\documentclass[10pt,journal]{IEEEtran}
\input{psfig.sty}
\input{epsf.sty}
\usepackage{hyperref} 
\usepackage{fancybox}   
\usepackage{psfig}      
\usepackage{graphicx}
\usepackage{amsmath}
\usepackage{color}
\usepackage{epsfig}
\usepackage{amssymb}
\usepackage{verbatim}
\usepackage{enumitem}
\usepackage{bbm}

\IEEEoverridecommandlockouts

\newcommand{\remove}[1]{}
\newcommand{\argmin}{\operatornamewithlimits{argmin}}

\newtheorem{theorem}{Theorem}
\newtheorem{corollary}{Corollary}
\newtheorem{lemma}{Lemma}

\newtheorem{definition}{Definition}

\newtheorem{remark}{Remark}

\newtheorem{claim}{Claim}

\newtheorem{algorithm}{Algorithm}
\newtheorem{assumption}{Assumption}

\newcommand{\qed}{\hfill \ensuremath{\Box}}

\usepackage{setspace}
\usepackage{subfigure}

\begin{document}

\title{Sequential Decision Algorithms for Measurement-Based Impromptu Deployment of a Wireless Relay Network along a Line
\thanks{The contents of this paper have been arXived in \cite{chattopadhyay-etal15measurement-based-impromptu-deployment-arxiv-v1}.}
\thanks{Arpan Chattopadhyay and Anurag Kumar are with the Department of ECE, 
Indian Institute of Science, Bangalore, India; email: arpanc.ju@gmail.com, anurag@ece.iisc.ernet.in. 
Marceau Coupechoux is with Telecom ParisTech and CNRS LTCI, 
Dept. Informatique et Reseaux, 23, avenue d'Italie, 
75013 Paris, France; email: marceau.coupechoux@telecom-paristech.fr.}
 \thanks{The research 
reported in this paper was supported by a Department of Electronics and Information 
Technology (DeitY, India) and NSF (USA) funded project on Wireless 
Sensor Networks for Protecting Wildlife and Humans, by an Indo-French Centre for 
Promotion of Advance Research (IFCPAR) funded project, and by the Department of Science and Technology (DST, India), 
 via  J.C. Bose Fellowship.} \thanks{{\bf All appendices are provided in the supplementary material.}}}


\newcounter{one}
\setcounter{one}{1}
\newcounter{two}
\setcounter{two}{2}



\author{
Arpan~Chattopadhyay, Marceau~Coupechoux, and Anurag~Kumar\\
}

\maketitle
\thispagestyle{empty}

\begin{abstract}
We are motivated by the need, in some applications, for impromptu or as-you-go deployment of wireless sensor networks. 
A person walks  along a line, starting from a sink node (e.g., a base-station), and 
proceeds towards a source node (e.g., a sensor) which is at an a priori unknown location. 
At equally spaced locations, he makes link quality measurements to the previous relay, 
and deploys relays at some of these locations, with the aim to connect the source to the sink by a multihop wireless path. 
In this paper, we consider two approaches for impromptu deployment: (i) the deployment agent can only 
move forward (which we call {\em a pure as-you-go} approach), and (ii) the deployment agent can make 
measurements over several consecutive steps before selecting a placement location among them 
(the {\em explore-forward} approach). We consider a very light traffic regime, and
formulate the problem as 
a Markov decision process, where the trade-off is among the power used by the nodes, the outage 
probabilities in the links, and the number of relays placed per unit distance. 
We obtain the structures of the 
optimal policies for  the {\em pure as-you-go} approach as well as for the {\em explore-forward} approach. 
We also consider natural heuristic algorithms, for comparison. 
Numerical examples show that the explore-forward approach significantly outperforms the 
pure as-you-go approach in terms of network cost. 
Next, we propose two learning algorithms for the explore-forward approach, based on Stochastic Approximation, 
which asymptotically converge to the set of optimal policies, without
using any knowledge of the radio propagation model.
We demonstrate numerically that 
the learning algorithms can converge (as deployment progresses)  to the set of optimal policies reasonably fast 
and, hence, can be practical model-free algorithms for deployment over large regions.  Finally, we demonstrate the 
end-to-end traffic carrying capability of such networks via field deployment.
\end{abstract}

\vspace{-6mm}
\section{Introduction}\label{Introduction}
\vspace{-2mm}

A wireless sensor network (WSN) typically comprises sensor nodes (sources of measurements), 
a base station (or sink), and wireless relays for multihop communication between the sources and the sink. 
There are situations in which a WSN  needs to be deployed 
(i.e., the relays and the sensors need to be placed) in an {\em impromptu or as-you-go} fashion. 
One such situation is in 
emergencies, e.g., situational awareness networks deployed by first-responders such as fire-fighters 
or anti-terrorist squads. As-you-go deployment is also of interest when deploying multihop wireless networks for sensor-sink 
interconnection 
over large terrains, such as forest trails (see \cite{dyo-etal10wildlife-wsn}  for 
an application of multi-hop WSNs in wildlife 
monitoring, and \cite[Section~$5$]{alkhatib14review-forest-fire} for application of WSN in forest fire detection), 
where it may be difficult to make exhaustive 
measurements at all possible deployment locations before placing the relay nodes. 
As-you-go deployment would be particularly useful when the  network is temporary and needs to be quickly 
redeployed at a different place (e.g., to monitor a moving phenomenon such as groups of wildlife).\footnote{In remote 
places, cellular network coverage may not be available or  practicable. Hence, a multi-hop WSN is required 
for monitoring purposes.}

Our work is motivated by the need for as-you-go deployment of a WSN over large terrains, 
such as forest trails, where planned deployment (requiring exhaustive measurements over the deployment region) 
would be time consuming and difficult. 
Abstracting the above-mentioned problems, we  consider the problem of  deployment of relay nodes
along a line, between a sink node (e.g., the WSN base-station) and a source node (e.g., a sensor) (see
Figure~\ref{fig:why-impromptu}), where a {\em single deployment agent} (the person who is carrying out the deployment)
starts from the sink node, places relay nodes along the line, and
places the source node where required.  In applications, the location at which sensor placement 
is required might only be discovered as the deployment agent walks (e.g., in an animal monitoring application, 
by finding a concentration of pugmarks, or a watering hole).

\begin{figure}[!t]
\begin{center}
\includegraphics[height=1.7cm, width=9cm]{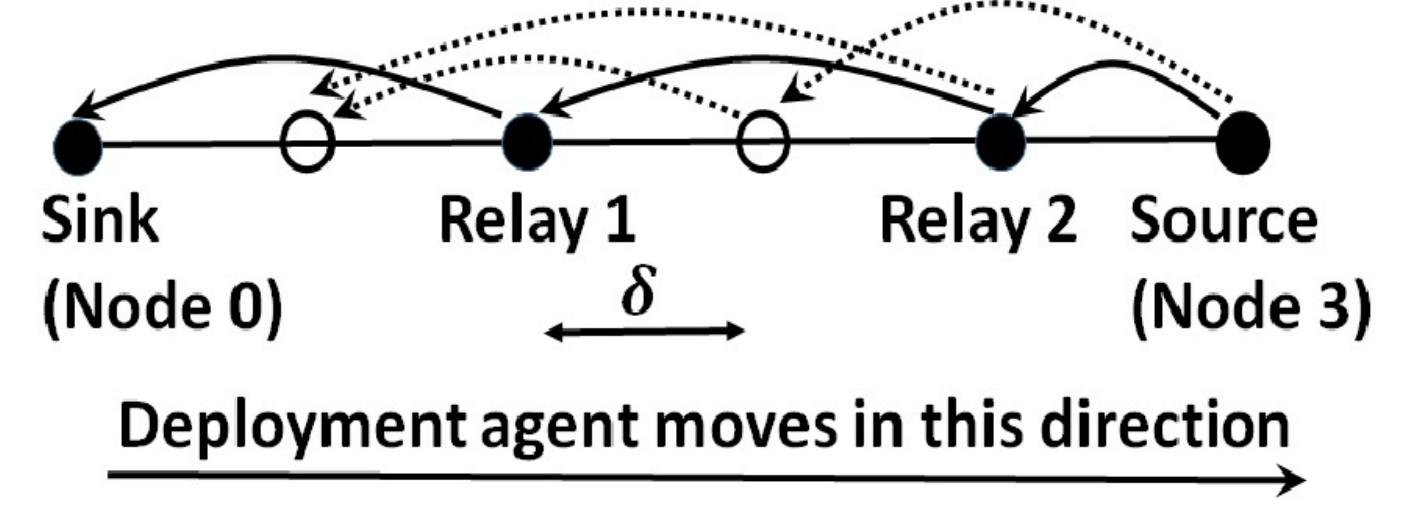}
\end{center}
\vspace{-6mm}
\caption{A wireless relay network, placed along a line, connecting a source to a sink. The dots (filled and unfilled) 
denote potential locations for node placement, and are successively $\delta$ meters apart. 
The deployed network comprises two relays (filled dots) placed at two of the potential locations; the solid 
arrows show the path from the source to the sink. {\em The dotted arrows show some more 
possible links between pairs of potential  locations.}}
\label{fig:why-impromptu}
\vspace{-6mm}
\end{figure}

In the perspective of an optimal \emph{planned} 
deployment, we would need to place relay nodes at \emph{all} potential locations 
(for example, with reference to Figure~\ref{fig:why-impromptu}, this would mean placing relays at all the four dots 
in between the source and the sink) and  measure the 
qualities of all possible links  
in order to decide where to place the relays. 
This approach would provide the global optimal solution, but the time and effort required 
might not be acceptable in the applications mentioned earlier. 
With {\em impromptu} 
deployment, the next relay
placement locations depend on the radio link qualities to the
previously placed nodes; these link qualities and also the source location are discovered as the agent walks along the line. 
Such an approach 
requires fewer measurements compared to planned deployment, but, in general, is suboptimal.

In this paper, we mathematically formulate the  problems of
impromptu deployment  along a line as {\em optimal sequential decision problems}. 
The cost 
of a deployment is evaluated as a linear combination of three components: the sum transmit 
power along the path, the sum outage 
probability along the path, and the number of relays deployed; we provide a motivation for this cost structure.  
We formulate relay placement problems to  
minimize the expected average cost per-step. 
Our channel model accounts for path-loss, 
shadowing, and fading.  

We explore deployment with two approaches: (i) the {\em pure as-you-go }  approach
and (ii) the {\em explore-forward} approach. 
In the pure as-you-go approach, the deployment agent can only move forward; this approach is a necessity 
if the deployment needs to be quick. Due to shadowing, 
the path-loss over a link of a given length is random, and a more efficient deployment can be expected if 
link quality measurements at several  locations along the line are compared and an optimal choice is made among 
these; we call this approach 
{\em  explore-forward}.  
Explore-forward would require the deployment agent to retrace his steps; but this 
might provide a good compromise between deployment 
speed and deployment efficiency.  

We formulate each of these problems as a Markov decision process (MDP), obtain the optimal policy structures, 
illustrate their performance numerically 
and compare with reasonable heuristics. 
Next, we propose several learning algorithms and 
prove that each of them asymptotically converges to the optimal policy if we seek to minimize the long run average cost per unit 
distance. We also demonstrate 
the convergence rate of the learning algorithms via numerical exploration.  Finally, we demonstrate
the end-to-end traffic carrying capability of such networks via
field deployment.

\vspace{-4mm}
\subsection{Related Work}\label{subsec:related_work}
\vspace{-1mm}
Until recently, problems of impromptu deployment of wireless networks 
have been addressed primarily by heuristics and by experimentation. 
Howard et al., in \cite{howard-etal02incremental-self-deployment-algorithm-mobile-sensor-network}, 
provide heuristic algorithms for incremental deployment of
sensors in order to cover the deployment area; their problem is related to that
of self-deployment of autonomous robot teams.  
Souryal et al., in \cite{souryal-etal07real-time-deployment-range-extension}, address the
problem of impromptu wireless network deployment by experimental study of indoor RF link quality variation; 
a similar approach is taken in \cite{aurisch-tlle09relay-placement-emergency-response} also. 
The authors of \cite{liu-etal10breadcrumb} describe a {\em breadcrumbs} system for aiding firefighting inside buildings. 
Their work addresses the same class of problems as ours, with the requirement that the deployment agent has to 
stay connected to $k$ 
previously placed nodes in the deployment process. Their work considers the trade-off between link qualities 
and the deployment rate, but does not provide any optimality guarantee of their deployment schemes. 
Their next 
work \cite{liu-etal11multiuser-breadcrumb} provides a reliable multiuser {\em breadcrumbs} system.  
Bao and Lee, in \cite{bao-lee07rapid-deployment-wireless}, study the scenario where a group of first-responders, 
starting from a command centre, enter a large area where there is no communication infrastructure, and as they walk they 
place relays at suitable locations in order to stay connected among themselves and with the command centre. 
However, these approaches are based on heuristic algorithms, rather than 
on rigorous formulations; hence they do not provide any provable performance guarantee.

In our work we have formulated impromptu deployment as a sequential decision problem, and  have derived optimal deployment policies.  
Recently, Sinha et al. (\cite{sinha-etal12optimal-sequential-relay-placement-random-lattice-path}) 
have provided an algorithm based on an MDP formulation in order to establish a multi-hop 
network between a sink and an unknown source location, 
by placing relay nodes along a random lattice path. Their model uses a deterministic mapping between power and wireless 
link length, and, hence, does not consider statistical variability (due to shadowing) of the transmit power 
required to maintain the link quality over links having the same length. 
The statistical variation of link qualities over space requires measurement-based 
deployment, in which the deployment agent makes placement decisions at a point based on the measurement of 
the power required to establish a link (with a given quality) 
to the previously placed node.

 We view the current paper as a continuation of 
our   papers \cite{chattopadhyay-etal13measurement-based-impromptu-placement_wiopt} 
(which provides the first theoretical formulation of measurement-based impromptu deployment) and  
\cite{chattopadhyay-etal14deployment-experience} 
(which provides field deployment results using our algorithms).

\vspace{-3mm}
\subsection{Organization}\label{subsec:organization}
\vspace{-2mm}
The system model and notation have been described in 
Section~\ref{sec:system_model_and_notation}.  
Impromptu deployment with a pure as-you-go approach  
has been discussed in Section~\ref{section:no_backtracking_average_cost}. 
Section~\ref{sec:backtracking_average_cost} presents our work on the  explore-forward approach. 
A numerical comparison between these two approaches are made in 
Section~\ref{section:numerical_results_comparison_average_cost_backtracking_no_backtracking}. 
Section~\ref{section:learning_backtracking_given_xio_xir} 
and Section~\ref{section:learning_backtracking_adaptive_with_outage_cost} 
describe the learning algorithms for the explore-forward approach approach. 
Numerical results are provided in Section~\ref{section:convergence_speed_learning_algorithms} 
on the rate of convergence of the learning algorithms. Experimental results demonstrating the traffic carrying 
capability of the deployed networks are provided in Section~\ref{section:real_deployment}, followed by the conclusion.

\vspace{-3mm}
\section{System Model and Notation}\label{sec:system_model_and_notation}
\vspace{-1mm}

The line is discretized into steps of length $\delta$ 
(Figure~\ref{fig:why-impromptu}), starting from the sink.  
Each point, located at a distance of an integer multiple of $\delta$ from the sink, 
is considered to be a potential location where a 
relay can be placed. 
As the {\em single} deployment agent walks along the line, at each step 
or at some subset of steps, he measures the link quality from the 
current location to the previous node; 
these measurements are used to decide the location and transmit power of the next relay.

As shown in Figure~\ref{fig:why-impromptu}, the sink is called Node $0$, the relay 
closest to the sink is called Node $1$, and the relays are enumerated 
as nodes $\{1,2,3,\cdots\}$ as we walk away from the sink. 
The link whose transmitter is Node $i$ and receiver is Node 
$j$ is called link $(i,j)$. A generic link is denoted by $e$. 
The length of each link is an integer multiple of $\delta$.

\begin{figure}[!t]
\begin{center}
\includegraphics[height=1.8cm, width=9cm]{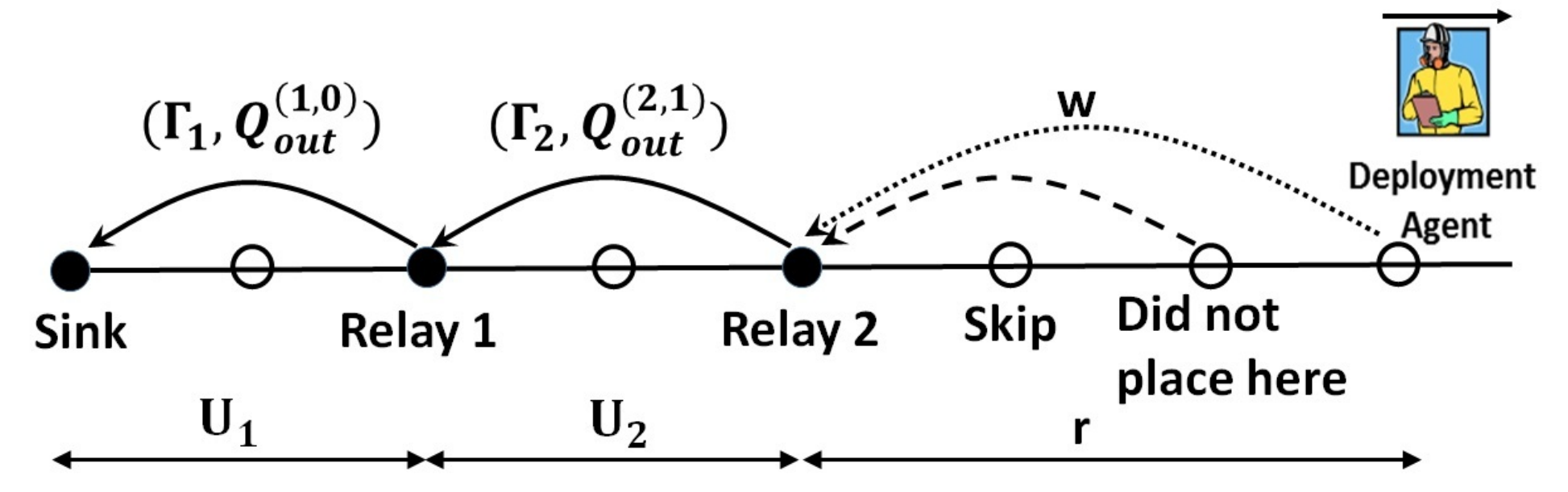}
\end{center}
\vspace{-4mm}
\caption{ Illustration of pure as-you-go deployment with $A=1$ and $B=3$. 
In this ``snap-shot'' of the deployment process, the deployment agent has already placed Relay~$1$ and Relay~$2$ 
at distances $U_1$ and $U_2$, has set their transmit powers to $\Gamma_1$ and $\Gamma_2$, 
thereby achieving outage probabilities $Q_{out}^{(1,0)}$ and $Q_{out}^{(2,1)}$ 
(links shown by solid arrows). Having placed Relay~$2$, 
he skips the next location (since $A=1$); based on measurements made at the next location (dashed arrow), 
the algorithm advises him to not place a relay and move on. The diagram shows the agent in the process of 
evaluating the next location  at $r=3 \delta$ distance from Relay~$2$ (dotted arrow). 
Based on these measurements, the deployment agent will decide whether to place a relay at $r=3 \delta$; if a relay is not 
placed here, it must be placed at the next location, since $B=3$.}
\label{fig:pure-as-you-go}
\vspace{-7mm}
\end{figure}

\vspace{-3mm}
\subsection{Channel Model and Outage Probability}\label{subsection:channel_model}
\vspace{-1mm}
We consider the usual 
aspects of path-loss, shadowing, and fading to model the wireless channel. 
The received power of a packet (say the $k$-th packet, $k \geq 1$) in a particular link (i.e., a transmitter-receiver pair)  
of length $r$ is given by:
\begin{equation}
 P_{rcv,k}=P_T c (\frac{r}{r_0×})^{-\eta}H_kW \label{eqn:channel_model}
\end{equation}
where $P_T$ is the transmit power, $c$ is the path-loss at the reference distance $r_0$, 
$\eta$ is the path-loss exponent, $H_k$ denotes the fading random variable seen by the $k$-th packet (e.g., it 
is an exponentially distributed  
random variable for the Rayleigh fading model), and $W$ 
denotes the shadowing random variable. 
$H_k$ captures the variation of the received power over time, and it takes  
independent values over different coherence times. 

The path-loss between a transmitter and a receiver at a given distance can have a large spatial  
variability around the mean path-loss (averaged over fading), as the transmitter is moved over different points at the 
same distance from the receiver; this is called shadowing. 
Shadowing is usually modeled as a log-normally distributed, 
random, multiplicative path-loss factor; in dB, shadowing is distributed with values of standard 
deviation as large as $8$ to $10$~dB. Also, shadowing is spatially uncorrelated over distances that depend on 
the sizes of the objects in the propagation environment (see \cite{agarwal-patwari07correlated-shadow-fading-multihop}); 
{\em our measurements in a 
forest-like region of our Indian Institute of Science (IISc) campus established log-normality of the shadowing 
and gave a shadowing decorrelation distance of $6$~meters 
(see \cite{chattopadhyay-etal14deployment-experience})}. In this paper, 
{\em we assume that the shadowing at any two different links in the network 
are independent, i.e., $W_{(e_1)}$ is independent of $W_{(e_2)}$ for $e_1 \neq e_2$.} This is a 
reasonable assumption 
if $\delta$ is chosen to be at least the decorrelation distance 
(see \cite{agarwal-patwari07correlated-shadow-fading-multihop}) of the shadowing. Thus, from our experiments 
in the forest-like region in the IISc campus, we can safely assume independent shadowing at 
different potential locations if $\delta$ is greater than $6$~m. In this paper, 
$W$ is assumed to take values from a set $\mathcal{W}$. We will denote by $p_W(w)$ 
the probability mass function or probability density function of $W$, depending on 
whether $\mathcal{W}$ is a countable set or an uncountable set (e.g.,  log-normal shadowing).

A link is considered to be in \emph{outage} if  the received 
signal power (RSSI) drops (due to fading) below $P_{rcv-min}$  (e.g., below $-88$~dBm, a 
figure that we have obtained via experimentation for the popular TelosB ``motes,'' 
see \cite{bhattacharya-etal13smartconnect-comsnets}). 
Since practical radios can only be set to transmit at a finite set of power levels, 
the transmit power of each node can be chosen from a discrete set, $\mathcal{S}:=\{P_1, P_2, \cdots, P_M \}$, where 
$P_1 \leq P_2 \leq \cdots \leq P_M $.  
For a link of length $r$, a transmit power $\gamma$ and any particular realization of shadowing $W=w$, 
the outage probability is denoted by $Q_{out}(r,\gamma,w)$, which is increasing in $r$ and decreasing 
in $\gamma$, $w$ (according to (\ref{eqn:channel_model})). 

$Q_{out}(r,\gamma,w)$ depends on the 
fading statistics. For a link with shadowing realization $w$, if the transmit power 
is $\gamma$, the received power of a packet will be $P_{rcv}=\gamma c (\frac{r}{r_0})^{-\eta}wH$. 
Outage is 
the event $P_{rcv} \leq P_{rcv-min}$. If $H$ is exponentially distributed with mean $1$ 
(i.e., for Rayleigh fading), then we have, 
$Q_{out}(r,\gamma,w)=\mathbb{P}( \gamma c (\frac{r}{r_0})^{-\eta}wH \leq P_{rcv-min} )=
1-e^{-\frac{P_{rcv-min}(\frac{r}{r_0})^{\eta}}{\gamma c w}}$. 
The outage probability of a randomly chosen link of given length and given transmit power is a 
random variable, where the randomness 
comes from shadowing $W$.  {\em Outage probability can be measured by sending a 
sufficiently large number of packets over a link and calculating the percentage of packets whose RSSI is below 
$P_{rcv-min}$. }

\vspace{-5mm}
\subsection{Deployment Process and Related Notation}\label{subsection:deployment_process_notation}
\vspace{-1mm}

\begin{figure}[!t]
\centering
\includegraphics[height=1.8cm, width=9cm]{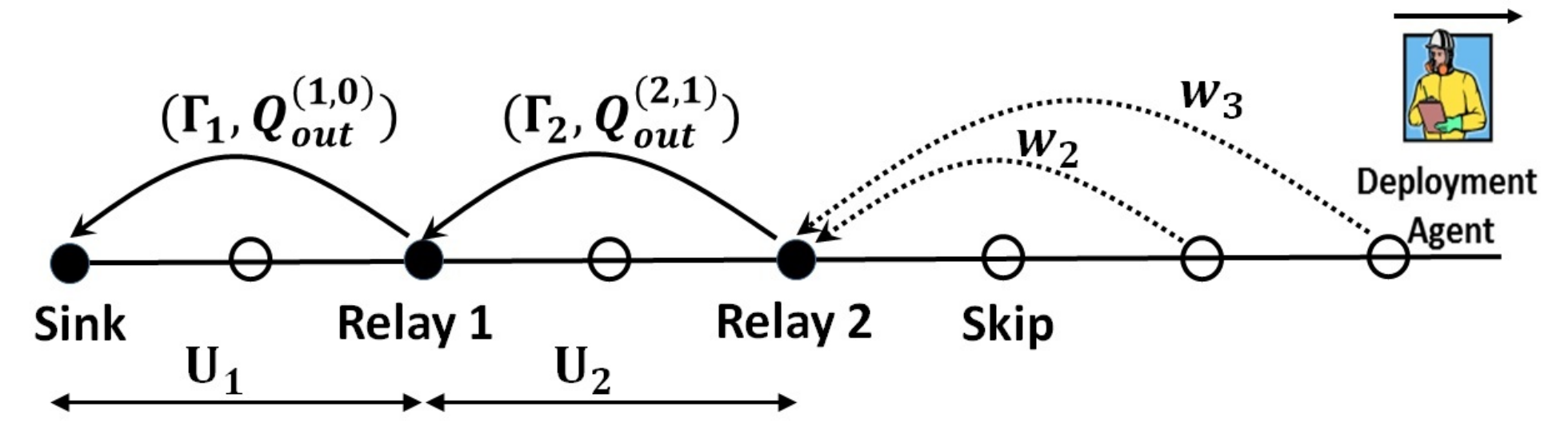}
\vspace{-7mm}
\caption{Illustration of explore-forward deployment with $A=1$ and $B=2$. 
In this ``snap-shot'' of the deployment process, the deployment agent has already placed 
Relay~$1$ and Relay~$2$ at distances $U_1$ and $U_2$, has set their transmit powers to $\Gamma_1$ and $\Gamma_2$, 
thereby achieving outage  $Q_{out}^{(1,0)}$ and $Q_{out}^{(2,1)}$ 
(links shown by  solid arrows). 
Having placed Relay~$2$, 
he skips the next location (since $A=1$). The agent then 
evaluates the next two locations (dotted arrows) ($B=2$). 
Then, based on the measurements at these two locations, 
the algorithm determines which of them to place the relay at and which power level to use.}
\label{fig:backtracking_illustration}
\vspace{-7mm}
\end{figure}

In this paper, we consider two approaches for deployment.

{\em Pure as-you-go deployment:} After placing a relay, the agent skips the next $A$~steps, and 
sequentially measures the outage probabilities from locations $(A+1),(A+2),\cdots,(A+B)$ to the previously placed node, 
at all transmit power levels   $\gamma \in \mathcal{S}$. 
As the agent explores the locations 
$(A+1),\cdots,(A+B-1)$ and makes link quality measurements.\footnote{At a 
distance $r$ from the previous node, he measures the outage probabilities $\{Q_{out}(r,\gamma,w)\}_{\gamma \in \mathcal{S}}$ 
from the current location to the previous node, where $w$ is the realization of the shadowing in the link being evaluated.}  
At each step he decides whether to 
place a relay there, and if the decision is to place a relay, then he 
also decides the transmit power for the placed relay. This   has been depicted in 
Figure~\ref{fig:pure-as-you-go}. 
In this process, if he has walked $(A+B)$ steps away from the previous 
relay, or if he encounters the source location, then he must place a node. 
$A$ and $B$ will be fixed before deployment begins. \qed

{\em Explore-forward deployment:}  
After placing a node, 
the deployment agent skips the next $A$ locations ($A \geq 0$) 
and measures the outage probabilities  
to the previous node from locations $(A+1), \cdots,(A+B)$, at each power level 
from the set $\mathcal{S}$. 
Then, based on these $B |\mathcal{S}|$ measurements\footnote{Let us denote by $w_u$ the realization of shadowing in the potential link 
between the $u$-th location (starting from the previously placed node) and the previous node 
(see Figure~\ref{fig:backtracking_illustration}). The agent measures the outage 
probabilities $\{Q_{out}(u,\gamma,w_u)\}_{A+1 \leq u \leq A+B, \gamma \in \mathcal{S}}$ in order to make a placement decision.}  of the outage probability values, 
he places the relay at location $u^* \in \{A+1, \cdots,A+B\}$, sets its transmit power to 
$\gamma^* \in \mathcal{S}$, and repeats the same process 
for placing the next relay. This procedure 
is illustrated in Figure~\ref{fig:backtracking_illustration}. 
If the source location is encountered within $(A+B)$ steps from the previous node, 
then the source is placed. \qed

{\em Choice of $A$ and $B$:} If the propagation environment is very good, 
or if we need to place a limited number of relays over a long 
line, it is very unlikely that a relay will be placed within the first few locations from the previous node. In such cases, 
we can skip measurements at locations $1,2,\cdots,A$ and make measurements from locations 
$(A+1),\cdots,(A+B)$. 
In general, the choice of $A$ and $B$ will depend on the constraints 
and requirements for the deployment. Larger   $A$ will 
result in faster exploration, but very large $A$ will result in very high outage 
in each link. For a fixed $A$, a large  $B$ results in more measurements, but we can expect a better performance.

\vspace{-4mm}
\subsection{Traffic Model}\label{subsection:traffic_model}
\vspace{-1mm}

In order to develop the problem formulation, we assume that  the traffic is so
low that there is only one packet in the network at a time; we call
this the ``lone packet model.''  Hence, there are no simultaneous
transmissions to cause interference. This permits us to easily write down the
communication cost on a path over the deployed relays. However, this assumption does not trivialize the deployment problem, 
since the deployment must still take into account the stochastic shadowing and fading in the links, 
and the effects of these factors on the number of nodes deployed and the powers they use.

The lone packet traffic model is realistic for sensor networks that carry low duty cycle 
measurements, or just carry an occasional alarm packet. For example, recently 
there has been an effort to design passive infra-red (PIR) sensor platforms that 
can detect intrusion of a human or animal, and also can classify whether the intruder is a human or an animal 
(\cite{raviteja-etal15animation-intrusion-classification-PIR}). The data rate generated by such a platform 
deployed in a forest  will be very low. The authors in 
\cite[Section~3.2]{dyo-etal10wildlife-wsn} 
use only a $1.1\%$ duty cycle for a multi-hop wireless sensor network used for the purpose 
of wildlife monitoring. The sensors gather data from RFID collars on the animals; hence, the traffic to be 
supported by the network is light.  Lone packet model is also realistic for condition monitoring/industrial telemetry applications 
(\cite{aghaei11wsn-water-gas}) as well, where the time between  successive measurements is very large.   
Infrequent data model is  
common in machine-to-machine communication (\cite{adame-etal14m2m}). 
Table~$1$ and Table~$3$ of 
\cite{mainwaring-etal02wsn-habitat} illustrate sensors whose  sampling rate and the size of the sampled data packets 
are small; it shows data rate requirement as small as several   bytes per second for habitat monitoring.

Even though the network 
is designed for the lone packet traffic, it will be able to carry some amount of 
positive traffic. See Section~\ref{section:real_deployment} 
for experimental evidence of this claim; a five-hop line network deployed 
using one algorithm proposed in this paper, over a $500$~m long line in a forest-like 
environment, was able to carry $127$~byte packets at a rate of $4$~packets per second, with end-to-end packet loss 
probability less than $1 \%$, which is  sufficient for the applications mentioned above.

Lone packet model is also valid when interference-free communication 
is achieved via multi-channel access. 
Recently there 
have been efforts to use multiple channels available in $802.15.4$ radio in a network; 
see \cite{lohier-etal12multichannel-wsn}, \cite{abdeddaim-etal12multichannel-cluster-tree-wsn}, 
\cite{toscano-bello12multichannel-superframe-scheduling-wsn}, 
\cite{bardella-etal10experimental-multichannel-transmission-wsn}. In a line topology, this reduces to frequency 
reuse after certain hops, which, in turn, mitigates interference in the network. 
Thus,  as with the lone  packet assumption,  the availability of multiple channels,  
and appropriate channel allocation over the network,  eliminates the need to optimize over link schedules.

It has been proved that design with the lone packet model can be the starting point for a design with desired positive traffic 
(see \cite{bhattacharya-kumar12qos-aware-survivable-network-design}). Network design for carrying a given positive 
traffic rate is left as a future research work.

\vspace{-2mm}
\subsection{Network Cost Structure}\label{subsection:network_cost}
\vspace{-1mm}

In this section we develop the cost that we use to evaluate the performance of a given deployment policy. 
Given the current location of the deployment
agent with respect to the previous relay, and given the measurements made 
to the previous relay, a policy will provide the placement decision 
(in the case of pure as-you-go deployment, whether or not to place the relay, and if place then at what power, 
and in the case of explore-forward deployment, where among the $B$ locations to place the relay and at which power).

Let us denote the number of placed relays up to $x$~steps (i.e., $x \delta$ meters) from the sink by $N_x$ 
($ \leq x$); define $N_0 = 0$. Since deployment decisions are based on measurements to already placed relays, and 
since the path-loss over a link is a random variable (due to shadowing), we see that $\{N_x\}_{x \geq 1}$ 
is a random process. In this paper we have assumed that each node forwards each packet to the 
immediately previously placed relay (e.g., with reference to Figure~\ref{fig:why-impromptu}, 
the source forwards all packets to Relay~$2$, 
which, in turn, forwards all packets to Relay~$1$, etc.). 
See \cite{chattopadhyay-etal13measurement-based-impromptu-placement_wiopt} 
for the considerably more  complex possibility of relay skipping while forwarding packets.

When the node $i$ is placed, the deployment policy also prescribes the transmit power that this node should use, 
say, $\Gamma_i$; then the outage probability over the link $(i, i-1)$, so created, is denoted by $Q_{out}^{(i, i-1)}$ 
(see Figure~\ref{fig:pure-as-you-go} and Figure~\ref{fig:backtracking_illustration}). 
We evaluate the cost of the deployed network, up to $x  \delta$ steps, as a linear combination of three cost measures: 

\begin{enumerate}[label=(\roman{*})]
\item The number of relays placed, i.e., $N_x$.
\item The sum outage, i.e., $\sum_{i=1}^{N_x} Q_{out}^{(i, i-1)}$. 
The motivation for this measure is that, for small values of $Q_{out}$, the sum-outage is
approximately the probability that a packet sent from the point $x$ to the source
encounters an outage along the path from the point $x$ back to the sink.

\item The sum power over the hops, i.e., $\sum_{i=1}^{N_x} \Gamma_i$. 

\end{enumerate}

These three costs are combined into one cost measure by combining them linearly and taking expectation (under a policy $\pi$), 
as follows:

\begin{equation}
\mathbb{E}_{\pi} ( \sum_{i=1}^{N_x} \Gamma_i + \xi_{out} \sum_{i=1}^{N_x}Q_{out}^{(i,i-1)}+ \xi_{relay} N_x )
\label{eqn:cost_function_sum_power_sum_outage}
\end{equation}

The multipliers $\xi_{out} \geq 0$ and $\xi_{relay} \geq 0$ can be viewed as capturing the emphasis we wish 
to place on the corresponding measure of cost. For example, a large value of $\xi_{out}$ will aim 
for a network deployment with smaller end-to-end expected outage. We can view $\xi_{relay}$ as the cost 
of placing a relay.

{\em A Motivation for the Sum Power Objective:} 
In case all the nodes have {\em wake-on radios,} the nodes normally stay in sleep mode, and each sleeping node 
draws a very small current from the battery 
(see \cite{vodel-hardt13energy-efficient-communication-distributed-embedded-systems}). 
When a node has a packet, it sends a wake-up tone to the intended receiver. 
The receiver wakes up and the sender transmits the 
packet. The receiver sends an ACK packet in reply. Clearly, the energy spent in transmission and reception of data packets 
governs the lifetime of a node, given that the ACK size is negligible.  
We assume that a fixed modulation scheme is used, so that the transmission bit rate over all 
links is the same (e.g., in IEEE~$802.15.4$ radios, that are commonly used for sensor networking, 
the standard modulation scheme provides a bit rate of $250$~Kbps). We also assume a fixed packet length. 
Let $t_p$ be the transmission duration of a packet over a link, 
and suppose that the node~$i$ ($1 \leq i \leq N_x$) uses power $\Gamma_i$ during transmission. 
Let $P_{r}$ denote the packet reception power expended in the electronics at any receiving node. 
If the packet generation rate $\zeta$ at the source is very small, 
the lifetime of the $k$-th node ($1 \leq k \leq N_x$) is $T_k:=\frac{E}{\zeta(\Gamma_k+P_{r})t_p }$ seconds 
($E$ is the total energy in a fresh battery). Hence, the rate at 
which we have to replace the batteries in the network from the sink up to distance $x$~steps is given by 
$\sum_{k=1}^{N_x} \frac{1}{T_k}=\sum_{k=1}^{N_x} \frac{\zeta(\Gamma_k+P_{r})t_p }{E}$. 
The term $\frac{\zeta P_{r}t_p }{E}$ can be absorbed 
into $\xi_{relay}$. 
Hence, the battery depletion rate is 
proportional to $\sum_{k=1}^{N_x} \Gamma_k$.

\vspace{-4mm}
\subsection{Deployment Objective}\label{subsection:deployment_objective}
\vspace{-1mm}

We assume  that  the distance $L$ to the source from the sink 
is a priori unknown, and its distribution is also unknown. 
Hence, we assume that $L=\infty$ (deployment along a line of infinite length) and 
develop deployment policies that seek to minimize the average cost per step. 
{\em This setting can  be useful in practice when $L$ is large (e.g., a long forest trail). 
Also, if we seek to create networks along multiple trails in a forest, and if deployment is done serially 
along multiple trails, then this is effectively equivalent to deployment along a single long line, provided 
that the trails have statistically identical radio propagation environment. 
Note that, in case we deploy serially along multiple lines but 
use this formulation, it means that we seek to optimize the per-step cost averaged over multiple lines.}

\subsubsection{Unconstrained Problem}\label{subsubsection:the_unconstrained_problem}
 
Motivated by the cost structure   and the $L=\infty$ 
model, we seek to solve the following:

\footnotesize
\begin{eqnarray}
 \inf_{\pi \in \Pi} \limsup_{x \rightarrow \infty} \frac{\mathbb{E}_{\pi}\sum_{i=1}^{N_x}(\Gamma_i+\xi_{out}Q_{out}^{(i,i-1)}+\xi_{relay})}{x}
\label{eqn:unconstrained_problem_average_cost_with_outage_cost}
\end{eqnarray}
\normalsize
where $\pi$ is a placement {\em policy}, and 
$\Pi$ is the set of all possible placement policies (to be formalized later). 
We formulate (\ref{eqn:unconstrained_problem_average_cost_with_outage_cost}) as a 
long-term average cost Markov decision process (MDP). 
\subsubsection{Connection to a Constrained Problem}\label{subsubsection:connection_between_constrained_unconstrained_problem}
Note that, (\ref{eqn:unconstrained_problem_average_cost_with_outage_cost}) is the relaxed version of the following 
constrained problem  where we seek to minimize 
the mean power per step subject to a constraint on the mean outage per step and a constraint on the 
mean number of relays per step:

\footnotesize
\begin{eqnarray}
&& \inf_{\pi \in \Pi} \limsup_{x \rightarrow \infty} \frac{\mathbb{E}_{\pi}\sum_{i=1}^{N_x}\Gamma_i}{x} \nonumber\\
&s.t.& \, \limsup_{x \rightarrow \infty} \frac{\mathbb{E}_{\pi}\sum_{i=1}^{N_x}Q_{out}^{(i,i-1)}}{x} \leq \overline{q}  
\text{ and } \limsup_{x \rightarrow \infty} \frac{\mathbb{E}_{\pi}N_x}{x} \leq \overline{N} \nonumber\\
&& \label{eqn:constrained_problem_average_cost_with_outage_cost}
\end{eqnarray}
\normalsize

The following standard result tells us how to choose the {\em Lagrange multipliers} 
$\xi_{out}$ and $\xi_{relay}$ (see \cite{beutler-ross85optimal-policies-controlled-markov-chains-constraint}, 
Theorem~$4.3$):

\begin{theorem}\label{theorem:how-to-choose-optimal-Lagrange-multiplier}
 Consider the constrained problem (\ref{eqn:constrained_problem_average_cost_with_outage_cost}). If there exists a pair 
$\xi_{out}^* \geq 0$, $\xi_{relay}^* \geq 0$ and a policy $\pi^*$ 
such that $\pi^*$ is the optimal policy of the unconstrained problem 
(\ref{eqn:unconstrained_problem_average_cost_with_outage_cost}) under $(\xi_{out}^*, \xi_{relay}^*)$ and the constraints in 
(\ref{eqn:constrained_problem_average_cost_with_outage_cost}) are met with equality under $\pi^*$, 
then $\pi^*$ is an optimal policy for (\ref{eqn:constrained_problem_average_cost_with_outage_cost}) also.\qed
\end{theorem}

\vspace{-3mm}
\section{Pure As-You-Go Deployment}\label{section:no_backtracking_average_cost}
\vspace{-2mm}
\subsection{Markov Decision Process (MDP) Formulation}\label{subsection:mdp_formulation_sum_power_sum_outage_no_backtracking}
Here we seek to solve problem~(\ref{eqn:unconstrained_problem_average_cost_with_outage_cost}), for 
the pure as-you-go approach. When the agent is $r$ steps away from the previous node ($A+1 \leq r \leq A+B$), 
he measures the outage probabilities $\{Q_{out}(r,\gamma,w)\}_{\gamma \in \mathcal{S}}$ 
on the link from the current location to the previous node, where $w$ is the realization of the shadowing 
random variable in the link being evaluated. 
He uses the knowledge of $r$ and the outage probabilities to decide whether to place a node at his current location, 
and what transmit power 
$\gamma \in \mathcal{S}$ to use if he places a relay. 
In this case, we formulate the impromptu deployment problem as a  Markov 
Decision Process (MDP) with state space $\{A+1,\cdots,A+B\} \times \mathcal{W}$. 
At state $(r,w), (A+1) \leq r \leq (A+B-1), w \in \mathcal{W}$, the action is either to 
place a relay and select a transmit power, or not to place. When $r=A+B$, the only 
feasible action is to place and select a transmit power $\gamma \in \mathcal{S}$. 
If, at state $(r,w)$, a relay is placed and it is set to use transmit power $\gamma$, 
a hop-cost of $\gamma+\xi_{out}Q_{out}(r,\gamma,w)+\xi_{relay}$ is incurred.
\footnote{We have taken $(r,w)$ as a typical state for simplicity of representation;  
so long as the channel model given by (\ref{eqn:channel_model}) is valid, we can also take  
$(r,\{Q_{out}(r,\gamma,w)\}_{\gamma \in \mathcal{S}})$ as a typical state. This happens because the cost of an action 
depends on the state $(r,w)$ only via the outage probabilities.}

A deterministic Markov policy $\pi$ is a sequence of mappings $\{\mu_k\}_{k \geq 1}$ 
from the state space to the action space, and it is called a stationary policy if $\mu_k=\mu$ for all $k$. 
Given the state (i.e., the measurements), the placement decision is made according to the policy. 

\vspace{-2mm}
\subsection{Formulation for $L \sim Geometric(\theta)$}\label{subsection:mdp_formulation_sum_power_sum_outage_no_backtracking_initial_geometric_length}
Under the pure as-you-go approach, we will first 
minimize the expected total cost for $L \sim Geometric(\theta)$, and then take 
$\theta \rightarrow 0$; this approach provides the policy structure for the 
average cost problem (see \cite{bertsekas07dynamic-programming-optimal-control-2}, Chapter~4).

In the $L \sim Geometric(\theta)$ case, the deployment process regenerates (probabilistically) after placing a relay, 
because of the memoryless property of the geometric distribution, and because of the fact that deployment of 
a new node will involve measurement of qualities of new links not measured before, and the new links have i.i.d. shadowing 
independent of the previously measured links. The (special) state of the system at such regeneration points is denoted by $\mathbf{0}$ 
(apart from the  states of the form $(r,w)$). 
When the source is placed at the end of the line, the process terminates. 
Suppose $N$ is the (random) number of relays placed, and 
node $N+1$ is the source node (as shown in Figure~\ref{fig:why-impromptu}). 
We first seek to solve the following:
\begin{equation}
 \min_{\pi \in \Pi} \mathbb{E}_{\pi} \bigg( \sum_{i=1}^{N+1}\Gamma_i+ 
\xi_{out} \sum_{i=1}^{N+1}Q_{out}^{(i,i-1)} +\xi_{relay} N \bigg)
\label{eqn:sum_power_discounted_no_backtracking}
\end{equation}

We will first investigate this approach assuming finite $\mathcal{W}$.

\subsection{Bellman Equation}\label{subsection:bellman_equation_sum_power_sum_outage_no_backtracking}
\vspace{-1mm} 

Let us denote the optimal expected cost-to-go at state $(r,w)$ and at 
state $\mathbf{0}$ be $J(r,w)$ and $J(\mathbf{0})$ respectively. 
Note that here we have an infinite horizon
total cost MDP with a finite state space and finite action
space. The assumption P of Chapter~$3$ in \cite{bertsekas07dynamic-programming-optimal-control-2} is satisfied, 
since the single-stage costs are nonnegative. 
Hence, by the theory developed in \cite{bertsekas07dynamic-programming-optimal-control-2}, 
we can focus on the class of stationary deterministic
Markov policies. 

By Proposition~$3.1.1$ of \cite{bertsekas07dynamic-programming-optimal-control-2}, the optimal value function $J(\cdot)$
satisfies the Bellman equation which is given by, for all $(A+1) \leq r \leq (A+B-1)$,

\footnotesize
\begin{eqnarray}
 J(r,w)&=&\min \bigg\{ \min_{\gamma \in \mathcal{S}} ( \gamma+\xi_{out} Q_{out}(r,\gamma,w) )+\xi_{relay} + J(\mathbf{0}), \nonumber\\
&& \theta \mathbb{E}_W \min_{\gamma \in \mathcal{S}} (\gamma+ \xi_{out}  Q_{out} (r+1,\gamma,W)) \nonumber\\
&& + (1-\theta)\mathbb{E}_W J(r+1,W)  \bigg\},\nonumber\\
J(A+B, w)&=&\min_{\gamma \in \mathcal{S}} (\gamma+\xi_{out} Q_{out}(A+B,\gamma,w)+\xi_{relay} ) + J(\mathbf{0}) \nonumber\\
J(\mathbf{0})&=& \sum_{k=1}^{A+1} (1-\theta)^{k-1}\theta \mathbb{E}_W \min_{\gamma \in \mathcal{S}} (\gamma+ \xi_{out} Q_{out} (k,\gamma,W))\nonumber\\
&& +(1-\theta)^{A+1}\mathbb{E}_W J(A+1,W)\label{eqn:bellman_equation_sum_power_sum_outage_no_backtracking}
\end{eqnarray}
\normalsize

These equations are understood as follows. If the current state is $(r,w), (A+1) \leq r \leq (A+B-1)$ and the line has not ended yet, we can either place a relay 
and set its transmit power to $\gamma \in \mathcal{S}$, or we may not place. If we place, 
the cost $\min_{\gamma \in \mathcal{S}} (\gamma+\xi_{out} Q_{out}(r,\gamma,w)+\xi_{relay} )$ is incurred at the current step, 
and the 
cost-to-go from there is $J(\mathbf{0})$. If we do not 
place a relay, the line will end with probability $\theta$ in the next step, in which case a 
cost $\mathbb{E}_W \min_{\gamma \in \mathcal{S}} (\gamma+ \xi_{out} Q_{out} (r+1,\gamma,W))$ will be 
incurred. If the line does not end in the next step, the next state will be a random state $(r+1,W)$ 
and a mean cost of 
$\mathbb{E}_W J(r+1,W)$ will be incurred. At state $(A+B,w)$ the only possible decision is to place a relay. 
At state $\mathbf{0}$, the deployment agent starts walking until he encounters the source location or location $(A+1)$; if the line 
ends at step $k, 1 \leq k \leq A+1$ (with probability $(1-\theta)^{k-1}\theta$), a cost of  
$\mathbb{E}_W \min_{\gamma \in \mathcal{S}} (\gamma+ \xi_{out} Q_{out} (k,\gamma,W))$ is incurred. If the line does not end within 
$(A+1)$ steps (this event has probability $(1-\theta)^{A+1}$), the next state will be $(A+1,W)$.

\vspace{-4mm}
\subsection{Value Iteration}\label{subsection:value_iteration_sum_power_sum_outage_no_backtracking}

The value iteration for (\ref{eqn:sum_power_discounted_no_backtracking}) is obtained by 
replacing $J(\cdot)$ in (\ref{eqn:bellman_equation_sum_power_sum_outage_no_backtracking}) 
by $J^{(k+1)}(\cdot)$ on the L.H.S (left hand side) and by $J^{(k)}(\cdot)$ on the R.H.S 
(right hand side), and by taking $J^{(0)}(\cdot)=0$ for all states. 
The standard MDP theory says that $J^{(k)}(\cdot) \uparrow J(\cdot)$ for all states as $k \rightarrow \infty$.

\vspace{-2mm}
\subsection{Policy Structure: OptAsYouGo Algorithm}\label{subsection:policy_structure_sum_power_sum_outage_no_backtracking}

\begin{lemma}\label{lemma:value_function_properties_sum_power_sum_outage_no_backtracking}
 $J(r,w)$ is increasing in $r$, $\xi_{out}$ and $\xi_{relay}$, decreasing in $w$, and jointly concave in 
$\xi_{out}$ and $\xi_{relay}$. $J(\mathbf{0})$ is increasing and jointly concave in $\xi_{out}$ and $\xi_{relay}$.
\end{lemma}

\begin{proof}
 See Appendix~\ref{appendix:average_cost_sum_outage_no_backtracking}.
\end{proof}

{\em Next, we propose an optimal algorithm 
OptAsYouGo (Optimal algorithm with pure As-You-Go approach).}

\begin{algorithm}\label{algorithm:OptAsYouGo}
{\em (OptAsYouGo Algorithm)}  At state $(r,w)$ (where $A+1 \leq r \leq A+B-1$), place a relay if and only if  
$\min_{\gamma \in \mathcal{S}} (\gamma+\xi_{out} Q_{out}(r,\gamma,w) ) \leq c_{th}(r)$, where 
 $c_{th}(r):=\theta \mathbb{E}_W \min_{\gamma \in \mathcal{S}} (\gamma+ \xi_{out}  Q_{out} (r+1,\gamma,W)) 
 + (1-\theta)\mathbb{E}_W J(r+1,W) -(\xi_{relay} + J(\mathbf{0}))$ is a 
threshold increasing in $r$. If the decision is to place a relay, the optimal power to be selected is given by 
$\argmin_{\gamma \in \mathcal{S}} (\gamma+\xi_{out} Q_{out}(r,\gamma,w))$. At state $(A+B,w)$, 
select  transmit power $\argmin_{\gamma \in \mathcal{S}} (\gamma+\xi_{out} Q_{out}(A+B,\gamma,w))$.\qed
\end{algorithm}

\begin{theorem}\label{theorem:policy_structure_sum_power_sum_outage_no_backtracking}
Under the pure as-you-go approach, Algorithm~\ref{algorithm:OptAsYouGo} provides the optimal policy for 
Problem~(\ref{eqn:unconstrained_problem_average_cost_with_outage_cost}).
\end{theorem}

\begin{proof}
 See Appendix~\ref{appendix:average_cost_sum_outage_no_backtracking}.
\end{proof}

{\em Remark:} The trade-off in the impromptu deployment problem is that if we place relays far apart, the cost due to outage 
increases, but the cost of placing the relays decreases. 
The intuition behind the threshold structure of the policy is that if at distance $r$ we get a good link with the combination of 
power and outage less than a threshold, 
then we should accept that link because moving forward is unlikely to yield a better link.
 $c_{th}(r)$ is increasing in $r$.   
Since $Q_{out}(r,\gamma,w)$ is increasing in $r$ for any $\gamma,w$, and since shadowing is i.i.d across links, 
the probability of a link (to the previous node) having desired QoS decreases as we move away from the previous 
node. Hence, the optimal policy will try to place relays as soon as possible if $r$ is large, and 
this explains why  $c_{th}(r)$ is increasing in $r$. 
Note that the threshold $c_{th}(r)$ does not depend on $w$, due to the fact that 
shadowing is i.i.d. across links.

\vspace{-2mm}
\subsection{Computation of the Optimal Policy}\label{subsection:policy_computation_sum_power_sum_outage_no_backtracking}
\vspace{-2mm}
Let us write $V(r):=\mathbb{E}_W J \left(r, W \right)=\sum_{w \in \mathcal{W}} p_W(w) J \left(r, w \right)$, and  
$V(\mathbf{0}):=J(\mathbf{0})$. Also, for each stage $k \geq 0$ of the value iteration, 
define $V^{(k)}(r):=\mathbb{E}_W J^{(k)}\left(r, W \right)$ and 
$V^{(k)}(\mathbf{0}):=J^{(k)}(\mathbf{0})$. Multiplying both sides of the value iteration by 
$p_W(w)$ and summing over $w \in \mathcal{W}$, we obtain an iteration in terms of $V^{(k)}(\cdot)$ and 
this iteration does not involve $J^{(k)}(\cdot)$. Since $J^{(k)}(r,w) \uparrow J(r,w)$ for each $r$, $w$ and 
$J^{(k)}(\mathbf{0}) \uparrow J(\mathbf{0})$ as $k \uparrow \infty$, we can argue that 
$V^{(k)}(r)\uparrow \mathbb{E}_W J(r, W)=V(r)$ for all 
$r$ (by Monotone Convergence Theorem) and  
$V^{(k)}(\mathbf{0})\uparrow J(\mathbf{0})=V(\mathbf{0})$. Then  
we can compute $c_{th}(r)$ by knowing $V(\cdot)$ itself 
(see the expression of $c_{th}(r)$ in Algorithm~\ref{algorithm:OptAsYouGo}); 
we need not keep track of the cost-to-go values 
$J^{(k)}(r, w)$ for each state $(r,w)$, at each stage $k$. Here we simply need to keep track of $V^{(k)}(\cdot)$. 
Similar iterations were proposed in 
\cite{chattopadhyay-etal13measurement-based-impromptu-placement_wiopt} (Section~III-A-5).

\vspace{-2mm}
\subsection{Average Cost Problem: Optimality of OptAsYouGo}
\label{subsection:average-cost-no-backtracking-limit approach}
\vspace{-2mm}
Note that the problem 
(\ref{eqn:sum_power_discounted_no_backtracking}) can be considered as an infinite 
horizon discounted cost problem with discount factor 
$(1-\theta)$. Hence, keeping in mind that we have finite state and action spaces, we observe that  
for the discount factor sufficiently close to $1$, i.e., for $\theta$ sufficiently close to $0$, 
{\em the optimal policy for problem in (\ref{eqn:sum_power_discounted_no_backtracking}) is optimal for the 
problem in (\ref{eqn:unconstrained_problem_average_cost_with_outage_cost})} 
(see \cite[Proposition~4.1.7]{bertsekas07dynamic-programming-optimal-control-2}). In particular, the optimal average 
cost per step with pure as-you-go approach, $\lambda^*$, is given by $\lambda^*=\lim_{\theta \rightarrow 0}\theta J_{\theta}(\mathbf{0})$ 
 (see \cite[Section~4.1.1]{bertsekas07dynamic-programming-optimal-control-2}), 
where $J_{\theta}(\mathbf{0})$ is the optimal cost for problem (\ref{eqn:sum_power_discounted_no_backtracking})  
under pure as-you-go with the probability of the line ending in the next step is $\theta$.

In case $\mathcal{W}$ is a Borel subset of $\mathbb{R}$, we still have a finite action 
space, and bounded, nonnegative cost per step. By 
\cite[Theorem~$5.5.4$]{lerma-lasserre96mdp-book}, one can show that 
the optimal average cost per step is again $\lambda^*=\lim_{\theta \rightarrow 0}\theta J_{\theta}(\mathbf{0})$. 
As $\theta \downarrow 0$, we will obtain a sequence of optimal policies (i.e., mappings from the state 
space to the action space), and a limit point of them will be an average cost optimal policy.

\vspace{-3mm}
\subsection{HeuAsYouGo: A Suboptimal Pure As-You-Go Heuristic}\label{subsection:HeuAsYouGo}
\vspace{-1mm}

\begin{algorithm}\label{algorithm:HeuAsYouGo}
{\em (HeuAsYouGo)} The power used by the relays is set to a 
fixed value. At each potential location, the deployment
agent checks whether the outage to the previous relay meets a
certain predetermined target with this fixed transmit power level. After placing a relay, the next relay is placed
at the last location where the target outage is met; or place
at the $(A+1)$-st location (after the previously placed relay) in the
unlikely situation where the target outage is violated in the
$(A+1)$-st location itself. If the agent reaches the $(A+B)$-th step and if all previous locations violate the outage 
target, he must place the next relay at step $(A+B)$. \qed
\end{algorithm}

HeuAsYouGo is a modified version of the heuristic deployment algorithm proposed in 
\cite{souryal-etal07real-time-deployment-range-extension}. HeuAsYouGo is not exactly a pure as-you-go algorithm 
since it sometimes requires the agent to move one step back in case the outage target is violated.

\vspace{-3mm}
\section{Explore Forward Deployment}\label{sec:backtracking_average_cost}
\vspace{-1mm}

\subsection{Semi-Markov Decision Process (SMDP) Formulation}\label{subsection:smdp_formulation_sum_power_sum_outage_backtracking}
\vspace{-1mm}

Let us recall explore-forward deployment from Section~\ref{subsection:deployment_process_notation}; 
we denote by $w_u$ the realization of shadowing in the potential link 
between the $u$-th location (starting from the previously placed node) and the previous node 
(see Figure~\ref{fig:backtracking_illustration}). The agent measures the outage 
probabilities $\{Q_{out}(u,\gamma,w_u)\}_{A+1 \leq u \leq A+B, \gamma \in \mathcal{S}}$ from locations 
$(A+1),\cdots,(A+B)$ at all available transmit power levels from the set $\mathcal{S}$, in order to make a placement decision.

Here we seek to solve the unconstrained problem~(\ref{eqn:unconstrained_problem_average_cost_with_outage_cost}). 
We formulate our problem as a 
Semi-Markov Decision Process (SMDP) with state space $\mathcal{W}^B$ and action space 
$\{A+1,A+2,\cdots,A+B\} \times \mathcal{S}$. 
The vector $\underline{w}:=(w_{A+1},w_{A+2},\cdots,w_{A+B})$, i.e., the shadowing from $B$ locations,  
is the state in our SMDP. 
In the state $\underline{w}$, an action $(u,\gamma) \in \{A+1,A+2,\cdots,A+B\} \times \mathcal{S}$ is taken 
where $u$ is the distance of the next relay (from the previous relay) that would be placed and 
$\gamma$ is the transmit power that this relay will use. In this case, 
a hop-cost of $\gamma+\xi_{out}Q_{out}(u,\gamma,w_u)+\xi_{relay}$ is incurred. 
After placing a relay, the next state becomes 
$\underline{w}^{'}:=(w_{A+1}^{'},w_{A+2}^{'},\cdots,w_{A+B}^{'})$ with probability 
$g(\underline{w}^{'}):=\prod_{r=A+1}^{A+B}p_{W_{r}}(w_{r}^{'})$ (since shadowing is i.i.d. across links). 

Let us denote, by the vector $\underline{W}(k)$, the (random) state at the $k$-th decision instant, and by 
$\mu_k(\underline{W}(k))$ the action at the $k$-th decision instant. 
For a deterministic Markov policy $\pi := \{\mu_k\}_{k \geq 1}$, let us define the functions 
$\mu_k^{(1)}:\mathcal{W}^B \rightarrow \{A+1,A+2,\cdots,A+B\}$ 
and $\mu_k^{(2)}:\mathcal{W}^B \rightarrow \mathcal{S}$ as follows: if $\mu_k(\underline{w})=(u,\gamma)$, then 
$\mu_k^{(1)}(\underline{w})=u$ and $\mu_k^{(2)}(\underline{w})=\gamma$. 

\vspace{-2mm}
\subsection{Policy Structure: Algorithm OptExploreLim}\label{subsec:smdp-policy-structure}
 \vspace{-2mm}

Note that, $\underline{W}(k)$ is  
i.i.d across $k, k \geq 1$. The state space is a Borel space and the action space is finite. 
The hop cost and hop length (in number of steps) are uniformly bounded across all state-action pairs. 
Hence, we can work with stationary deterministic policies (see \cite{tijms03stochastic-models} for 
finite state space, i.e., finite $\mathcal{W}$, and \cite{amaya-vasquez00average-cost-optimality-SMDP-borel-spaces} 
for a general Borel state space, i.e., when $\mathcal{W}$ is a Borel set). Under our current scenario, 
the optimal average cost per step, $\lambda^{*}$, exists (in fact, the limit exists) 
and is same for all states  $\underline{w} \in \mathcal{W}^B$. 
{\em For simplicity, we work with finite $\mathcal{W}$, but the policy structure holds for 
Borel state space also.}

We next present a deployment algorithm called ``OptExploreLim,'' an optimal algorithm for limited exploration.

\begin{algorithm}\label{algorithm:policy_structure_smdp_backtracking}
{\em (OptExploreLim Algorithm:)} In the state 
$\underline{w}$ which is captured by the measurements $\{Q_{out}(u,\gamma,w_u)\}$ for $A+1 \leq u \leq A+B$, $\gamma \in \mathcal{S}$, 
place the new relay according to the stationary policy $\mu^*$ as follows: 

\vspace{-3mm}

\begin{footnotesize}
\begin{eqnarray}
 \mu^*(\underline{w})= \argmin_{u, \gamma} \bigg( \gamma + \xi_{out} Q_{out}(u,\gamma,w_u) + \xi_{relay} -\lambda^{*}u \bigg)  
\label{eqn:smdp-optimal-policy}
\end{eqnarray}
\end{footnotesize}

\vspace{-3mm}

where $\lambda^*$ (or $\lambda^*(\xi_{out},\xi_{relay})$) 
is the optimal average cost per step for the Lagrange multipliers $(\xi_{out},\xi_{relay})$. \qed
\end{algorithm}

\begin{theorem}\label{theorem:policy_structure_smdp_backtracking}
  The policy $\mu^*$ given by Algorithm~\ref{algorithm:policy_structure_smdp_backtracking} 
is optimal for the problem (\ref{eqn:unconstrained_problem_average_cost_with_outage_cost}) 
under the explore-forward approach.
\end{theorem}
 
\begin{proof}
 The optimality equation for the SMDP is given by 
(see \cite{tijms03stochastic-models}, Equation~7.2.2):
\begin{eqnarray}
 v^*(\underline{w}) &=& \min_{u, \gamma}\bigg\{ \gamma+  \xi_{out} Q_{out}(u,\gamma,w_u)+\xi_{relay} \nonumber\\
&& -\lambda^{*}u+\sum_{\underline{w}^{'} \in \mathcal{W}^B}g(\underline{w}^{'}) v^*(\underline{w}^{'}) \bigg\}\label{eqn:optimality-smdp}
\end{eqnarray}
$v^*(\underline{w})$ is the optimal differential cost corresponding to state $\underline{w}$. 
The structure of the optimal policy is obvious 
from (\ref{eqn:optimality-smdp}), since 
$\sum_{\underline{w}^{'} \in \mathcal{W}^B } g(\underline{w}^{'}) v^*(\underline{w}^{'})$ does not depend on $(u,\gamma)$.
\end{proof}

{\em Later we will also use the notation $\pi^*$ or $\pi^*(\xi_{out},\xi_{relay})$ to denote the 
OptExploreLim policy under the pair $(\xi_{out},\xi_{relay})$, since here $\pi^*=\{\mu^*, \mu^*, \mu^*, \cdots\}$.}

\begin{remark}
The same optimal policy structure will hold for a Borel state space, 
 by the theory presented in \cite{amaya-vasquez00average-cost-optimality-SMDP-borel-spaces}.
\end{remark}

\begin{remark}
 {\em The optimal decision depends on  state $\underline{w}$ only via 
the outage probabilities which can be easily measured.} 
\end{remark}

\begin{remark}
For an action $(u,\gamma)$, a cost 
$( \gamma +  \xi_{out} Q_{out}(u,\gamma,w_u)+ \xi_{relay})$ will be incurred. 
On the other hand, $\lambda^*u$ is the reference cost over  $u$ steps. 
The policy  minimizes the difference between these two for each link. 
\end{remark}

\begin{remark}
The policy requires the deployment agent to know $\lambda^*$, and computation of $\lambda^*$ will require 
perfect knowledge of propagation environment (e.g., the path-loss exponent $\eta$ in 
(\ref{eqn:channel_model}), the distribution of shadowing, etc.); see Section~\ref{subsection:smdp-policy-computation}. 

\end{remark}

\begin{theorem}\label{theorem:smdp-cost-vs-xi}
$\lambda^{*}(\xi_{out},\xi_{relay})$ is jointly concave, increasing and 
continuous in   $\xi_{out}$ and $\xi_{relay}$.
\end{theorem}

\begin{proof}
See Appendix~\ref{appendix:backtracking_average_cost}.
\end{proof}

Let us consider a sub-class of stationary deployment policies (parameterized by $\lambda \geq 0$, 
$\xi_{out} \geq 0$ and $\xi_{relay} \geq 0$) given by: 

\vspace{-3mm}

\begin{footnotesize}
\begin{eqnarray}
 \mu(\underline{w})= \argmin_{u, \gamma} \bigg( \gamma + \xi_{out} Q_{out}(u,\gamma,w_u) + \xi_{relay} -\lambda u \bigg)  
\label{eqn:smdp-optimal-policy_general_form}
\end{eqnarray}
\end{footnotesize}

\vspace{-5mm}

where $\lambda$ is not necessarily equal to $\lambda^*(\xi_{out}, \xi_{relay})$.

Under the class of policies given by (\ref{eqn:smdp-optimal-policy_general_form}), 
let $(U_k,\Gamma_k,Q_{out}^{(k,k-1)}), k \geq 1,$ denote the sequence of inter-node distances, transmit  
powers and link outage probabilities 
that the optimal policy yields during the deployment process. By the assumption of i.i.d. shadowing across links, 
it follows that $(U_k,\Gamma_k,Q_{out}^{(k,k-1)}), k \geq 1,$ is an i.i.d. sequence. 

Let $\overline{\Gamma}(\lambda,\xi_{out},\xi_{relay})$, $\overline{Q}_{out}(\lambda,\xi_{out},\xi_{relay})$ 
and $\overline{U}(\lambda,\xi_{out},\xi_{relay})$ 
denote the mean power per link, mean outage 
per link and mean placement distance (in steps) respectively, under the policy given 
by (\ref{eqn:smdp-optimal-policy_general_form}), where $\lambda$ is not necessarily equal to 
$\lambda^*(\xi_{out},\xi_{relay})$. Also, let 
$\overline{\Gamma}^*(\xi_{out},\xi_{relay})$, $\overline{Q}_{out}^*(\xi_{out},\xi_{relay})$ and $\overline{U}^*(\xi_{out},\xi_{relay})$ 
denote the optimal mean power per link, the optimal mean outage 
per link and the optimal mean placement distance (in steps) respectively, under the OptExploreLim algorithm 
(i.e., policy $\pi^*(\xi_{out},\xi_{relay})$ when $\lambda$ in (\ref{eqn:smdp-optimal-policy_general_form}) is replaced by 
$\lambda^*(\xi_{out},\xi_{relay})$). 
By the Renewal-Reward theorem, the optimal mean power per step, the optimal mean outage per step, 
and the optimal mean number of relays per step are given by 
$\frac{\overline{\Gamma}^*(\xi_{out},\xi_{relay})}{\overline{U}^*(\xi_{out},\xi_{relay})}$, 
$\frac{\overline{Q}_{out}^*(\xi_{out},\xi_{relay})}{\overline{U}^*(\xi_{out},\xi_{relay})}$ 
and $\frac{1}{\overline{U}^*(\xi_{out},\xi_{relay})}$.

\begin{theorem}\label{theorem:outage_decreasing_with_xio_placement_rate_decreasing_with_xir}
 For a given $\xi_{out}$, the mean number of relays per step under the OptExploreLim algorithm  
(Algorithm~\ref{algorithm:policy_structure_smdp_backtracking}), 
$\frac{1}{\overline{U}^*(\xi_{out},\xi_{relay})}$, decreases with $\xi_{relay}$. Similarly, 
for a given $\xi_{relay}$, the mean outage probability per step, $\frac{\overline{Q}_{out}^*(\xi_{out},\xi_{relay})}{\overline{U}^*(\xi_{out},\xi_{relay})}$, 
decreases with $\xi_{out}$ under the optimal policy.
\end{theorem}

\begin{proof}
 See Appendix~\ref{appendix:backtracking_average_cost}.
\end{proof}

{\em Remark:} The proof of Theorem~\ref{theorem:outage_decreasing_with_xio_placement_rate_decreasing_with_xir} is 
quite general; the results hold for the pure as-you-go approach also.

\begin{theorem}\label{theorem:optimality-condition-required-for-backtracking-average-cost-learning}
For Problem~(\ref{eqn:unconstrained_problem_average_cost_with_outage_cost}), 
under the optimal policy (with explore-forward approach)  
characterized by $\lambda^*$ (i.e., under the OptExploreLim algorithm), we have 
$\mathbb{E}_{\underline{W}} \min_{u, \gamma}(\gamma+ \xi_{out} Q_{out}(u,\gamma,W_u)+\xi_{relay}-\lambda^*u )=0$.
\end{theorem}
\begin{proof}
See Appendix~\ref{appendix:backtracking_average_cost}.
\end{proof}

\vspace{-6mm}
\subsection{Policy Computation}\label{subsection:smdp-policy-computation}
\vspace{-2mm}
We adapt a policy iteration (from \cite{tijms03stochastic-models}) based algorithm 
to calculate $\lambda^{*}$. The algorithm generates a sequence of stationary  policies 
$\{\mu_k\}_{k \geq 1}$ (note that the notation $\mu_k$ was used for a different purpose in 
Section~\ref{subsection:smdp_formulation_sum_power_sum_outage_backtracking}; here each $\mu_k$ is a stationary, deterministic, Markov policy), 
such that for any $k \geq 1$, 
$\mu_k(\cdot):\mathcal{W}^B \rightarrow \{A+1,\cdots,A+B\} \times \mathcal{S}$ maps a state into some action. 
Define the sequence $ \{\mu^{(1)}_k, \mu^{(2)}_k \}_{k \geq 1}$ of functions as 
follows: if $\mu_k(\underline{w})=(u,\gamma)$, then 
$\mu^{(1)}_k (\underline{w})=u$ and $\mu^{(2)}_k (\underline{w})=\gamma$.

\begin{algorithm}\label{algorithm:policy_iteration_smdp}
The policy iteration  algorithm is as follows:

{Step~$0$ (Initialization):} Start with an initial   policy $\mu_0$. 

{Step~$1$ (Policy Evaluation):} Calculate the average cost $\lambda_{k}$ corresponding to the policy $\mu_{k}$, for $k \geq 0$. 
$\lambda_k$ is equal to the following quantity (by the Renewal Reward Theorem): 

\footnotesize
\begin{eqnarray*}
\frac{  \xi_{relay}+\sum_{\underline{w}} g(\underline{w}) \bigg(\mu^{(2)}_k(\underline{w})+ \xi_{out} Q_{out}(\mu^{(1)}_k(\underline{w}),\mu^{(2)}_k(\underline{w}),w_{\mu^{(1)}_k(\underline{w})})\bigg) } 
{ \sum_{\underline{w}} g(\underline{w}) \mu^{(1)}_k(\underline{w})  } 
\end{eqnarray*}
\normalsize

{Step~$2$ (Policy Improvement):} Find a new policy $\mu_{k+1}$ by solving the following:
\footnotesize
\begin{eqnarray}
 \mu_{k+1}(\underline{w})=\argmin_{(u,\gamma) } \bigg( \gamma+Q_{out}(u,\gamma,w_u)+\xi_{relay} 
-\lambda_{k}u \bigg)\label{eqn:smdp-policy-improvement}
\end{eqnarray}
\normalsize

If $\mu_{k}$ and $\mu_{k+1}$ are  same  (i.e., if $\lambda^{(k-1)}=\lambda_k$), then stop and declare $\mu^{*}=\mu_{k}$, $\lambda^{*}=\lambda_{k}$. Otherwise, go to 
Step~$1$.\qed
\end{algorithm}

{\em Remark:} {\em By the theory in \cite{tijms03stochastic-models}, this 
policy iteration will converge (to $\lambda^*$) in a finite number of iterations, for 
finite state and action spaces.} 
For a general Borel state space (e.g., for log-normal shadowing), 
only asymptotic convergence to $\lambda^*$ can be guaranteed.

{\em Computational Complexity:} The finite state space has cardinality $|\mathcal{W}|^B$. Then, $O(|\mathcal{W}|^B)$ 
addition operations are required to compute $\lambda_{k}$ from the policy evaluation step. 
However, careful manipulation leads 
to a drastic reduction in this computational requirement, as shown by the following theorem.

\begin{theorem}\label{theorem:complexity-reduction-smdp-policy-iteration}
 In the policy evaluation step in Algorithm~\ref{algorithm:policy_iteration_smdp}, we can reduce 
the number of computations in each iteration from $|\mathcal{W}|^B$ to $O(B^2 M^2 |\mathcal{W}|^2)$. 
\end{theorem}
\begin{proof}
 See Appendix~\ref{appendix:backtracking_average_cost}.
\end{proof}

\vspace{-4mm}
\subsection{HeuExploreLim: An Intuitive but Suboptimal Heuristic}\label{subsection:HeuExploreLim}
\vspace{-2mm}
A natural heuristic for (\ref{eqn:unconstrained_problem_average_cost_with_outage_cost}) under 
the explore-forward approach is the following HeuExploreLim Algorithm (Heuristic Algorithm for Limited Explore-Forward):

\begin{algorithm}
{\em (HeuExploreLim Algorithm)} Under the explore-forward setting as discussed in Section~\ref{sec:backtracking_average_cost}, at state $\underline{w}$, 
make the decision according to the following rule:
\begin{equation*}
 (u^*,\gamma^*)=\argmin_{u, \gamma } \frac{\gamma+\xi_{out} Q_{out}(u,\gamma,w_u)+\xi_{relay}}{u}\,\,\,\,\qed
\end{equation*}
\end{algorithm}

Under any stationary deterministic policy $\mu$, 
let us denote the cost of a link by $C_{\mu}$ (a random variable) and the length 
of a link by $U_{\mu}$ (under any stationary deterministic policy $\mu$, the 
deployment process regenerates at the placement points). 

\begin{lemma}\label{lemma:what_is_the_heuristic_doing}
 HeuExploreLim solves $\inf_{\mu}\mathbb{E}_{\mu}( \frac{C_{\mu}}{U_{\mu}} )$.
\end{lemma}

\begin{proof}
See Appendix~\ref{appendix:backtracking_average_cost}.
\end{proof}

{\em Remark:} This  heuristic is not optimal. Our optimal policy given in Theorem~(\ref{theorem:policy_structure_smdp_backtracking}) solves 
$\inf_{\mu} \frac{\mathbb{E}_{\mu}(C_{\mu})}{\mathbb{E}_{\mu}(U_{\mu})}$. However, HeuExploreLim solves 
$\inf_{\mu}\mathbb{E}_{\mu}( \frac{C_{\mu}}{U_{\mu}} )$, which is, in general, different from 
$\inf_{\mu} \frac{\mathbb{E}_{\mu}(C_{\mu})}{\mathbb{E}_{\mu}(U_{\mu})}$. Note that 
$\mathbb{E}_{\mu}( \frac{C_{\mu}}{U_{\mu}} )=\frac{\mathbb{E}_{\mu}(C_{\mu})}{\mathbb{E}_{\mu}(U_{\mu})}$ 
{\em if and only if} the variance of $U_{\mu}$ is zero. 
But this does not happen  due to the variability in shadowing over space.

\vspace{-2mm}
\section{Comparison between Explore-Forward and Pure As-You-Go Approaches}
\label{section:numerical_results_comparison_average_cost_backtracking_no_backtracking}
\vspace{-2mm}

Let us denote the optimal average cost per step (for a given $\xi_{out}$ and $\xi_{relay}$) under 
the explore-forward and pure as-you-go approaches by $\lambda_{ef}^*$ and $\lambda_{ayg}^*$.
\begin{theorem}\label{theorem:comparison-backtracking-no-backtracking}
 $\lambda_{ef}^* \leq \lambda_{ayg}^*$.
\end{theorem}
\begin{proof}
 See Appendix~\ref{appendix:comparison_average_cost_backtracking_no_backtracking}. 
\end{proof}

Next, we numerically compare various deployment algorithms, in order 
to select the best algorithm  for  deployment.

\vspace{-6mm}
\subsection{Parameter Values  Used in the Numerical Comparisons}\label{subsection:parameter_values}
\vspace{-2mm}
We consider deployment for a given $\xi_{out}$ and a given $\xi_{relay}$, for the objective in 
(\ref{eqn:unconstrained_problem_average_cost_with_outage_cost}). 
We provide numerical results for deployment with iWiSe motes (\cite{iwise}) (based on the 
Texas Instrument (TI) CC2520 
which implements the IEEE~$802.15.4$ PHY in the $2.4$~GHz ISM band, yielding a 
bit rate of $250$~Kbps, with a CSMA/CA medium access control (MAC)) equipped with $9$~dBi antennas. 
The set of transmit power levels $\mathcal{S}$ is taken to be $\{-18,-7,-4,0,5\}$~dBm, 
which is a subset of the transmit power levels available in the chosen device. 
For the channel model as in 
(\ref{eqn:channel_model}), our measurements 
in a forest-like environment inside the Indian Institute of Science Campus gave path-loss exponent 
$\eta=4.7$ and $c=10^{0.17}$ (i.e., $1.7$~dB); 
see \cite{chattopadhyay-etal14deployment-experience}. 
Shadowing $W$ was found to be log-normal; $W=10^{\frac{Y}{10}}$ with 
$Y \sim \mathcal{N}(0, \sigma^2)$, where $\sigma=7.7$~dB. Shadowing decorrelation distance was found to be $6$~meters. 
Fading is assumed to be Rayleigh; $H \sim Exponential(1)$.

We define outage to be the event when the received signal 
power of a packet falls below $P_{rcv-min}=10^{-9.7}$~mW (i.e., $-97$~dBm); for a commercial implementation of the PHY/MAC of
IEEE~$802.15.4$, $-97$~dBm received power corresponds to a $2\%$ packet
loss probability for $127$~byte packets for iWiSe motes, as per our measurements.

We consider deployment along a line with step size $\delta=20$~meters, $A=0$, $B=5$. 
Given $A=0$, we chose $B$ is the following way. 
Define a link to be good if its outage probability is less than
$3\%$, and choose $B$ to be the largest integer such that the
probability of finding a good link of length $B \delta$ is more than
$20\%$, when the highest transmit power is used (this will ensure that the agent does not measure very long links having 
poor outage probabilities). For the parameters $\eta=4.7$, $c=10^{0.17}$, 
$\sigma=7.7$~dB, and $5$~dBm transmit power, $B$ turned out to be $5$. If $B$ is increased further, 
the probability of getting a good link will be very small. \qed

\vspace{-4mm}
\subsection{Numerical Comparison Among Deployment Policies}\label{subsection:numerical_comparison_among_policies}
\vspace{-2mm}

Assuming these parameter values, we computed (by MATLAB) 
the mean power per step (in mW), mean outage per step,  mean placement distance (in steps), mean cost per step and 
mean number of measurements made per step, for the four deployment algorithms presented so far. The results are shown  
in Table~\ref{table:comparison_among_algorithms}.      
In order to make a fair comparison, {\em we used the mean power per node 
for OptAsYouGo as the fixed node transmit power for HeuAsYouGo, and the mean outage per link of OptAsYouGo 
as the  target outage for HeuAsYouGo.} The mean number of measurements per step is defined as the ratio of the mean 
number of links evaluated for deployment of one node and the mean placement distance (in steps). The numerator of this ratio is 
$B=5$ for explore-forward algorithms (since $A=0$). OptAsYouGo makes one measurement per step, but HeuAsYouGo makes more than one measurements 
per step since the agent often evaluates a bad link, takes one step back and places the relay. 
\footnote{For planned deployment, we will have to 
evaluate all possible potential links; from each potential location, we need to measure link quality to $B=5$ preceding 
potential locations, which is not feasible.}

{\em We notice that the average per-step cost (COST in Table~\ref{table:comparison_among_algorithms}) of 
OptExploreLim (OEL) is the least. 
OEL uses the least mean power per step (POW column), 
places nodes the widest apart (DIST column), and the mean outage per step (OUT column) 
is second to lowest. On the other hand, OEL requires about twice as many measurements per 
step as compared to OptAsYouGo.\footnote{A more detailed comparison among the algorithms can be found in 
Appendix~\ref{appendix:comparison_average_cost_backtracking_no_backtracking}, along with elaborate discussion.} 
Hence, we can conclude that the algorithms based on the 
explore-forward approach significantly outperform the algorithms based on the pure-as-you-go approach, at the cost 
of slightly more  measurements per step. 
Hence, for applications that do not require very rapid deployment, such as
deployment along a long forest trail for wildlife monitoring, 
explore-forward is a better approach to take. Thus, for the learning
algorithms presented later, we will consider only the explore-forward approach.  
However, under the requirement of fast deployment (e.g., emergency deployment by first responders), pure as-you-go or deployment 
without measurements (as in \cite{sinha-etal12optimal-sequential-relay-placement-random-lattice-path}) might be more 
suitable.}

\begin{table}[t!]
\footnotesize
\centering
\begin{tabular}{|c |c |c |c |c |c |}
\hline
Algorithm & POW (mW) & OUT   & DIST (steps) & COST  & MEAS  \\ 
\hline
OEL & 0.1955  &  0.001969  &  2.2859  & 0.8312  & 2.1873 \\
\hline
HEL & 0.2432  & 0.002507   &  2.673  & 0.8684  &  1.8706 \\
\hline
OAYG & 0.2904  &  0.003607  &  1.5  & 1.265  &  1\\
\hline
HAYG & 0.3318  &  0.001752  & 1.313   & 1.267  &  1.6969\\
\hline
\end{tabular}
\vspace{-0mm}
\caption{Numerical comparison among various algorithms for $\xi_{out}=100$ and $\xi_{relay}=1$. 
Abbreviations: OEL-OptExploreLim, 
HEL-HeuExploreLim,   OAYG-OptAsYouGo, HAYG-HeuAsYouGo. POW-Mean power (in mW unit) per step, OUT- Mean outage  
per step, DIST-Mean placement distance, COST-Mean cost per step, MEAS-Mean number of measurements 
per step.}
\vspace{-0mm}
\label{table:comparison_among_algorithms}
\vspace{-8mm}
\end{table}
\normalsize
\vspace{-0mm}

\vspace{-4mm}
\section{OptExploreLimLearning: Learning with Explore-Forward, for Given 
$\xi_{out}$ and $\xi_{relay}$}\label{section:learning_backtracking_given_xio_xir}
\vspace{-1mm}
Based on the discussion in Section~\ref{section:numerical_results_comparison_average_cost_backtracking_no_backtracking}, 
we proceed, in the rest of this paper, 
with developing learning algorithms based on the  policy   OptExploreLim (to solve 
problem (\ref{eqn:unconstrained_problem_average_cost_with_outage_cost})). 
We observe that the optimal policy 
(given by Algorithm~\ref{algorithm:policy_structure_smdp_backtracking}) can be completely specified 
by the optimal average cost per step $\lambda^*$, for given values of $\xi_{out}$ and $\xi_{relay}$. 
But the computation of $\lambda^*$ 
requires policy iteration. Policy iteration requires the channel model parameters $\eta$ and $\sigma$, and it is  
computationally intensive. 
In practice, these parameters of the channel model might not be available. Under this situation, 
the agent measures $\{Q_{out}(u,\gamma,w_u):A+1 \leq u \leq A+B, \gamma \in \mathcal{S} \}$ before deploying each relay, but 
he has to {\em learn} the optimal average cost per step in the process of deployment, and, 
use the corresponding updated policy each time he places a new relay. In order 
to address this requirement, we propose 
an algorithm which will maintain a running estimate of $\lambda^*$, and update it 
each time a relay is placed. The algorithm is motivated by  
the theory of Stochastic Approximation (see \cite{borkar08stochastic-approximation-book}), and it uses, as input, the 
measurements made for each placement, in order to improve the estimate of $\lambda^*$. 
We prove that, as the number of deployed relays goes to infinity, the running estimate of average network cost per step 
converges to 
$\lambda^*$ almost surely. 

After 
the deployment is over, let us denote the length, transmit power and outage values of the link between node~$k$ and 
node~$(k-1)$ by $u_k$, $\gamma_k$ and $Q_{out}^{(k,k-1)}$. After placing the $(k-1)$-st node, we will place 
node $k$, and consequently $u_k$, $\gamma_k$ and $Q_{out}^{(k,k-1)}$ will be decided  
by the following algorithm.

\begin{algorithm}\label{algorithm:learning_backtracking_given_xio_xir_general_step_size}
{\em (OptExploreLimLearning)} Let $\lambda^{(k)}$ be the estimate of the optimal average cost per step after placing the $k$-th relay 
(sink is  node~$0$), and 
let $\lambda^{(0)}$ be the initial estimate. In the process of placing relay $(k+1)$, 
if the measured outage probabilities are $\{Q_{out}(u,\gamma,w_u):A+1 \leq u \leq A+B, \gamma \in \mathcal{S} \}$, 
then place relay $(k+1)$ using the following policy:
\vspace{-2mm}

 \begin{footnotesize}
\begin{eqnarray*}
 (u_{k+1},\gamma_{k+1})= \argmin_{u , \gamma } \bigg( \gamma + \xi_{out} Q_{out}(u,\gamma,w_u) + \xi_{relay} -\lambda^{(k)} u \bigg) 
\label{eqn:learning_backtracking_given_xio_xir_decision_part}
\end{eqnarray*}
\end{footnotesize}
After placing relay $(k+1)$, update $\lambda^{(k)}$ as follows (using the measurements made in the 
process of placing relay $(k+1)$): 

\begin{footnotesize}
\begin{eqnarray}
&& \lambda^{(k+1)} \nonumber\\
&=& \lambda^{(k)}+ a_{k+1} \min_{u, \gamma } \bigg( \gamma + \xi_{out} Q_{out}(u,\gamma,w_u) + \xi_{relay} -\lambda^{(k)} u \bigg)\nonumber\\
&=& \lambda^{(k)}+ a_{k+1} \bigg( \gamma_{k+1} + \xi_{out} Q_{out}^{(k+1,k)} + \xi_{relay} -\lambda^{(k)} u_{k+1} \bigg) \label{eqn:learning_backtracking_given_xio_xir_update_part}
\end{eqnarray}
\end{footnotesize}
$\{a_k\}_{k \geq 1}$ is a decreasing sequence such that $a_k >0 \, \forall \, k \geq 1$, $\sum_k a_k=\infty$ and $\sum_k a_k^2< \infty$. 
One example is $a_k=\frac{1}{k}$.\qed
\end{algorithm}

\begin{theorem}\label{theorem:learning_backtracking_given_xio_xir_general_step_size}
If we employ Algorithm~\ref{algorithm:learning_backtracking_given_xio_xir_general_step_size} in the deployment process, 
we will have $\lambda^{(k)} \rightarrow \lambda^*$ almost surely.
\end{theorem}
\begin{proof}
By Theorem~\ref{theorem:optimality-condition-required-for-backtracking-average-cost-learning}, 
under OptExploreLim, we have 
$\mathbb{E}_{\underline{W}} \min_{u,\gamma} ( \gamma+ \xi_{out} Q_{out} (u,\gamma,W_u)+ \xi_{relay}-\lambda^*u )=0$; 
this   leads to the stochastic approximation update 
in Algorithm~\ref{algorithm:learning_backtracking_given_xio_xir_general_step_size}. The detailed proof can be 
found in Appendix~\ref{appendix:learning_backtracking_given_xio_xir}.
\end{proof}

While Algorithm~\ref{algorithm:learning_backtracking_given_xio_xir_general_step_size} utilizes the general stochastic 
approximation update, Algorithm~\ref{algorithm:learning_backtracking_given_xio_xir_special_step_size} 
ensures that the iterate  $\lambda^{(k)}$  
is the actual average network cost per step up to the $k$-th relay.

\begin{algorithm}\label{algorithm:learning_backtracking_given_xio_xir_special_step_size}
 Start with any $\lambda^{(0)}>0$. Let, for $k \geq 1$, $\lambda^{(k)}$ be the average cost per step 
for the portion of the network already deployed between the sink and the $k$-th relay, i.e., 
\begin{equation*}
 \lambda^{(k)}=\frac{\sum_{i=1}^k (\gamma_i+\xi_{out}Q_{out}^{(i,i-1)}+\xi_{relay})}{\sum_{i=1}^k u_i}
\end{equation*}

Place the $(k+1)$-st relay according to the following policy:

 \begin{footnotesize}
\begin{eqnarray*}
 (u_{k+1},\gamma_{k+1})= \argmin_{u , \gamma } \bigg( \gamma + \xi_{out} Q_{out}(u,\gamma,w_u) + \xi_{relay} -\lambda^{(k)} u \bigg) 
\label{eqn:learning_backtracking_given_xio_xir_decision_part}\qed
\end{eqnarray*}
\end{footnotesize}
\end{algorithm}

\begin{corollary}\label{corollary:learning_backtracking_given_xio_xir_special_step_size}
 
Under Algorithm~\ref{algorithm:learning_backtracking_given_xio_xir_special_step_size} in the deployment process, 
we will have $\lambda^{(k)} \rightarrow \lambda^*$ almost surely.
\end{corollary}

\begin{proof}
 See Appendix~\ref{appendix:learning_backtracking_given_xio_xir}.
\end{proof}

\vspace{-5mm}
\section{OptExploreLimAdaptiveLearning with Constraints on Outage Probability and Relay Placement Rate}
\label{section:learning_backtracking_adaptive_with_outage_cost}
\vspace{-2mm}
 In Section~\ref{section:learning_backtracking_given_xio_xir}, 
we provided a stochastic approximation algorithm for relay deployment, 
with given multipliers $\xi_{out}$ and $\xi_{relay}$, without knowledge of the propagation 
parameters. 
Let us recall that Theorem~\ref{theorem:how-to-choose-optimal-Lagrange-multiplier} tells us how to choose the 
Lagrange multipliers $\xi_{out}$ and $\xi_{relay}$ (if they exist) 
in (\ref{eqn:unconstrained_problem_average_cost_with_outage_cost}) 
in order to solve the problem given in (\ref{eqn:constrained_problem_average_cost_with_outage_cost}). 
However, we need to know the radio propagation parameters (e.g., $\eta$ and $\sigma$) in order to 
compute an optimal pair $(\xi_{out}^*, \xi_{relay}^*)$ (if it exists) so that both constraints in 
(\ref{eqn:constrained_problem_average_cost_with_outage_cost}) are met with equality. In real deployment scenarios, 
these propagation parameters might not be known.  
Hence, in this section, we provide a sequential placement and learning algorithm such that, as the relays are placed, 
the placement policy iteratively converges to the set of optimal policies for the constrained problem displayed 
in (\ref{eqn:constrained_problem_average_cost_with_outage_cost}). 
The policy is of the OptExploreLim type, and the cost of the deployed network converges 
to the optimal cost. We modify the 
OptExploreLimLearning algorithm so that a running estimate $(\lambda^{(k)},\xi_{out}^{(k)},\xi_{relay}^{(k)})$ gets 
updated each time a new relay is placed. 
The objective is to make sure that the running estimate $(\lambda^{(k)},\xi_{out}^{(k)},\xi_{relay}^{(k)})$ eventually converges to 
the set of optimal $(\lambda^*(\xi_{out},\xi_{relay}),\xi_{out},\xi_{relay})$ tuples as the deployment progresses. 
Our approach is via two time-scale stochastic approximation (see \cite[Chapter~$6$]{borkar08stochastic-approximation-book}).

 \vspace{-4mm}
\subsection{OptExploreLim: Effect of Multipliers $\xi_{out}$ and $\xi_{relay}$}
 \vspace{-2mm}
Consider the constrained problem in (\ref{eqn:constrained_problem_average_cost_with_outage_cost}) and its relaxed 
version in (\ref{eqn:unconstrained_problem_average_cost_with_outage_cost}). 
We will seek a policy for the problem in (\ref{eqn:constrained_problem_average_cost_with_outage_cost}) 
in the class of OptExploreLim policies (see (\ref{eqn:smdp-optimal-policy})).  
Clearly, there exists at least one tuple 
$(\overline{q},\overline{N})$ for which there exists a pair $\xi_{out}^*>0, \xi_{relay}^*>0$ such that, under 
the optimal policy $\pi^*(\xi_{out}^*, \xi_{relay}^*)$, both constraints are met with equality.  
In order to see this, choose any $\xi_{out}>0,\xi_{relay}>0$ 
and consider the corresponding optimal policy $\pi^*(\xi_{out}, \xi_{relay})$ (provided by OptExploreLim). 
Suppose that the mean outage per step and mean number of relays per step, under the policy 
$\pi^*(\xi_{out}, \xi_{relay})$, are $q_0$ and $n_0$, respectively. Now, if we set the constraints 
$\overline{q}=q_0$ and $\overline{N}=n_0$ in 
(\ref{eqn:constrained_problem_average_cost_with_outage_cost}), 
we obtain one instance of such a tuple $(\overline{q},\overline{N})$. 

On the other hand, there exist $(\overline{q},\overline{N})$ pairs which are not feasible. One 
example is the case $\overline{N}=\frac{1}{A+B}$ (i.e., inter-node distance is always $(A+B)$), 
along with $\overline{q}< \frac{\mathbb{E}_W Q_{out}(A+B,P_M,W)}{A+B}$, where $P_M$ is the 
maximum available transmit power level at each node. In this case, the outage constraint cannot be satisfied 
while meeting the constraint on the mean number of relays per step, since 
even use of the highest transmit power $P_M$ at each node will not satisfy the per-step outage constraint.

\begin{definition}
Let us denote the optimal mean power per step for problem 
(\ref{eqn:constrained_problem_average_cost_with_outage_cost}) by $\gamma^*$, for a given 
$(\overline{q},\overline{N})$. The set 
$\mathcal{K}(\overline{q},\overline{N})$ is defined as follows:

\footnotesize
\begin{eqnarray*}
&& \mathcal{K}(\overline{q},\overline{N}) := \bigg\{(\lambda^*(\xi_{out},\xi_{relay}),\xi_{out},\xi_{relay}): \\
&& \frac{\overline{\Gamma}^*(\xi_{out},\xi_{relay})}{\overline{U}^*(\xi_{out},\xi_{relay})}=\gamma^* , 
\frac{ \overline{Q}_{out}^*(\xi_{out},\xi_{relay}) }{ \overline{U}^*(\xi_{out},\xi_{relay}) } \leq \overline{q} \\
&& \frac{1}{ \overline{U}^*(\xi_{out},\xi_{relay})} \leq \overline{N}, 
\xi_{out} \geq 0, \xi_{relay} \geq 0 \bigg\}
\end{eqnarray*}
\normalsize

where  the optimal average cost per step of the 
unconstrained problem (\ref{eqn:unconstrained_problem_average_cost_with_outage_cost}) under 
OptExploreLim is $\lambda^*(\xi_{out},\xi_{relay})$.\qed 
\end{definition}

$\mathcal{K}(\overline{q},\overline{N})$ can possibly be empty (in case 
$(\overline{q},\overline{N})$ is not a feasible pair). 
{\em Hence, we make the following assumption which ensures the non-emptiness of $\mathcal{K}(\overline{q},\overline{N})$.} 

\begin{assumption}\label{assumption:existence_of_xio_xir}
 The constraint parameters $\overline{q}$ and $\overline{N}$ in 
(\ref{eqn:constrained_problem_average_cost_with_outage_cost}) are such that there exists 
at least one pair $ \xi_{out}^* \geq 0, \xi_{relay}^* \geq 0$ for which  
$(\lambda^*(\xi_{out}^*,\xi_{relay}^*),\xi_{out}^*,\xi_{relay}^*) \in \mathcal{K}(\overline{q},\overline{N})$.\qed
\end{assumption}

{\em Remark:} Assumption~\ref{assumption:existence_of_xio_xir} implies that the constraints are consistent 
(in terms of achievability). If $\xi_{out}^*>0, \xi_{relay}^*>0$, it would imply that both of the constraints   
are active. If $\xi_{out}^*=0$, it would imply that we can keep the mean outage per step strictly 
less than $\overline{q}$ by using the minimum available power at each node, while meeting the constraint on the 
relay placement rate. The optimal policy in Algorithm~\ref{algorithm:policy_structure_smdp_backtracking},  
under $\xi_{out}=0$, will place relays with inter-relay distance 
$(A+B)$ steps, and use the minimum available power level at each node. 
$\xi_{out}^*=\infty$ implies that the outage constraint 
cannot be met even with the highest power level at each node, under the relay placement rate constraint. Similar 
arguments apply to $\xi_{relay}^*$. \qed

We now establish some structural properties of $\mathcal{K}(\overline{q},\overline{N})$.

\begin{theorem}\label{theorem:structure_of_mathcal_K_q_n}
 If $\mathcal{K}(\overline{q},\overline{N})$ is non-empty, then:
 \begin{itemize}
  \item Suppose that there exists $\xi_{out}^*>0$, $\xi_{relay}^*>0$ such that the policy 
  $\pi^*(\xi_{out}^*,\xi_{relay}^*)$ satisfies both constraints in 
  (\ref{eqn:constrained_problem_average_cost_with_outage_cost}) with equality. Then, there does not 
  exist $\xi_{out}'\geq 0,\xi_{relay}' \geq 0$ satisfying (i) 
  $(\lambda^*(\xi_{out}',\xi_{relay}'),\xi_{out}',\xi_{relay}') \in \mathcal{K}(\overline{q},\overline{N})$,  
  and (ii) $\frac{\overline{Q}_{out}^*(\xi_{out}',\xi_{relay}')}{\overline{U}^*(\xi_{out}',\xi_{relay}')} < \overline{q}$ or 
  $\frac{1}{\overline{U}^*(\xi_{out}',\xi_{relay}')} < \overline{N}$.
  \item If there exists a $\xi_{relay}' \geq 0$ such that 
  $(\lambda^*(0,\xi_{relay}'),0,\xi_{relay}') \in \mathcal{K}(\overline{q},\overline{N})$, then, $\forall \xi_{relay} \geq 0$, 
  we have $(\lambda^*(0,\xi_{relay}),0,\xi_{relay}) \in \mathcal{K}(\overline{q},\overline{N})$.\qed
 \end{itemize}
\end{theorem}
\begin{proof}
 See Appendix~\ref{appendix:learning_backtracking_adaptive_with_outage_cost}, 
 Section~\ref{subsection:proof_of_structure_of_mathcal_K_q_n}.
\end{proof}

\begin{assumption}\label{assumption:shadowing_continuous_random_variable}
 The shadowing random variable $W$ has a continuous probability density 
function (p.d.f.) over $(0,\infty)$; for any $w \in (0,\infty)$, 
$\mathbb{P}(W=w)=0$. 
One example could be log-normal shadowing.\qed
\end{assumption}

\begin{theorem}\label{theorem:placement_rate_mean_outage_per_step_continuous_in_xio_and_xir}
Suppose that Assumption~\ref{assumption:shadowing_continuous_random_variable} holds. Under the 
OptExploreLim algorithm,  
the optimal mean power per step $\frac{\overline{\Gamma}^*(\xi_{out},\xi_{relay})}{\overline{U}^*(\xi_{out},\xi_{relay})}$, 
the optimal mean number of relays per step $\frac{1}{\overline{U}^*(\xi_{out},\xi_{relay})}$ 
and the optimal mean outage per step $\frac{\overline{Q}_{out}^*(\xi_{out},\xi_{relay})}{\overline{U}^*(\xi_{out},\xi_{relay})}$, 
are continuous in $\xi_{out}$ and $\xi_{relay}$.
\end{theorem}
\begin{proof}
  See Appendix~\ref{appendix:learning_backtracking_adaptive_with_outage_cost}, 
  Section~\ref{subsection:proof_of_placement_rate_mean_outage_per_step_continuous_in_xio_and_xir}.
\end{proof}

{\em Remark:} Note that, by Theorem~\ref{theorem:placement_rate_mean_outage_per_step_continuous_in_xio_and_xir}, 
we need not do any randomization (see \cite{ma-makowski88steering-policies-recurrence-condition} for 
reference) among deterministic policies in order to meet the constraints with equality.

\vspace{-1mm}
\subsection{OptExploreLimAdaptiveLearning Algorithm}\label{subsection:learning_backtracking_adaptive_with_outage_cost_algorithm}
\vspace{-1mm}

\begin{algorithm}\label{algorithm:learning_backtracking_adaptive_with_outage_cost_algorithm}
  This algorithm iteratively updates $\lambda^{(k)}, \xi_{out}^{(k)}, \xi_{relay}^{(k)}$ after each relay is 
placed. Let $(\lambda^{(k)}, \xi_{out}^{(k)}, \xi_{relay}^{(k)})$ be the iterates  
after placing the $k$-th relay (the sink is called node~$0$), and 
let $(\lambda^{(0)}, \xi_{out}^{(0)}, \xi_{relay}^{(0)})$ be the initial estimates. 
In the process of deploying the $k$-th relay, if the 
shadowing (which is measured indirectly only via $Q_{out}(u,\gamma,w_u)$ for 
$A+1 \leq u \leq A+B$ and $\gamma \in \mathcal{S}$) is $\underline{w}=\{w_{A+1},\cdots,w_{A+B}\}$, 
then place the $k$-th relay according to the following policy:

\vspace{-2mm}

\begin{footnotesize}
\begin{eqnarray}
(u_k,\gamma_k)= \argmin_{u, \gamma} \bigg( \gamma + \xi_{out}^{(k-1)} Q_{out}(u,\gamma,w_u) + \xi_{relay}^{(k-1)} 
-\lambda^{(k-1)} u \bigg)\label{eqn:learning_backtracking_adaptive_with_outage_cost_decision_part}
\end{eqnarray}
\end{footnotesize}

\vspace{-2mm}

After placing the $k$-th relay, let us denote the transmit power, distance (in steps) and outage probability   
from relay $k$ to relay $(k-1)$ by $\gamma_k$, $u_k$ and $Q_{out}(u_k,\gamma_k, w_{u_k})$. 
After placing the $k$-th relay, make the following updates (using the measurements made in the 
process of placing the $k$-th relay): 

\vspace{-2mm}

\begin{footnotesize}
\begin{eqnarray}
 \lambda^{(k)} &=&   \lambda^{(k-1)}+ a_k \min_{u, \gamma } \bigg( \gamma + \xi_{out}^{(k-1)} Q_{out}(u,\gamma,w_u)  \nonumber\\
&& + \xi_{relay}^{(k-1)}-\lambda^{(k-1)} u \bigg) \nonumber\\
\xi_{out}^{(k)}&=& \Lambda_{[0,A_2]}\bigg(\xi_{out}^{(k-1)}+ b_k (Q_{out}(u_k,\gamma_k, w_{u_k})-\overline{q}u_k) \bigg)\nonumber\\
\xi_{relay}^{(k)}&=& \Lambda_{[0,A_3]}\bigg(\xi_{relay}^{(k-1)} + b_k (1-\overline{N}u_k) \bigg)
\label{eqn:learning_backtracking_adaptive_with_outage_cost_update_part}
\end{eqnarray}
\end{footnotesize}

\vspace{-2mm}

where $\Lambda_{[0,A_2]}(x)$ denotes the projection of $x$ on the interval $[0,A_2]$. 
$A_2$ and $A_3$ need to be chosen carefully; the reason is explained in the discussion later in this section 
(along with a brief discussion on how $A_2$ and $A_3$ have to be chosen).

$\{a_k\}_{k \geq 1}$ and $\{b_k\}_{k \geq 1}$ are two decreasing sequences such that $a_k, b_k >0, \forall k \geq 1$, 
$\sum_k a_k=\infty$, $\sum_k a_k^2< \infty$, $\sum_k b_k=\infty$, $\sum_k b_k^2< \infty$ 
and $\lim_{k \rightarrow \infty}\frac{b_k}{a_k}=0$. In particular, we can use $a_k=C_1 k^{-n_1}$ and 
$b_k=C_2 k^{-n_2}$ where $C_1>0$, $C_2>0$, $\frac{1}{2} <n_1 < n_2 \leq 1$.\qed
\end{algorithm}

Note that, for $(\xi_{out},\xi_{relay}) \in [0,A_2] \times [0,A_3]$, we 
have $0 < \lambda^*(\xi_{out},\xi_{relay}) \leq (P_M+A_2+A_3) $. Let us define the set 
$\hat{\mathcal{K}}(\overline{q},\overline{N}):=\mathcal{K}(\overline{q},\overline{N}) \cap ([0, (P_M+A_2+A_3)] \times [0,A_2] \times [0,A_3])$ which 
is a subset of $\mathcal{K}(\overline{q},\overline{N})$.

\begin{theorem}\label{theorem:convergence_learning_backtracking_adaptive_with_outage_cost_algorithm}
 Under Assumption~\ref{assumption:existence_of_xio_xir}, Assumption~\ref{assumption:shadowing_continuous_random_variable}   
 and under proper choice of 
$A_2$ and $A_3$, 
the iterates $(\lambda^{(k)}, \xi_{out}^{(k)}, \xi_{relay}^{(k)})$ in 
Algorithm~\ref{algorithm:learning_backtracking_adaptive_with_outage_cost_algorithm} converge almost surely to 
$\hat{\mathcal{K}}(\overline{q},\overline{N})$ as $k \rightarrow \infty$.
\end{theorem}
\begin{proof}
 See Appendix~\ref{appendix:learning_backtracking_adaptive_with_outage_cost}, 
 Section~\ref{subsection:convergence_proof_of_optexplorelimadaptive_learning}.
 \end{proof}

{\em Remark:} Algorithm~\ref{algorithm:learning_backtracking_adaptive_with_outage_cost_algorithm} 
induces a nonstationary  
policy. But 
Theorem~\ref{theorem:convergence_learning_backtracking_adaptive_with_outage_cost_algorithm} 
establishes that the sequence of policies generated by 
Algorithm~\ref{algorithm:learning_backtracking_adaptive_with_outage_cost_algorithm} converges to the set 
of optimal stationary, deterministic policies 
(for  problem (\ref{eqn:constrained_problem_average_cost_with_outage_cost})).

{\bf Discussion of Theorem~\ref{theorem:convergence_learning_backtracking_adaptive_with_outage_cost_algorithm}:}
\begin{enumerate}[label=(\roman{*})]
\item {\em Two timescales:} The update scheme (\ref{eqn:learning_backtracking_adaptive_with_outage_cost_update_part}) 
can be rewritten as a two-timescale 
stochastic approximation (see \cite{borkar08stochastic-approximation-book}, Chapter~$6$). 
Note that, $\lim_{k \rightarrow \infty}\frac{b_k}{a_k}=0$, i.e., $\xi_{out}$ and $\xi_{relay}$ are adapted in a 
{\em slower} timescale compared to $\lambda$ (which is adapted in the {\em faster} timescale). The dynamics behaves as if 
$\xi_{out}$ and $\xi_{relay}$ are updated simultaneously in a slow outer loop, and, between two successive updates of 
$\xi_{out}$ and $\xi_{relay}$, we update $\lambda$ in an inner loop for a long time. 
Thus, the $\lambda$ update equation views $\xi_{out}$ and $\xi_{relay}$ as quasi-static, 
while the $\xi_{out}$ and $\xi_{relay}$ update 
equations view the $\lambda$ update equation as almost equilibrated. 

\item {\em Structure of the iteration:} Note that, $(Q_{out}(u_k,\gamma_k, w_{u_k})-\overline{q}u_k)$ is the excess outage compared 
to the allowed outage $\overline{q}u_k$ for the $k$-th link. If this quantity is positive (resp., negative), the algorithm 
increases (resp., decreases) $\xi_{out}$ in order to reduce (resp., increase) the outage probability in subsequent steps. 
Similarly, if $u_k<\frac{1}{\overline{N}}$, the algorithm increases $\xi_{relay}$ in order to reduce the relay placement rate. 
The goal is to ensure  
$\lim_{k \rightarrow \infty} ( \overline{Q}_{out}^*(\xi_{out}^{(k)}, \xi_{relay}^{(k)})-\overline{q}\overline{U}^*(\xi_{out}^{(k)}, \xi_{relay}^{(k)}) )=0$ 
and $\lim_{k \rightarrow \infty} ( 1-\overline{N}\overline{U}^*(\xi_{out}^{(k)}, \xi_{relay}^{(k)}) )=0$. 
In the faster timescale, our aim is to ensure that  
$\lim_{k \rightarrow \infty} \mathbb{E}_{\underline{W}}\min_{u,\gamma}(\gamma+\xi_{out}^{(k)}Q_{out}(u,\gamma,W_u)+\xi_{relay}^{(k)}-\lambda^{(k)}u)=0$.

\item {\em Outline of the proof:} The proof proceeds in five steps.

We first prove the almost sure boundedness of $\{ \lambda^{(k)} \}_{k \geq 1}$.

Next, we prove 
that the difference between the sequences $\lambda^{(k)}$ and $\lambda^*(\xi_{out}^{(k)},\xi_{relay}^{(k)})$ 
converges to $0$ almost surely; this will  prove the desired convergence in the faster timescale. 
This result has been proved using the 
theory in \cite[Chapter~$6$]{borkar08stochastic-approximation-book}  and 
Theorem~\ref{theorem:learning_backtracking_given_xio_xir_general_step_size}.

{\em In order to ensure   boundedness of the slower timescale iterates, we have used the projection 
operation in the slower timescale. We pose the slower timescale iteration in the same form as a projected stochastic approximation iteration 
(see \cite[Equation~$5.3.1$]{kushner-clark78SA-constrained-unconstrained}).}

In order to prove the desired convergence of the projected 
stochastic approximation, 
we show that our iteration satisfies certain conditions 
given in \cite{kushner-clark78SA-constrained-unconstrained} (see   
\cite[Theorem~$5.3.1$]{kushner-clark78SA-constrained-unconstrained}).

Next, we argue  
(using Theorem~$5.3.1$ of \cite{kushner-clark78SA-constrained-unconstrained}) 
that the slower timescale iterates converge to the set of stationary points of a suitable ordinary differential 
equation (o.d.e.). But, in general, a stationary point on the boundary of the 
closed set $[0,A_2] \times [0,A_3]$ in the $(\xi_{out},\xi_{relay})$ plane may not correspond to a point 
in $\mathcal{K}(\overline{q},\overline{N})$. {\em Hence, we will need to ensure that if $(\xi_{out}',\xi_{relay}')$ 
is a stationary point of the o.d.e., then 
$(\lambda^*(\xi_{out}',\xi_{relay}'),\xi_{out}',\xi_{relay}') \in \mathcal{K}(\overline{q},\overline{N})$. 
In order to ensure this, we need to choose $A_2$ and $A_3$ properly.  
The choice of $A_2$ and $A_3$ is rather technical, and is 
explained in detail in 
Appendix~\ref{appendix:learning_backtracking_adaptive_with_outage_cost}, 
Section~\ref{subsubsection:finishing_two_timescale_adaptive_learning_with_outage_convergence_proof}.  
Here we will just provide the method of choosing $A_2$ and $A_3$, without 
any explanation of why they should be chosen in this way.}  
The number $A_2$ has to be chosen so large 
that under $\xi_{out}=A_2$ and for all $A+1 \leq u \leq A+B$, we will have 
$\mathbb{P}(\argmin_{\gamma \in \mathcal{S}}(\gamma+A_2 Q_{out}(u,\gamma,W))=P_M)>1-\kappa$ 
for some small enough $\kappa>0$. We must also have  
$\frac{\overline{Q}_{out}^*(A_2,0)}{\overline{U}^*(A_2,0)} \leq \overline{q}$. 
The number $A_3$ has to be chosen so large that 
for any $\xi_{out} \in [0,A_2]$, we will have 
$\overline{U}^*(\xi_{out},A_3) > \frac{1}{\overline{N}}$ (provided that $\frac{1}{\overline{N}}<A+B$). 
The numbers $A_2$ and $A_3$ have to be chosen so large that there exists at least one 
$(\xi_{out}',\xi_{relay}') \in [0,A_2] \times [0,A_3]$ such that 
$(\lambda^*(\xi_{out}',\xi_{relay}'),\xi_{out}',\xi_{relay}') \in \mathcal{K}(\overline{q},\overline{N})$.

\item {\em Asymptotic behaviour of the iterates:}
If the pair $(\overline{q},\overline{N})$ is such that one can be met with strict inequality and the other 
can be met with equality while using the optimal mean power per step for this pair $(\overline{q},\overline{N})$, 
then one Lagrange multiplier will converge to $0$. This will happen if 
$\overline{q}> \frac{\mathbb{E}_W Q_{out}(A+B,P_1,W)}{A+B}$;  we 
will have $\xi_{out}^{(k)} \rightarrow 0$ (obvious from OptExploreLim with $\xi_{out}=0$) in this case. Here 
we will place all the relays at the $(A+B)$-th step and use the smallest power level at each node. 
On the other hand, if the constraints are not feasible, then either $\xi_{out}^{(k)} \rightarrow A_2$ 
or $\xi_{relay}^{(k)} \rightarrow A_3$ (since convergence to $\infty$ is not possible 
due to projection) or both will happen.

{\em  $\mathcal{K}(\overline{q},\overline{N})$ may have multiple tuples.  
But simulation results show that it has only one tuple in case it 
is nonempty.}\qed
\end{enumerate}

\vspace{-4mm}
\subsection{Asymptotic Performance of Algorithm~\ref{algorithm:learning_backtracking_adaptive_with_outage_cost_algorithm}}
\label{subsection:asymptotic_performance_adaptive_learning}
\vspace{-2mm}

Let us denote by $\pi_{oelal}$ the (nonstationary) deployment policy induced by 
Algorithm~\ref{algorithm:learning_backtracking_adaptive_with_outage_cost_algorithm}. 
We will now show that $\pi_{oelal}$ is an optimal policy for the constrained problem 
(\ref{eqn:constrained_problem_average_cost_with_outage_cost}).

\begin{theorem}\label{theorem:expected_average_cost_performance_of_optexplorelimadaptivelearning}
 Suppose that Assumption~{\ref{assumption:existence_of_xio_xir}} 
 and Assumption~\ref{assumption:shadowing_continuous_random_variable} hold. Then, under proper choice of $A_2$ and $A_3$, 
 the policy $\pi_{oelal}$ solves the problem (\ref{eqn:constrained_problem_average_cost_with_outage_cost}); i.e., we have:

\footnotesize
\begin{eqnarray*}
&& \limsup_{x \rightarrow \infty} \frac{\mathbb{E}_{\pi_{oelal}}\sum_{i=1}^{N_x}\Gamma_i}{x} = \gamma^* \nonumber\\
&& \, \limsup_{x \rightarrow \infty} \frac{\mathbb{E}_{\pi_{oelal}}\sum_{i=1}^{N_x}Q_{out}^{(i,i-1)}}{x} \leq \overline{q}, \,\,  
\limsup_{x \rightarrow \infty} \frac{\mathbb{E}_{\pi_{oelal}}N_x}{x} \leq \overline{N} \nonumber\\
\end{eqnarray*}
\normalsize
\end{theorem}

\begin{proof}
 See Appendix~\ref{appendix:learning_backtracking_adaptive_with_outage_cost}, 
Section~\ref{subsection:proof_of_expected_average_cost_performance_of_optexplorelimadaptivelearning}.
\end{proof}

\vspace{-2mm}
\section{Convergence Speed of Learning Algorithms: A Simulation Study}\label{section:convergence_speed_learning_algorithms}
\vspace{-2mm}
In this section, we provide a simulation study to demonstrate the convergence rate of 
Algorithm~\ref{algorithm:learning_backtracking_given_xio_xir_special_step_size} 
and  Algorithm~\ref{algorithm:learning_backtracking_adaptive_with_outage_cost_algorithm}. 
The simulations are provided for  $\eta=4.7$, $\sigma=7.7$~dB, $\delta=20$~m, 
$A=0$, $B=5$, 
$c=10^{0.17}$, $P_{rcv-min}=-97$~dBm, $\mathcal{S}=\{-18,-7,-4,0,5\}$~dBm 
(see Section~\ref{sec:system_model_and_notation} for   notation and Section~\ref{subsection:parameter_values} 
for   parameter values).

\vspace{-4mm}
\subsection{OptExploreLimLearning for Given $\xi_{out}$ and $\xi_{relay}$}
\label{subsection:convergence_speed_optexplorelimlearning}
\vspace{-2mm}

Let us choose $\xi_{out}=100$, $\xi_{relay}=1$. We assume that the propagation environment 
in which we are deploying is characterized 
by the parameters as in Section~\ref{subsection:parameter_values} (e.g., 
$\eta=4.7$, $\sigma=7.7$~dB). The optimal average cost per step, under these parameter values, is
$0.8312$ (computed numerically). 

On the other hand, for $\eta = 4$, $\sigma = 7$~dB, $\xi_{out}=100$ and $\xi_{relay}=1$, 
the optimal average cost per step   is $0.4577$, and it is $1.7667$ for $\eta = 5.5$,
$\sigma = 9$~dB. These two cases correspond to two different imperfect estimates of $\eta$ and $\sigma$ available to the agent 
before deployment starts.

 Suppose
that the actual $\eta = 4.7$, $\sigma = 7.7$~dB, but at the time of
deployment we have an initial estimate that $\eta = 4$, $\sigma = 7$~dB;
thus, we start with $\lambda^{(0)}=0.4577$. After placing the $k$-th relay,
the actual average cost per step of the relay network connecting the $k$-th relay to the sink is $\lambda^{(k)}$;
this quantity is a random variable whose realization depends
on the shadowing realizations over the links measured in the
process of deployment up to the $k$-th relay. We ran $10000$ simulations of 
Algorithm~\ref{algorithm:learning_backtracking_given_xio_xir_special_step_size},  
 starting with different seeds for the shadowing random process, and estimating 
$\mathbb{E} (\lambda^{(k)})$ as the average of the samples of $\lambda^{(k)}$ over these $10000$ simulations. We also do the same for 
$\lambda^{(0)}=1.7667$ (optimal cost for $\eta = 5.5$,
$\sigma = 9$~dB).

The estimates of $\mathbb{E} (\lambda^{(k)}), k \geq 1$ as a function of $k$, for the two initial values of $\lambda^{(0)}$, 
are shown in Figure~\ref{fig:optexplorelimlearning_given_xio_xir}. 
Also shown, in Figure~\ref{fig:optexplorelimlearning_given_xio_xir}, is the optimal value $\lambda^* = 0.8312$ for 
the true propagation parameters (i.e., $\eta = 4.7$, $\sigma = 7.7$~dB).  
From Figure~\ref{fig:optexplorelimlearning_given_xio_xir}, 
we observe that $\mathbb{E} (\lambda^{(k)})$ approaches the optimal cost $0.8312$ for the actual propagation 
parameters, as the number of deployed relays increases, and gets to within $10\%$ of the optimal 
cost by the time that $4$ or $5$ relays are placed, starting with two widely different initial guesses of 
the propagation parameters. Thus, OptExploreLimLearning could be useful even when the distance 
 can be covered by only $4$ to $5$ relays.

Note that, each simulation yields one sample path of the deployment process. 
We obtained the estimates of $\mathbb{E} (\lambda^{(k)})$ as a function of $k$ (by 
averaging over $10000$ sample paths); the 
convergence speed will vary across sample paths   even though $\lambda^{(k)} \rightarrow 0.8312$ almost surely as 
$k \rightarrow \infty$.

\begin{figure}[!t]
\centering
\includegraphics[width=9cm, height=3cm]{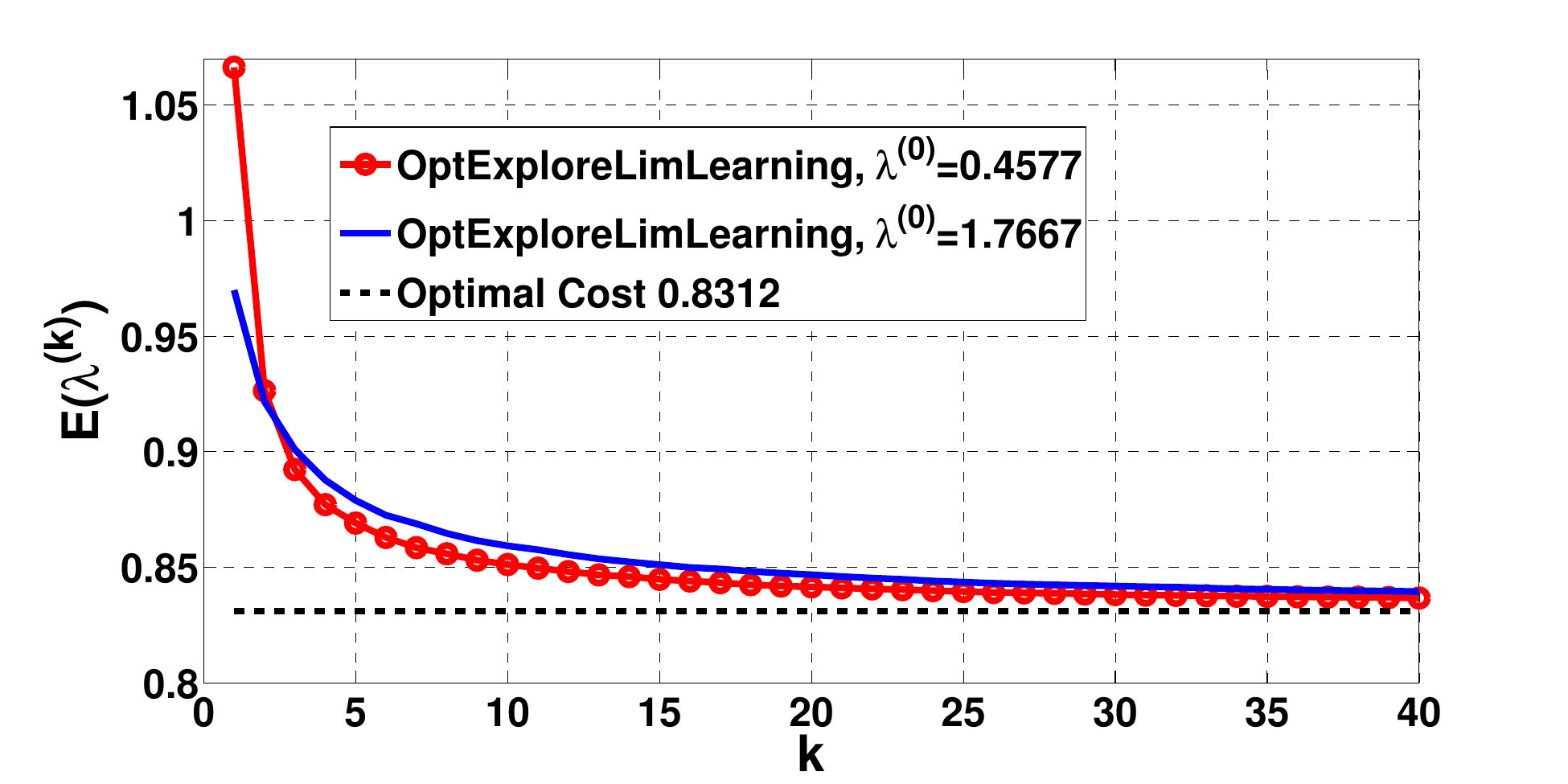}
\vspace{-7mm}
\caption{Demonstration of the convergence of OptExploreLimLearning 
(Algorithm~\ref{algorithm:learning_backtracking_given_xio_xir_special_step_size}) as deployment progresses. 
$\lambda^{(0)}$ has not been included here.}
\label{fig:optexplorelimlearning_given_xio_xir}
\vspace{-7mm}
\end{figure}

\begin{figure*}[t]
\begin{minipage}[r]{0.33\linewidth}
\subfigure{
\includegraphics[width=\linewidth, height=2.4cm]{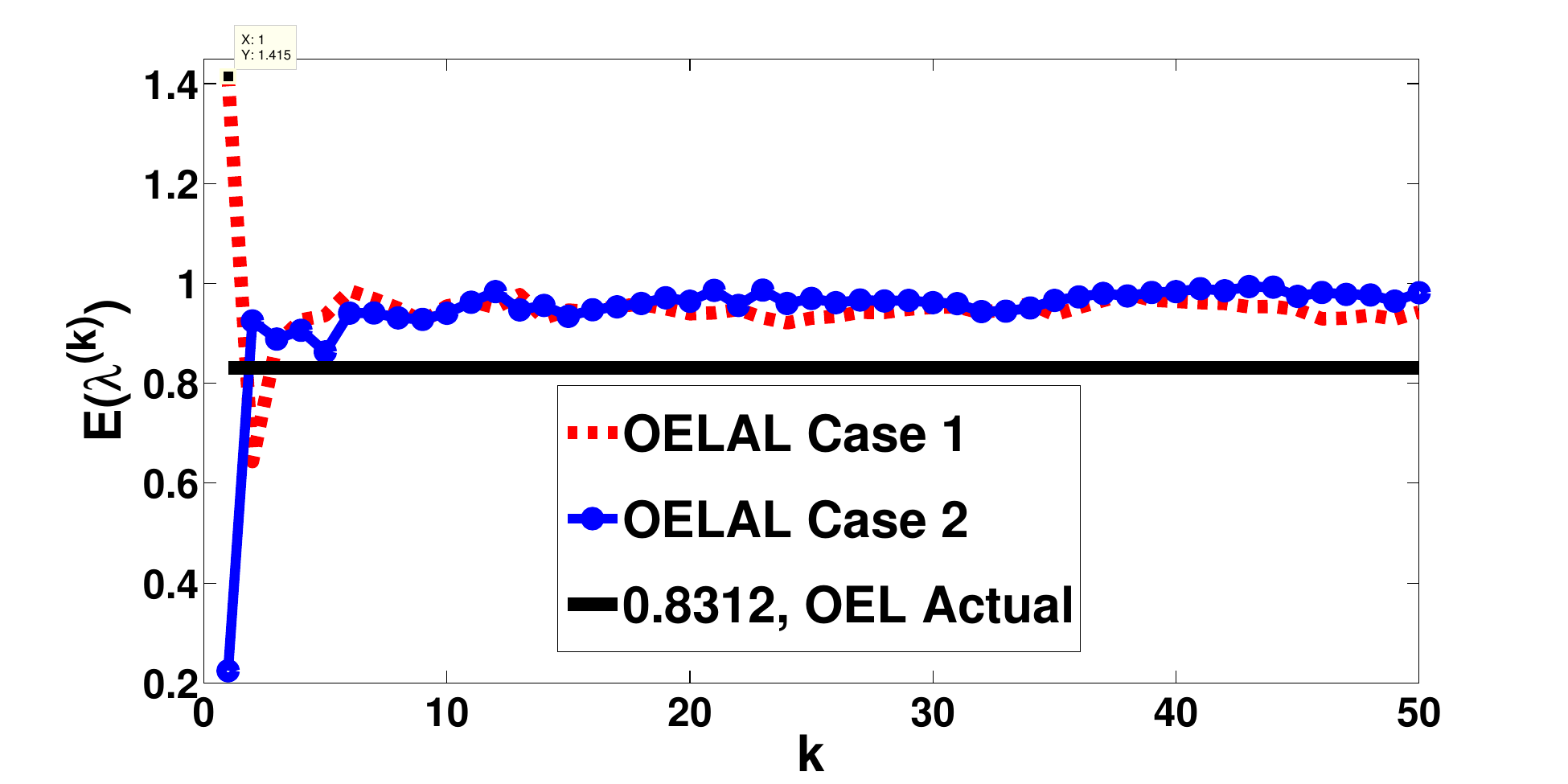}
\includegraphics[width=\linewidth, height=2.4cm]{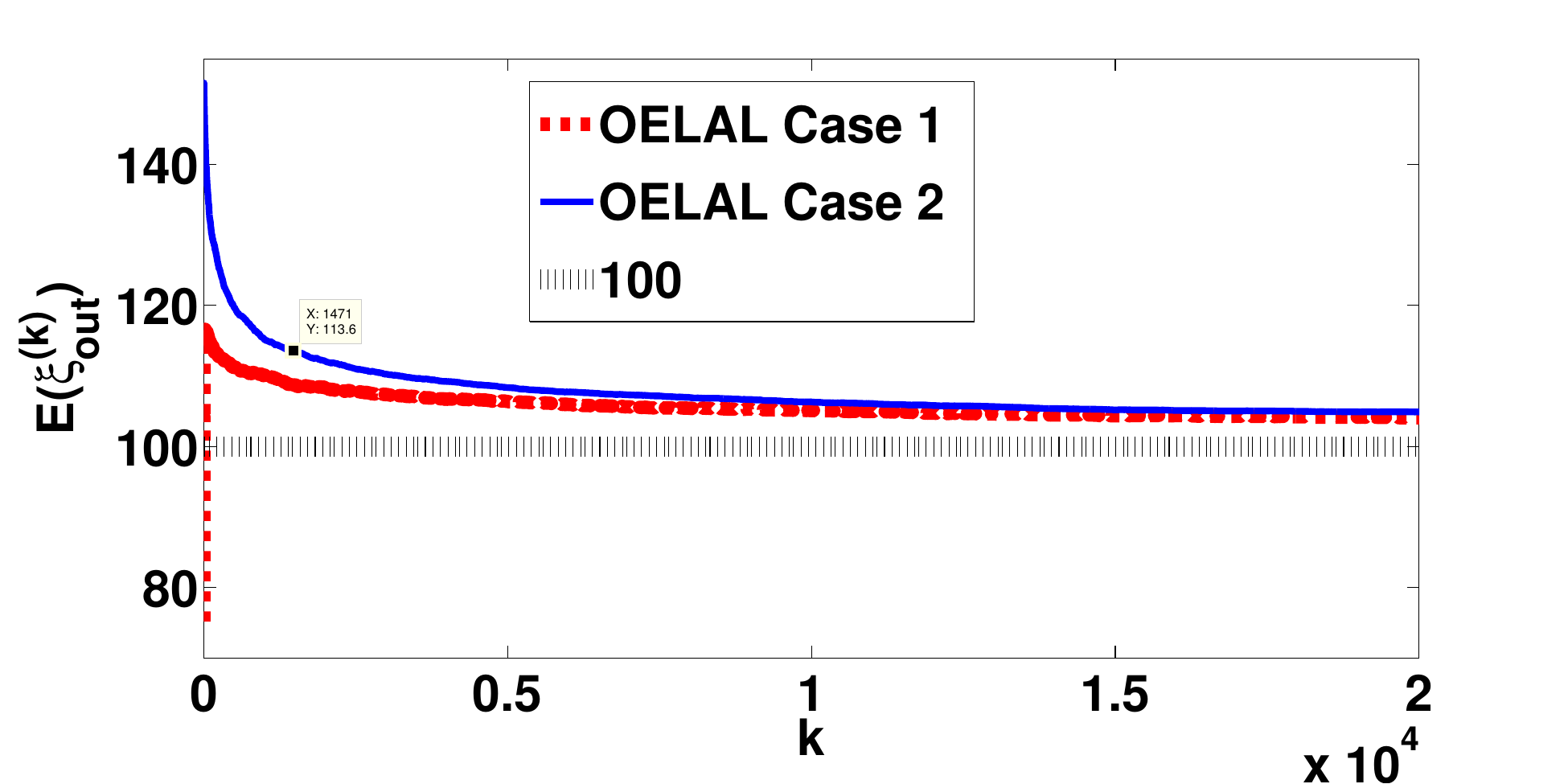}
\includegraphics[width=\linewidth, height=2.4cm]{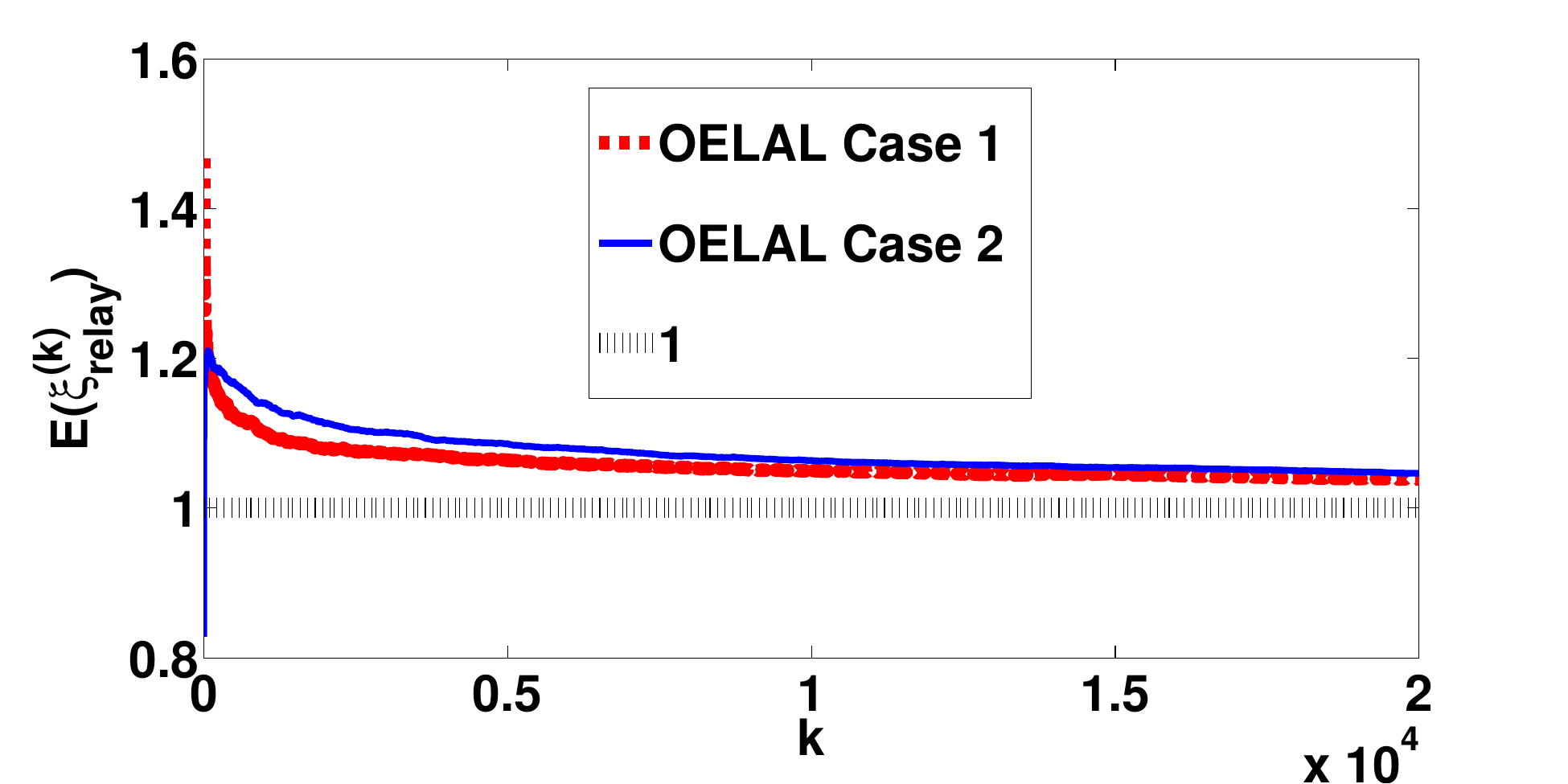}}
\end{minipage}  \hfill
\vspace{-20mm}
\end{figure*}
\begin{figure*}[t]
\begin{minipage}[c]{0.33\linewidth}
\subfigure{
\includegraphics[width=\linewidth, height=2.4cm]{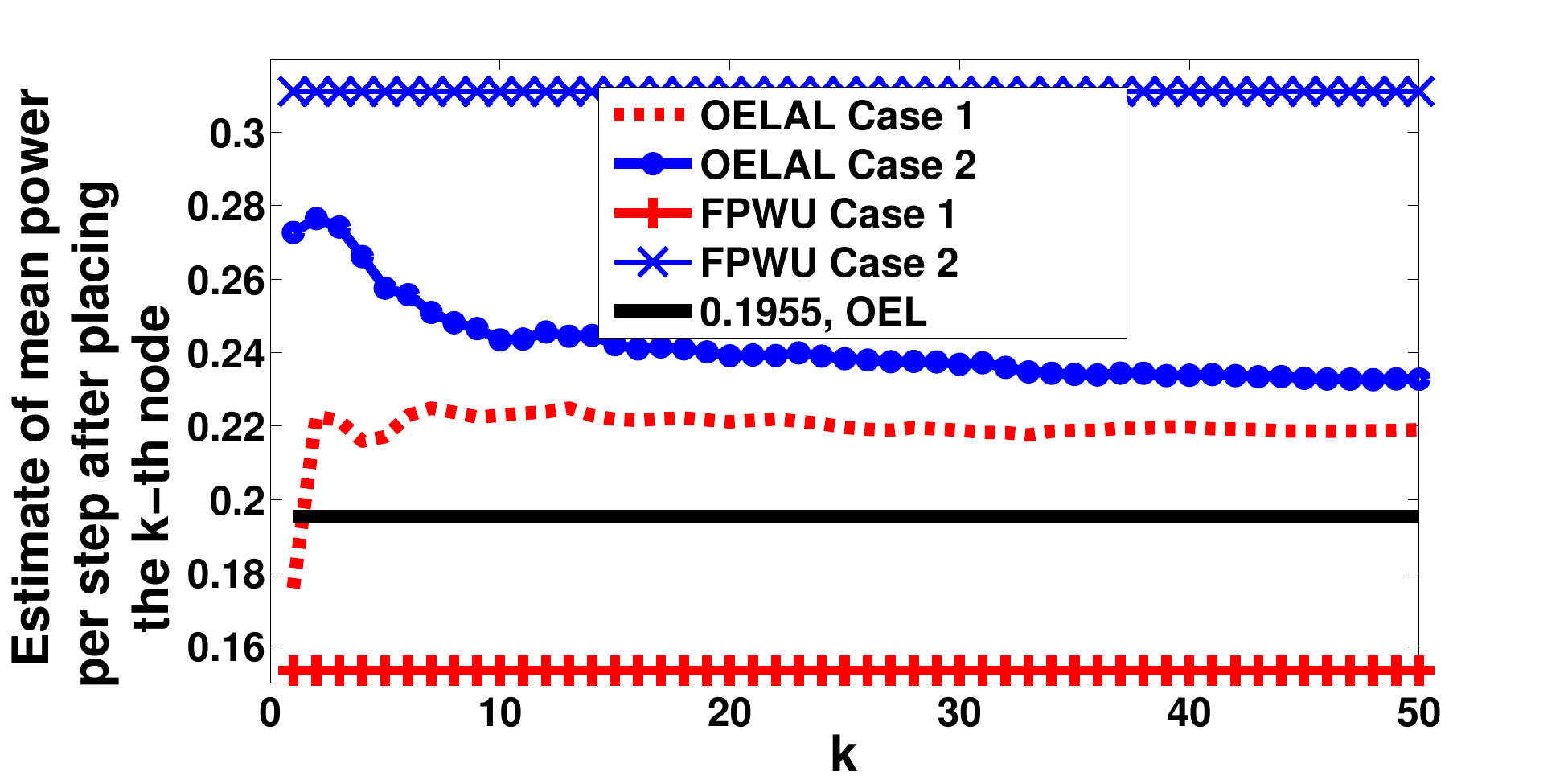}
\includegraphics[width=\linewidth, height=2.4cm]{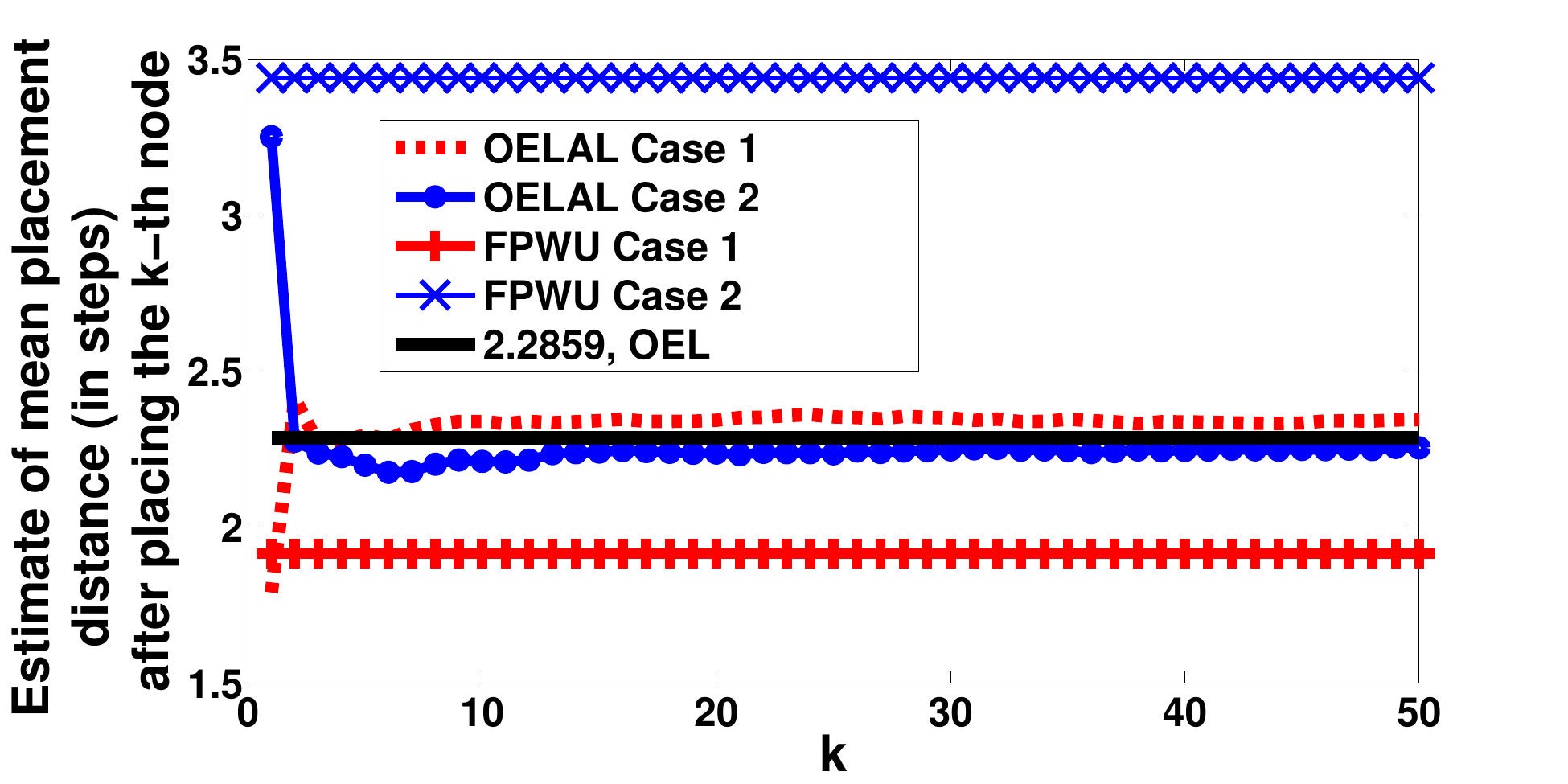}
\includegraphics[width=\linewidth, height=2.4cm]{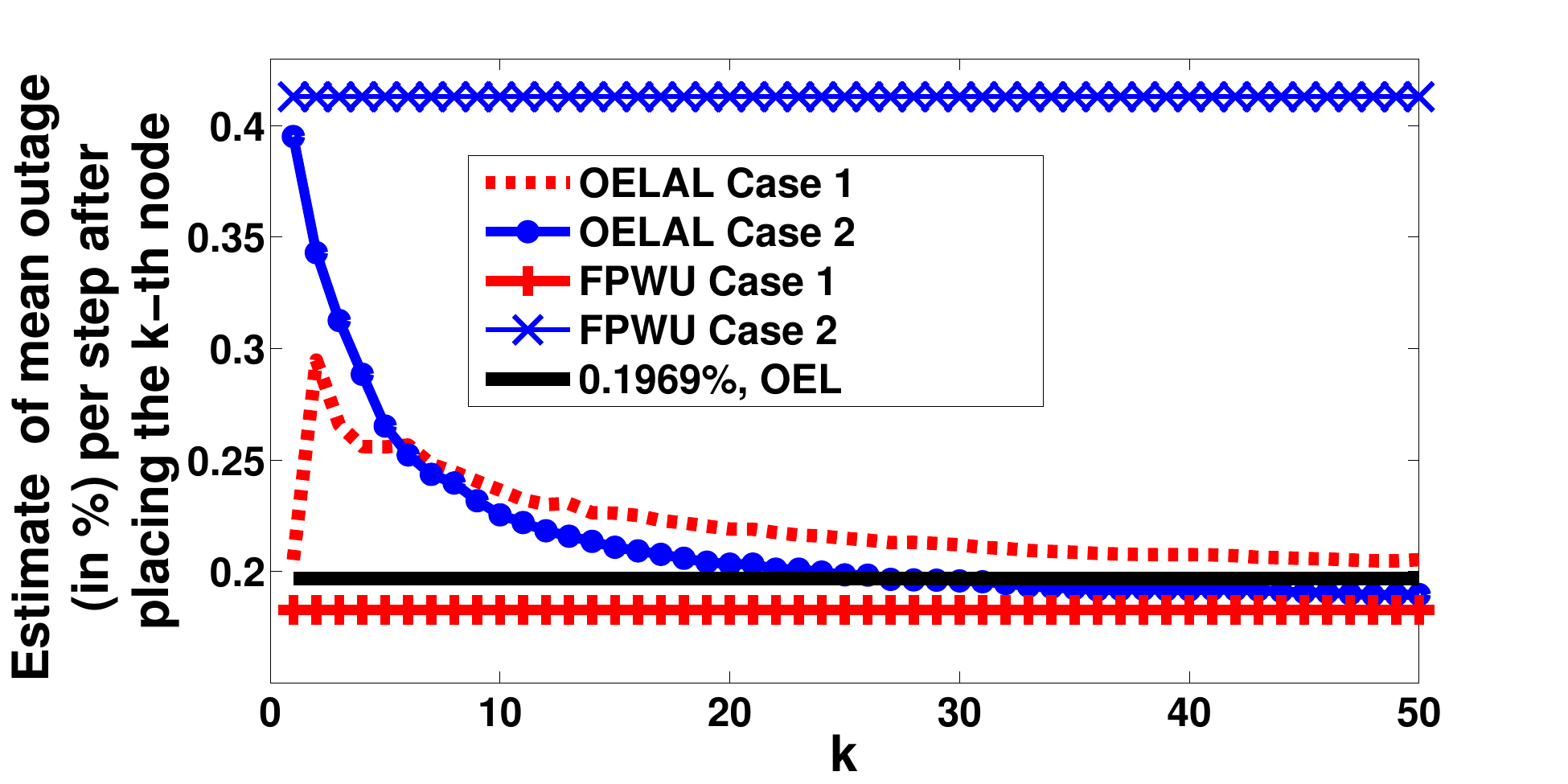}}
\end{minipage} \hfill
\vspace{-15mm}
\caption{Demonstration of the convergence of OptExploreLimAdaptiveLearning as deployment progresses. 
In the legends, ``OEL'' refers to the values that are obtained if OptExploreLim is used; 
these are the target values for OptExploreLimAdaptiveLearning. Note that, we have used 
line styles for $\xi_{out}$ and $\xi_{relay}$ updates, that are different from the line styles of other four 
plots. Also note that, outage probabilities are shown in percentage and not in decimal.}
\label{fig:two_timescale_plots}
\vspace{-7mm}
\end{figure*}

\vspace{-4mm}
\subsection{OptExploreLimAdaptiveLearning}
\label{subsection:convergence_speed_optexplorelimadaptivelearning}
\vspace{-1mm}

In this section, we will discuss how OptExploreLimAdaptiveLearning 
(Algorithm~\ref{algorithm:learning_backtracking_adaptive_with_outage_cost_algorithm}) 
performs for deployment over a 
finite distance under an unknown propagation environment. 
We assume that the true propagation parameters are given in 
Section~\ref{subsection:parameter_values} (e.g., $\eta=4.7$, $\sigma=7.7$~dB). 
If we know the true propagation environment, then, under the choice $\xi_{relay}=1$ and $\xi_{out}=100$, 
the optimal average cost per step will be $0.8312$, and this can be achieved by 
OptExploreLim (Algorithm~\ref{algorithm:policy_structure_smdp_backtracking}). The corresponding mean outage per step 
will be $\frac{0.0045}{2.2859}=0.001969 $ (i.e., $0.1969 \%$) and the mean number of relays per step 
will be $\frac{1}{2.2859}$.

Now, suppose that we wish to solve the constrained problem 
in (\ref{eqn:constrained_problem_average_cost_with_outage_cost}) with 
the targets $\overline{q}=0.001969$ (i.e., $0.1969 \%$) and $\overline{N}=\frac{1}{2.2859}$, 
but we do not know the true propagation 
environment. Hence, the deployment will use 
OptExploreLimAdaptiveLearning 
with some  choice of $\xi_{out}^{(0)}$, $\xi_{relay}^{(0)}$ and $\lambda^{(0)}$. 

We seek to compare among the following three scenarios: (i) $\eta$ and $\sigma$ are 
completely known (we use OptExploreLim with $\xi_{relay}=1$ and $\xi_{out}=100$ in this case), 
(ii) imperfect estimates of $\eta$ and $\sigma$ are available prior to deployment, and OptExploreLimAdaptiveLearning is 
used to learn the optimal policy, and (iii) imperfect estimates of $\eta$ and $\sigma$ are available prior to deployment, but 
a corresponding suboptimal policy is used throughout the deployment without any update. For 
convenience in writing, we introduce the abbreviations OELAL and 
OEL for OptExploreLimAdaptiveLearning and OptExploreLim, respectively. We also use the abbreviation FPWU for 
``Fixed Policy without Update.'' Now, we formally 
introduce the following cases that we consider in our simulations:

\begin{enumerate}[label=(\roman{*})]

{\bf \item OEL:} Here we know 
$\eta=4.7$, $\sigma=7.7$~dB, and use  OptExploreLim (Algorithm~\ref{algorithm:policy_structure_smdp_backtracking}) 
with $\xi_{out}=100$, $\xi_{relay}=1$, $\lambda^*=0.8312$. 
OEL will meet both the constraints with equality, and will minimize 
the mean power per step.

 {\bf \item  OELAL Case~$1$:} OELAL Case~$1$ is the case where the true $\eta$ and $\sigma$ (which are unknown to 
the deployment agent)  are specified by Section~\ref{subsection:parameter_values}, but 
we use OptExploreLimAdaptiveLearning with $\xi_{out}^{(0)}=75$, $\xi_{relay}^{(0)}=1.25$ and 
$\lambda^{(0)}=0.5007$, in order to meet the constraints specified earlier in this subsection.  
Note that, under $\xi_{out}=75$ and $\xi_{relay}=1.25$, 
the optimal mean cost per step is $0.5007$ for $\eta=4$, $\sigma=7$~dB. Hence, we start with a wrong choice of Lagrange 
multipliers, a wrong estimate of $\eta$ and $\sigma$, and an estimate of the optimal average cost per 
step which corresponds to these wrong choices. The goal is to see how fast the variables 
$\lambda^{(k)}$, $\xi_{out}^{(k)}$ and $\xi_{relay}^{(k)}$ converge to the desired target $0.8312$, $100$ and $1$, respectively. 
We also study how close to the desired target values are the quantities 
such as mean power per step, mean outage per step and mean placement distance for the relay network 
between $k$-th relay and the sink node.

 {\bf \item  OELAL Case~$2$:} OELAL Case~$2$ is different from OELAL Case~$1$ only in the aspect that 
$\lambda^{(0)}=1.7679$ is used in OELAL Case~$2$. 
Note that, under $\xi_{out}=75$ and $\xi_{relay}=1.25$, 
the optimal mean cost per step is $1.7679$ for $\eta=5.5$, $\sigma=9$~dB.

{\bf \item FPWU Case~$1$:} In this case, the true $\eta$ and $\sigma$ are unknown to the deployment agent.  
The deployment agent uses $\xi_{out}=75$, $\xi_{relay}=1.25$ and $\lambda^*=0.5007$ throughout the 
deployment process under the algorithm specified by (\ref{eqn:smdp-optimal-policy}). Clearly, he chooses a wrong set of 
Lagrange multipliers $\xi_{out}=75$, $\xi_{relay}=1.25$, and he has a wrong estimate $\eta=4$, $\sigma=7$~dB. 
The optimal average cost per step $\lambda^*=0.5007$ is computed for these wrong choice of parameters, and    
the corresponding suboptimal policy is used throughout the deployment process without any update; this will be used to 
demonstrate the gain in performance by updating the policy under OptExploreLimAdaptiveLearning, w.r.t. the case 
where the suboptimal policy  is used without any online update.

{\bf \item FPWU Case~$2$:} It differs from 
FPWU Case~$1$ only in the aspect that we use $\lambda^*=1.7679$ in FPWU Case~$2$. 
Recall that, under $\xi_{out}=75$ and $\xi_{relay}=1.25$, 
the optimal mean cost per step is $1.7679$ for $\eta=5.5$, $\sigma=9$~dB.
\end{enumerate}

\vspace{-2mm}

For simulation of OELAL, we chose the step sizes as follows. We chose $a_k=\frac{1}{k^{0.55}}$, 
chose $b_k=\frac{10000}{k^{0.8}}$ for the $\xi_{out}$ update and $b_k=\frac{1}{k^{0.8}}$ for the $\xi_{relay}$ update 
(note that, both $\xi_{out}$ and $\xi_{relay}$ are updated in the same timescale). 
We simulated $10000$ independent network deployments (i.e., $10000$ sample paths of the 
deployment process) with OptExploreLimAdaptiveLearning, and 
estimated (by averaging over $10000$ deployments) the expectations of $\lambda^{(k)}$, $\xi_{out}^{(k)}$, $\xi_{relay}^{(k)}$, 
mean power per step $\frac{\mathbb{E}_{\pi_{oelal}}\sum_{i=1}^k \Gamma_i}{\mathbb{E}_{\pi_{oelal}}\sum_{i=1}^k U_i}$, mean outage per step 
$\frac{\mathbb{E}_{\pi_{oelal}}\sum_{i=1}^k Q_{out}^{(i,i-1)}}{\mathbb{E}_{\pi_{oelal}}\sum_{i=1}^k U_i}$ 
and mean placement distance $\frac{\mathbb{E}_{\pi_{oelal}}\sum_{i=1}^k U_i}{k}$, from the sink node 
to the $k$-th placed node. In each simulated network deployment, we placed $20000$~nodes, i.e., $k$ was allowed to go 
up to $20000$. Asymptotically the estimates are supposed to converge to the values provided by 
OEL.

{\bf Observations from the Simulations:} The results of the simulations are summarized in 
Figure~\ref{fig:two_timescale_plots}. We observe that, the estimates of 
the expectations of $\lambda^{(20000)}$, $\xi_{out}^{(20000)}$, $\xi_{relay}^{(20000)}$, 
mean power per step up to the $20000^{th}$ node, mean outage per step up to the $20000^{th}$ node, and 
mean placement distance (in steps) over $20000$ deployed nodes are $0.8551$, $104.0606$, $1.0385$, $0.2005$, $0.2 \%$ 
(i.e., $0.002$) and 
$2.2939$ for the OELAL Case~$1$, whereas 
those quantities are supposed to be equal to $0.8312$, $100$, $1$, $0.1955$, $0.1969\%$ (i.e., 
$0.001969$) and $2.2859$, respectively. 
We found similar results for OELAL Case~$2$ also. Hence, the quantities converge very close to the desired values. 
{\em We have shown convergence only up to $k=50$ deployments in most cases, 
since the convergence rate of the algorithms in the initial phase are most important in 
practice.}

All the quantities except 
expectation of $\xi_{out}^{(k)}$ and $\xi_{relay}^{(k)}$ 
(which are updated in a slower timescale) converge 
reasonably close to the desired values by the time the $50^{th}$ relay is placed, which will cover a distance of roughly  
$2-3.5$~km. distance.

FPWU Case~$1$ and  FPWU Case~$2$ either violate some constraint or 
uses significantly higher per-step power compared to  
OEL. But, by using the OptExploreLimAdaptiveLearning algorithm, 
we can achieve per-step power expenditure close to the optimal  
while (possibly) violating the constraints by small amount; 
even in case the performance of OELAL is not very close to the optimal 
performance, it will be significantly better than the performance under 
FPWU cases (compare OELAL Case~$2$ and 
FPWU Case~$2$ in Figure~\ref{fig:two_timescale_plots}).\qed

The speed of convergence will depend on the choice of the step sizes  
$a_k$ and $b_k$; optimizing the rate of convergence by choosing optimal step sizes is left 
for future endeavours in this direction. Also, note that, the choice of $\xi_{out}^{(0)}$, $\xi_{relay}^{(0)}$ and 
$\lambda^{(0)}$ will have a significant effect on the performance of the network over a finite length; 
the more accurate are the estimates of $\eta$ and $\sigma$, and the better are the initial choice of 
$\xi_{out}^{(0)}$, $\xi_{relay}^{(0)}$ and $\lambda^{(0)}$, the better will be the convergence speed 
of OptExploreLimAdaptiveLearning.

\vspace{-4mm}
\section{Physical Deployment Experiments}\label{section:real_deployment}
 \vspace{-2mm}

For completeness, we briefly summarize experimental results that were reported  in our conference paper \cite{chattopadhyay-etal14deployment-experience}.   
We performed an actual deployment experiment along
a long tree-lined road in our campus (not exactly a straight line, which is the reality in a  forest) with 
iWiSe motes equipped with $9$~dBi antennas. 
We chose $\xi_{out}=100$, $\xi_{relay}=1$, $B=5$~steps, $\delta=50$~meters, and  $\mathcal{S}=\{-7,-4,0,5\}$~dBm. 
We used the 
packet error rate (PER) of a link as a substitute for outage probability; this does not violate the 
assumptions of our formulation. For  $\eta=4$, 
$\sigma=7$~dB, $\xi_{out}=100$, $\xi_{relay}=1$, the optimal average cost per step is $1.0924$
(computed numerically). Taking $\lambda^{(0)}=1.0924$, 
we performed a real deployment experiment  with
OptExploreLimLearning.  The
deployed network (along with  power levels, outage probabilities and link lengths) is shown in Figure~\ref{fig:real_deployment_OptExploreLimLearning}. 
The  sink is denoted by the ``house'' symbol. 
The algorithm placed relays at successive distances of $150$~m, $50$~m, $50$~m, $100$~m, and $150$~m, 
thereby covering $500$~m until the source was placed.  
The two short ($50$~m long) links are created due to  significant path-loss at the turn in the road.  
After deployment, 
we used the last placed node as the source and sent periodic traffic (at 
various rates) from the source to the sink. The end-to-end packet loss probability increases 
with arrival rate (Figure~\ref{fig:real_deployment_OptExploreLimLearning}); this happens due to carrier sense
failures and collisions because of simultaneous transmissions 
from different nodes.  At very low arrival rate, the loss
probability is $0$ (but the sum PER under the lone
packet model is not $0$). This happens since there are 
link level retransmissions and  since the outage
durations are relatively short; in case a packet encounters an outage in a link, 
the retransmission attempts succeed with high probability. 
{\em The results demonstrate that, even though the design was for the lone packet model, the network can carry 
$4$~packets/second (packet size is $127$~bytes) with $P_{loss} \leq 1\%$, which is sufficient for many  applications. Hence, 
network design with the lone packet model assumption is reasonable for those applications.}

\begin{figure}[t!]
\begin{centering}
\includegraphics[height=2.3cm, width=9cm]{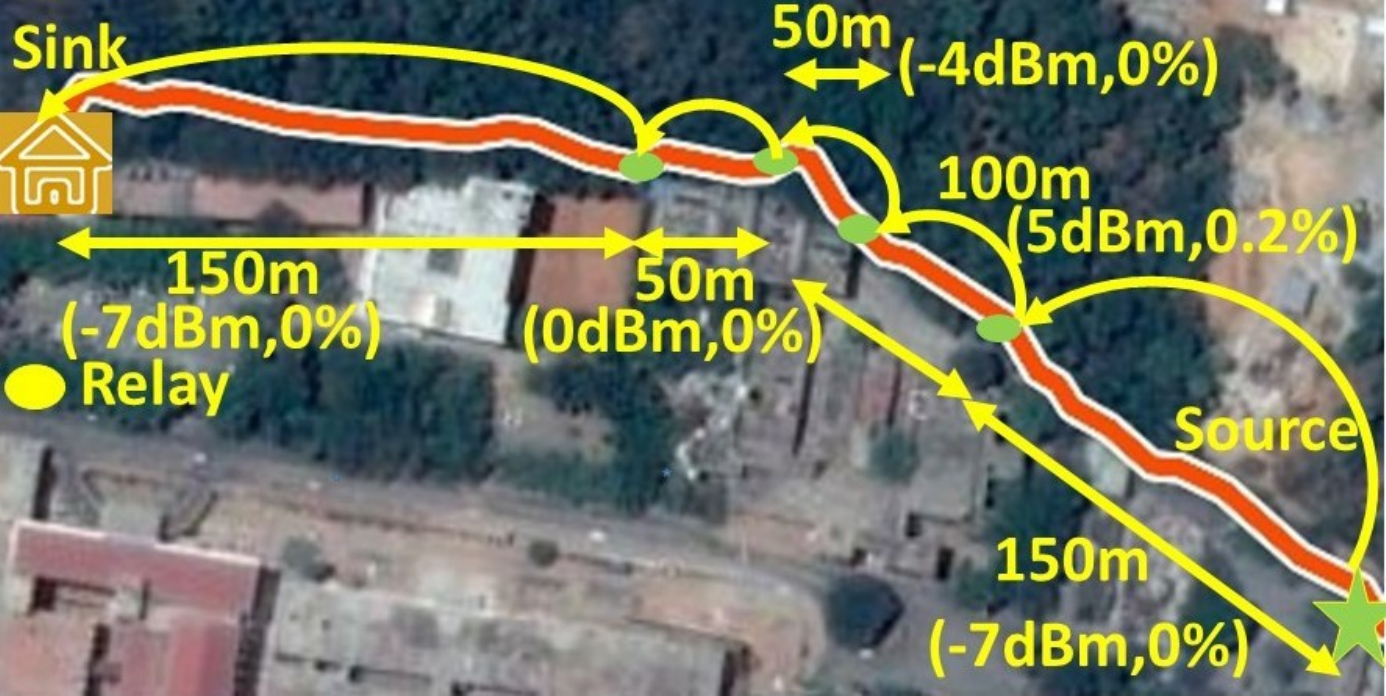}
\includegraphics[height=2.3cm, width=9cm]{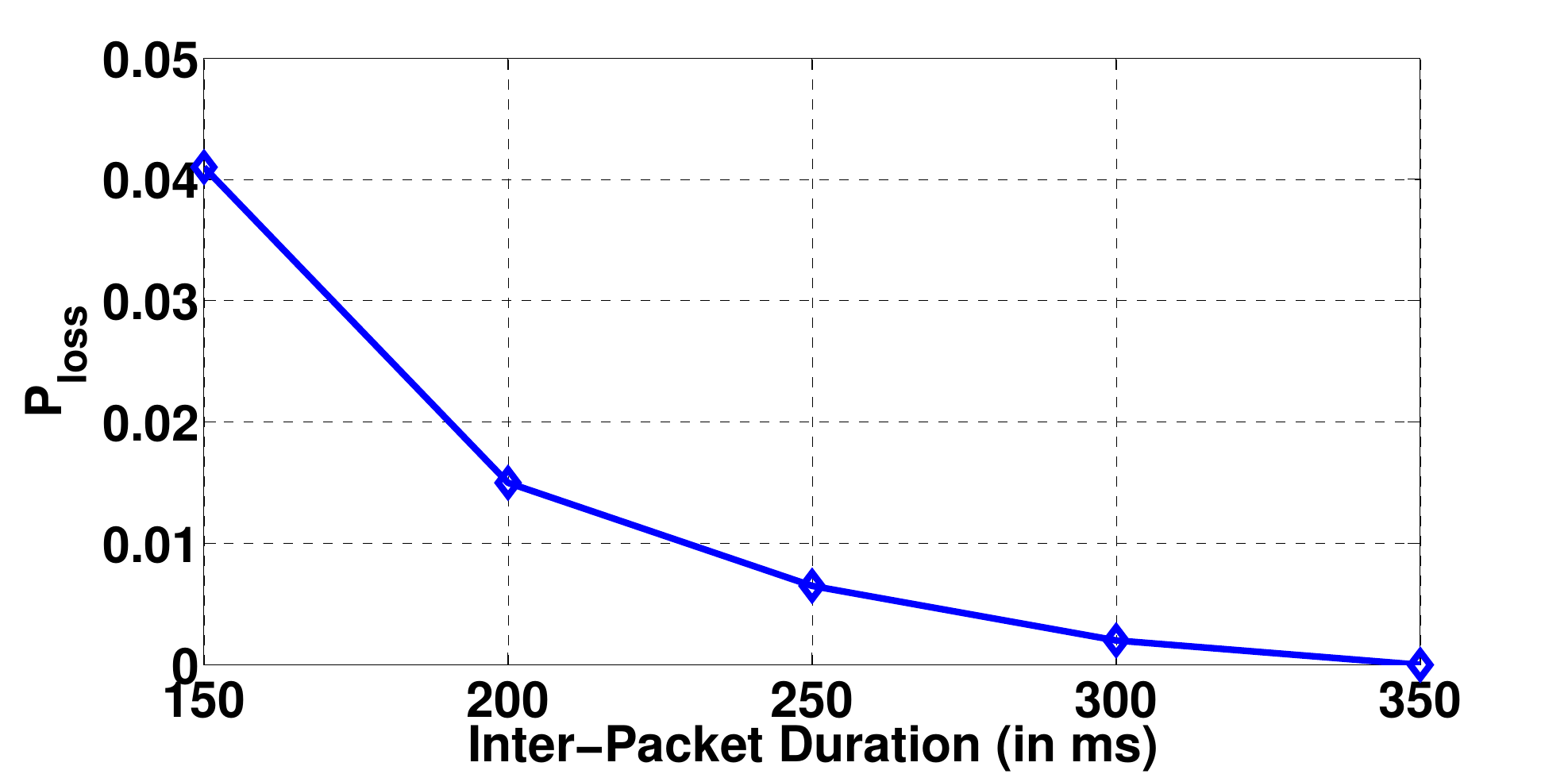}
\end{centering}
\vspace{-6mm}
\caption{Actual deployment along a long tree-lined road in the Indian Institute of Science Campus using OptExploreLimLearning  with iWiSe motes, $\xi_{out}=100$, $\xi_{relay}=1$: 
five nodes (including the source) are placed; link lengths, transmit powers, and $\%$ outage 
probabilities are shown; the plot shows variation of end-to-end loss probability with 
inter-packet duration, for periodic traffic generated from the source. Picture and plot are taken from 
\cite{chattopadhyay-etal14deployment-experience}.}
\label{fig:real_deployment_OptExploreLimLearning}
\vspace{-6mm}
\end{figure}

\vspace{-6mm}
\section{Conclusion}\label{sec:conclusion}
\vspace{-2mm}
We have developed several approaches for as-you-go deployment of wireless relay
networks  using on-line measurements, under a very light traffic assumption. Each problem was 
formulated as an MDP and its optimal policy structure was studied. 
We also studied a few learning algorithms that will asymptotically converge to the 
corresponding optimal policies. 
Numerical and experimental results have been provided to illustrate the performance and trade-offs. 

This work can be extended or modified in several ways: (i) Networks that are robust to node failures and long term link variations 
would either require each relay to have multiple neighbours (i.e., the deployment would  need to be multi-connected),  
or the nodes can choose their transmit powers adaptively as the environment changes.  
(ii) It  would be of interest to develop deployment algorithms for 2 and 3 dimensional regions, 
where a team of agents cooperates to carry out the deployment. 
(iii) We have assumed very light traffic conditions in our design (what we call ``lone packet" traffic), 
but our experiments show that these designs can carry a useful amount of positive traffic. It 
will be of interest, however, to develop deployment algorithms that can provide theoretical guarantees 
to achieve desired traffic rates. 

 \vspace{-4mm}

\bibliographystyle{unsrt}
\bibliography{arpan-techreport}

\vspace{-14mm}

\begin{IEEEbiography}[{\includegraphics[width=0.7in,height=0.7in,clip,keepaspectratio]{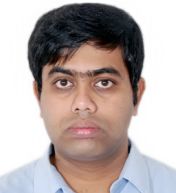}}]{Arpan 
Chattopadhyay} obtained his B.E. in Electronics and Telecommunication Engineering from Jadavpur University, 
Kolkata, India in the year 2008, and M.E. in Telecommunication Engineering from Indian Institute of Science, 
Bangalore, India in the year 2010. He is currently pursuing his PhD in ECE department, Indian Institute of Science, 
Bangalore. His research interest is in   networks.
    \end{IEEEbiography}

    \vspace{-17mm}

   \begin{IEEEbiography}[{\includegraphics[width=0.7in,height=0.7in,clip,keepaspectratio]{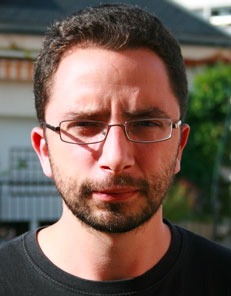}}]
   {Marceau Coupechoux}  is an Associate Professor at Telecom  ParisTech since 2005. He obtained his master from 
   Telecom ParisTech in 1999 and from University of Stuttgart, Germany in  2000, and his Ph.D. from Institut Eurecom, 
   Sophia-Antipolis, France, in 2004. From 2000 to 2005, he  was with Alcatel-Lucent. In the Computer and Network Science  
   department of Telecom ParisTech, he is working on cellular networks,  wireless networks,  
   cognitive networks, focusing mainly on  layer 2 protocols, scheduling and resource management. 
   From August 2011 to August 2012 he was a visiting  scientist at IISc Bangalore.
   \end{IEEEbiography}

   \vspace{-13mm}
   
    \begin{IEEEbiography}[{\includegraphics[width=0.7in,height=0.7in,clip,keepaspectratio]{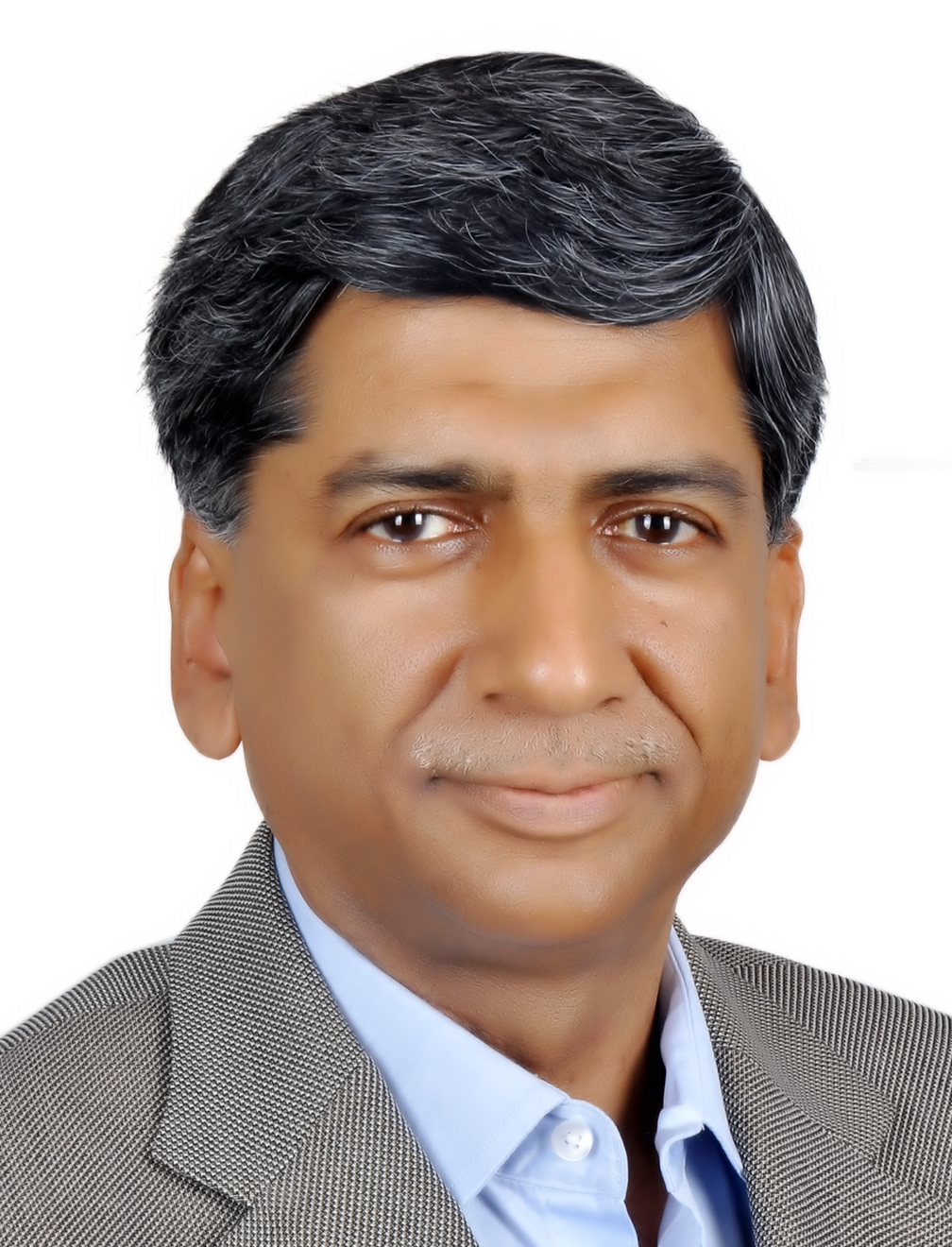}}]
    {Anurag Kumar} obtained his B.Tech. degree from the Indian Institute of 
 Technology at Kanpur, and the PhD degree from Cornell University, 
 both in Electrical Engineering. He was then with Bell Laboratories, 
 Holmdel, N.J., for over 6 years.  Since 1988 he has been on the 
 faculty of the Indian Institute of Science (IISc), Bangalore, in the 
 Department of Electrical Communication Engineering.  He is currently 
 also the Director of the Institute.  From 1988 to 2003 he was the 
 Coordinator at IISc of the Education and Research Network Project 
 (ERNET), India's first wide-area packet switching network.  His area 
 of research is communication networking, specifically, modeling, 
 analysis, control and optimisation problems arising in communication 
 networks and distributed systems. Recently his research has focused 
 primarily on wireless networking.  He is a Fellow of the IEEE, of the 
 Indian National Science Academy (INSA), of the Indian Academy of 
 Science (IASc), of the Indian National Academy of Engineering (INAE), 
 and of The World Academy of Sciences (TWAS). He is a recepient of the 
 Indian Institute of Science Alumni Award for Engineering Research for 
 2008.
 \end{IEEEbiography}

\begin{figure*}[t!]
\centering{\huge{\bf Supplementary Material} }
\end{figure*}

\newpage

\renewcommand{\thesubsection}{\Alph{subsection}}

\appendices

\section{Pure As-You-Go Deployment}\label{appendix:average_cost_sum_outage_no_backtracking}

\textbf{Proof of Lemma~\ref{lemma:value_function_properties_sum_power_sum_outage_no_backtracking}}
Note that the function $J^{(0)}(\cdot):=0$ satisfies all the assertions. Let us assume, as our induction 
hypothesis, that $J^{(k)}(\cdot)$ satisfies 
all the assertions. Now $Q_{out}(r,\gamma,w)$ is increasing in $r$ and decreasing in $w$ (by our 
channel modeling assumptions in Section~\ref{subsection:channel_model}), 
and the single stage 
costs are linear (hence concave) increasing in $\xi_{relay}$, $\xi_{out}$. 
Then from the value iteration, 
$J^{(k+1)}(r,w)$ is pointwise minimum of functions which are increasing in $r$, $\xi_{out}$ and $\xi_{relay}$, 
decreasing in $w$, and jointly concave in $\xi_{out}$ and $\xi_{relay}$. 
Hence, the assertions hold for $J^{(k+1)}(r,w)$. 
Similarly, we can show that the assertions hold for $J^{(k+1)}(\mathbf{0})$. 
Since $J^{(k)}(\cdot) \uparrow J(\cdot)$, the results follow.

\textbf{Proof of Theorem~\ref{theorem:policy_structure_sum_power_sum_outage_no_backtracking}}
Consider the Bellman equation (\ref{eqn:bellman_equation_sum_power_sum_outage_no_backtracking}). We will place a relay at state 
$(r,w)$ iff the cost of placing a relay, i.e., 
$\min_{\gamma \in \mathcal{S}} (\gamma+\xi_{out} Q_{out}(r,\gamma,w) )+\xi_{relay} + J(\mathbf{0})$ is less than or equal to 
the cost of not placing, i.e., $\theta \mathbb{E}_W \min_{\gamma \in \mathcal{S}} (\gamma+ \xi_{out}  Q_{out} (r+1,\gamma,W)) + 
(1-\theta)\mathbb{E}_W J(r+1,W)$. Hence, it is obvious 
that we will place a relay at state $(r,w)$ iff 
$\min_{\gamma \in \mathcal{S}} (\gamma+\xi_{out} Q_{out}(r,\gamma,w) ) \leq c_{th}(r)$ where the threshold $c_{th}(r)$ is given by:

\footnotesize
\begin{eqnarray}
 c_{th}(r)&=&\theta \mathbb{E}_W \min_{\gamma \in \mathcal{S}} (\gamma+ \xi_{out}  Q_{out} (r+1,\gamma,W)) \nonumber\\
&& + (1-\theta)\mathbb{E}_W J(r+1,W) -(\xi_{relay} + J(\mathbf{0})) \label{eqn:c_th_r_expression}
\end{eqnarray}
\normalsize

By Proposition~$3.1.3$ of \cite{bertsekas07dynamic-programming-optimal-control-2}, if there
exists a stationary policy $\{\mu,\mu,\cdots\}$ such that for each state, the action chosen by the policy is the action that
achieves the minimum in the Bellman equation, then that stationary policy will be an optimal policy, i.e., 
the minimizer in Bellman equation gives the optimal action. Hence, if the decision is to place a relay at state 
$(r,w)$, then the power has to be chosen as 
$\argmin_{\gamma \in \mathcal{S}} \bigg(\gamma+\xi_{out} Q_{out}(r,\gamma,w)\bigg)$. 

Since $Q_{out}(r,\gamma,w)$ and $J(r,w)$ is increasing in $r$ for each $\gamma,w$, it is easy to see 
that $c_{th}(r)$ is increasing in $r$.

\section{Explore Forward Deployment}\label{appendix:backtracking_average_cost}

\textbf{Proof of Theorem~\ref{theorem:smdp-cost-vs-xi}}
Let us recall the definition of the functions $\mu^{(1)}$ and $\mu^{(2)}$. Now, 
$\lambda_{\mu}:=\frac{\xi_{relay}+\sum_{\underline{w}} g(\underline{w}) \bigg(\mu^{(2)}(\underline{w})+ \xi_{out} Q_{out}(\mu^{(1)}(\underline{w}),\mu^{(2)}(\underline{w}),w_{\mu^{(1)}(\underline{w})})\bigg)} 
{ \sum_{\underline{w}} g(\underline{w}) \mu^{(1)}(\underline{w}) ×}$ is the average cost of a specific  
stationary deterministic policy $\mu$ (by the Renewal Reward Theorem, 
since the placement process regenerates at each placement point). For each policy $(\mu^{(1)},\mu^{(2)})$, 
the numerator is linear, increasing in  $\xi_{out}$ and $\xi_{relay}$  
and the denominator is independent of  $\xi_{out}$ and $\xi_{relay}$.  Now, 
$\lambda^{*}(\xi_{out},\xi_{relay}) = \inf_{\mu} \lambda_{\mu}$. 
Hence, the proof follows immediately 
since the pointwise infimum of increasing linear functions of  $\xi_{out}$ and $\xi_{relay}$ is increasing and 
jointly concave in  $\xi_{out}$ and $\xi_{relay}$, and since any increasing, concave function is continuous.

\textbf{Proof of Theorem~\ref{theorem:outage_decreasing_with_xio_placement_rate_decreasing_with_xir}:} 
We will prove only the second statement of the theorem since the proof of the first statement is similar.

Consider any $\kappa>0$. 

Now, since the mean cost per step is a linear combination of the mean power per step, mean outage per step and the mean 
number of relays per step, we can write:

\footnotesize
\begin{eqnarray}
&& \lambda^*(\xi_{out}, \xi_{relay}) \nonumber\\
&=& \frac{\overline{\Gamma}^*(\xi_{out},\xi_{relay})+\xi_{out}\overline{Q}_{out}^*(\xi_{out},\xi_{relay})+\xi_{relay}}{\overline{U}^*(\xi_{out},\xi_{relay})} \nonumber\\
& \leq & \frac{\overline{\Gamma}^*(\xi_{out}+\kappa,\xi_{relay})+\xi_{out}\overline{Q}_{out}^*(\xi_{out}+\kappa,\xi_{relay})+\xi_{relay}}{\overline{U}^*(\xi_{out}+\kappa,\xi_{relay})} \nonumber\\
&& \label{eqn:inequality1}
\end{eqnarray}
\normalsize

and 

\footnotesize
\begin{eqnarray}
&& \lambda^*(\xi_{out}+\kappa, \xi_{relay}) \nonumber\\
&=& \frac{\overline{\Gamma}^*(\xi_{out}+\kappa,\xi_{relay})+(\xi_{out}+\kappa)\overline{Q}_{out}^*(\xi_{out}+\kappa,\xi_{relay})+\xi_{relay}}{\overline{U}^*(\xi_{out}+\kappa,\xi_{relay})} \nonumber\\
& \leq & \frac{\overline{\Gamma}^*(\xi_{out},\xi_{relay})+(\xi_{out}+\kappa)\overline{Q}_{out}^*(\xi_{out},\xi_{relay})+\xi_{relay}}{\overline{U}^*(\xi_{out},\xi_{relay})} \nonumber\\
&& \label{eqn:inequality2}
\end{eqnarray}
\normalsize

where the inequality in (\ref{eqn:inequality1}) follows from the fact that $\pi^*(\xi_{out}, \xi_{relay})$ is 
an optimal policy for $(\xi_{out}, \xi_{relay})$, and the inequality in 
(\ref{eqn:inequality2}) follows from the fact that $\pi^*(\xi_{out}+\kappa, \xi_{relay})$ is 
an optimal policy for $(\xi_{out}+\kappa, \xi_{relay})$.

Adding the inequalities (\ref{eqn:inequality1}) and (\ref{eqn:inequality2}) and cancelling the common terms, 
we obtain that 
$\frac{\overline{Q}_{out}^*(\xi_{out}+\kappa,\xi_{relay})}{\overline{U}^*(\xi_{out}+\kappa,\xi_{relay})} 
\leq \frac{\overline{Q}_{out}^*(\xi_{out},\xi_{relay})}{\overline{U}^*(\xi_{out},\xi_{relay})}$.\qed

{\bf Proof of Theorem~\ref{theorem:optimality-condition-required-for-backtracking-average-cost-learning}:}
From (\ref{eqn:optimality-smdp}), we can write:
\begin{eqnarray*}
 \sum_{\underline{w}}g(\underline{w}) v^*(\underline{w}) &=& \sum_{\underline{w}}g(\underline{w}) \bigg( \min_{u, \gamma}\bigg\{ \gamma+  \xi_{out} Q_{out}(u,\gamma,w_u) \nonumber\\
&& +\xi_{relay}-\lambda^{*}u \bigg\} \bigg)+\sum_{\underline{w}^{'} \in \mathcal{W}^B}g(\underline{w}^{'}) v^*(\underline{w}^{'}) \nonumber\\ \nonumber\\
\end{eqnarray*}
Cancelling $\sum_{\underline{w}}g(\underline{w}) v^*(\underline{w})$ from both sides, we obtain the desired result.

{\bf Proof of Theorem~\ref{theorem:complexity-reduction-smdp-policy-iteration}:}
Note that in (\ref{eqn:smdp-policy-improvement}), if the minimum is achieved by more than one pair of $(u, \gamma)$, 
then any one of them 
can be considered to be the optimal action. Let us use the convention that among all minimizers the pair $(u,\gamma)$ 
with minimum $u$ will be considered 
as the optimal action, and if there are more than one such minimizing pair with same values of $u$, then the pair 
with smallest value of $\gamma$ will be considered. 
We recall that $\mathcal{S}=\{P_1,P_2,\cdots,P_M \}$.  
Let us denote, under policy $\mu_{k+1}$, 
the probability that the optimal control is $(u,\gamma)$ and 
the shadowing is $w$ at the $u$-th location, by $b_k(u,\gamma, w)$. 
Then,

\footnotesize
\begin{eqnarray}
 b_k(u,\gamma,w)&=& \Pi_{r=A+1}^{u-1} \mathbb{P} \bigg(\min_{\gamma^{'} \in \mathcal{S}} (\gamma^{'}+ \xi_{out} Q_{out}(r,\gamma^{'},W_r))-\lambda_k r \nonumber\\
&& >\gamma+\xi_{out} Q_{out}(u,\gamma,w)-\lambda_k u \bigg) \times p_W(w) \nonumber\\
&& \times \Pi_{r=u+1}^{A+B} \mathbb{P} \bigg(\min_{\gamma^{'} \in \mathcal{S}} (\gamma^{'}+ \xi_{out} Q_{out}(r,\gamma^{'},W_r))-\lambda_k r \nonumber\\ 
&& \geq \gamma+\xi_{out} Q_{out}(u,\gamma,w)-\lambda_k u\bigg) \nonumber\\
&& \times \mathbb{I}\bigg\{\gamma=\argmin\{P_1,P_2,\cdots,P_M\}: \nonumber\\
&& \gamma+Q_{out}(u,\gamma,w) \nonumber\\
&& =\min_{\gamma^{'}}(\gamma^{'}+\xi_{out} Q_{out}(u,\gamma^{'},w))\bigg\}  \label{eqn:smdp-probability-calculation}
\end{eqnarray}
\normalsize
Now, we can write,

\footnotesize
\begin{eqnarray}
&& \sum_{\underline{w}} g(\underline{w}) \bigg(\mu^{(2)}_k(\underline{w})+\xi_{out} Q_{out}(\mu^{(1)}_k(\underline{w}),\mu^{(2)}_k(\underline{w}),w_{\mu^{(1)}_k(\underline{w})})\bigg) \nonumber\\
&=& \sum_{u=A+1}^{A+B}\sum_{j=1}^{M} \sum_{w \in \mathcal{W}} b_{k-1}(u,P_j,w) \bigg(P_j+\xi_{out} Q_{out}(u,P_j,w)\bigg) \nonumber\\
\label{eqn:smdp-numerator-simpler}
\end{eqnarray}
\normalsize
and 

\footnotesize
\begin{eqnarray}
 \sum_{\underline{w}} g(\underline{w}) \mu^{(1)}_k(\underline{w})
&=& \sum_{u=A+1}^{A+B}\sum_{j=1}^{M} \sum_{w \in \mathcal{W}} b_{k-1}(u,P_j,w) u  \nonumber\\
&=& \sum_{u=A+1}^{A+B} u \sum_{j=1}^{M}\sum_{w \in \mathcal{W}} b_{k-1}(u,P_j,w) \label{eqn:smdp-denominator-simpler}
\end{eqnarray}
\normalsize

Now, for each $(u,\gamma,w)$, $b_{k-1}(u,\gamma,w)$ (in (\ref{eqn:smdp-probability-calculation})) can be computed in 
$O(BM|\mathcal{W}|)$ operations. 
Hence, total number of operations required to compute $b_{k-1}(u,\gamma,w)$ for all $u,\gamma,w$ is $O(B^2 M^2 |\mathcal{W}|^2)$. 
Now, only $O(BM|\mathcal{W}|)$ operations are required in (\ref{eqn:smdp-numerator-simpler}) and (\ref{eqn:smdp-denominator-simpler}). 
Hence, the number of computations required in each iteration is $O(B^2 M^2 |\mathcal{W}|^2)$.

Note that, the policy improvement step is not explicitly required in the policy iteration. This is because in the 
policy evaluation step, $\lambda_k$ is sufficient to compute $b_k(u, \gamma, w)$ for all $u,\gamma,w$ 
and thereby to compute $\lambda_{k+1}$. 
Hence, we need not store the policy in each iteration.\qed

{\bf Proof of Lemma~\ref{lemma:what_is_the_heuristic_doing}:} 
 Let us denote the HeuExploreLim policy by $\mu_h$ and any other stationary, deterministic 
policy by $\mu$. Let us denote the sequence of link costs 
incurred in the deployment process (for a semi-infinite line with given shadowing over all possible links) 
under policy $\mu_h$ by $c_{\mu_h,1}, c_{\mu_h,2}, \cdots$ and the corresponding 
link lengths by $u_{\mu_h,1}, u_{\mu_h,2}, \cdots$. Let us denote, under policy $\mu_h$, 
the shadowing observed at the $i$-th location (where $A+1 \leq i \leq A+B$) in the measurement process 
for the placement of the $l$-th node,   
by $w_{i,l}$. Now, let us couple the deployment processes 
under policies $\mu$ and $\mu_h$ in the following way. Suppose that, under policy $\mu$, 
the shadowing observed at the $i$-th location 
for the placement of the $l$-th node is again $w_{i,l}$ (this is valid since shadowing is i.i.d across links). 
Clearly, $\frac{c_{\mu_h,j}}{u_{\mu_h,j}} \leq \frac{c_{\mu,j}}{u_{\mu,j}}$. Hence, by the strong law of large 
numbers, $\mathbb{E}_{\mu_h}\bigg( \frac{C_{\mu_h}}{U_{\mu_h}} \bigg) \leq \mathbb{E}_{\mu}\bigg( \frac{C_{\mu}}{U_{\mu}} \bigg)$, 
since $\bigg( \frac{C_{\mu_h,j}}{U_{\mu_h,j}} \bigg)$ is i.i.d. across $j$ due to i.i.d. shadowing across links.

\vspace{-3mm}
\section{Comparison between Explore-Forward and Pure As-You-Go Approaches}
\label{appendix:comparison_average_cost_backtracking_no_backtracking}
\vspace{-1mm}

\begin{figure}[!t]
\centering
\includegraphics[scale=0.24]{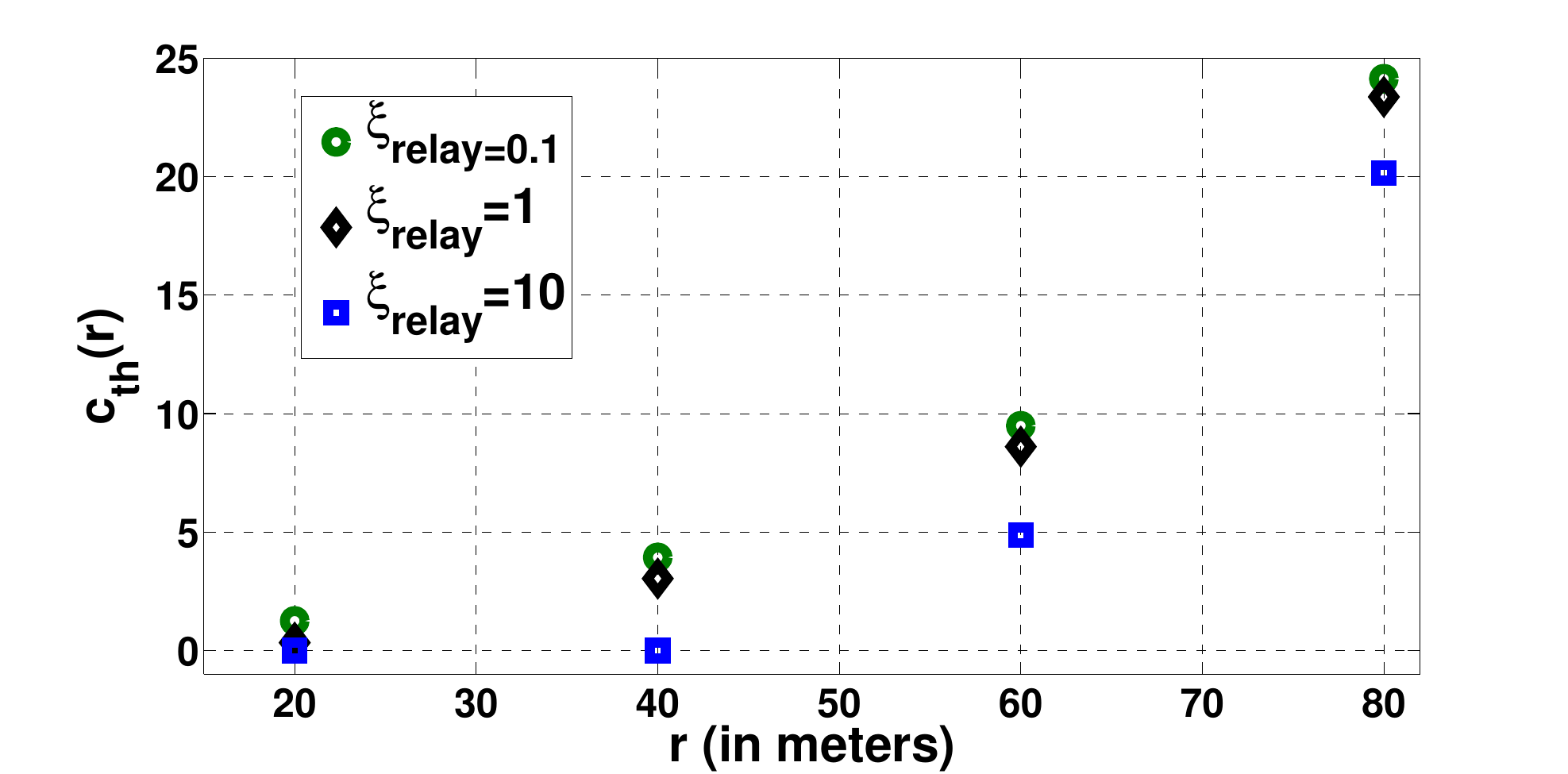}
\vspace{-0mm}
\caption{Pure as-you-go deployment; variation of $c_{th}(r)$ with $r$ for $\xi_{out}=100$ and various values of $\xi_{relay}$.}
\label{fig:threshold_vs_distance_various_relay_cost}
\vspace{-3mm}
\end{figure}
\vspace{-2mm}
\begin{figure}[!t]
\centering
\includegraphics[scale=0.24]{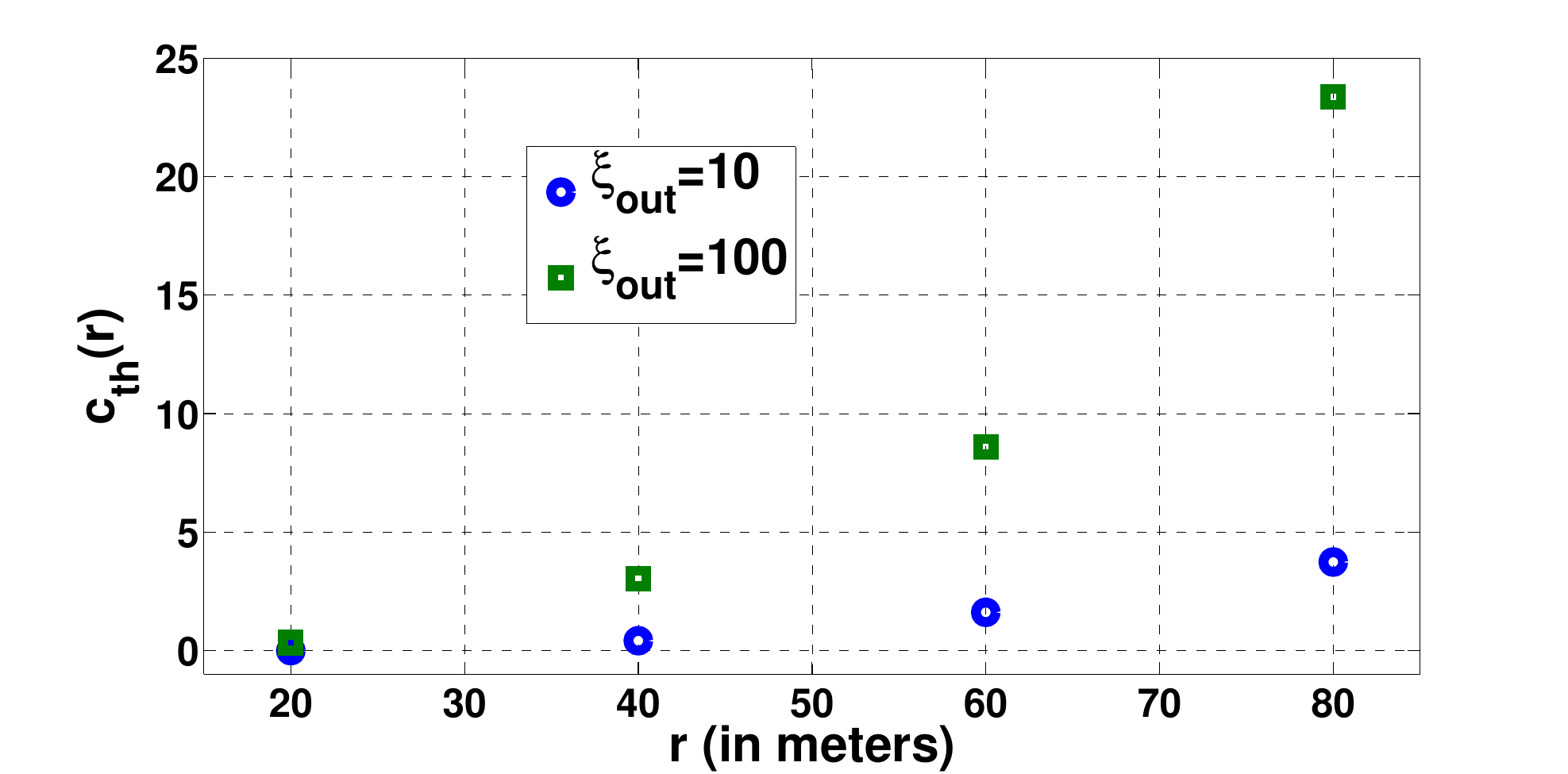}
\vspace{-0mm}
\caption{Pure as-you-go deployment; variation of $c_{th}(r)$ with $r$ for $\xi_{relay}=1$ and various values of $\xi_{out}$.}
\label{fig:threshold_vs_distance_various_outage_cost}
\vspace{-3mm}
\end{figure}

\begin{figure*}[t]
\begin{minipage}[r]{0.48\linewidth}
\subfigure{
\includegraphics[width=\linewidth, height=3.5cm]{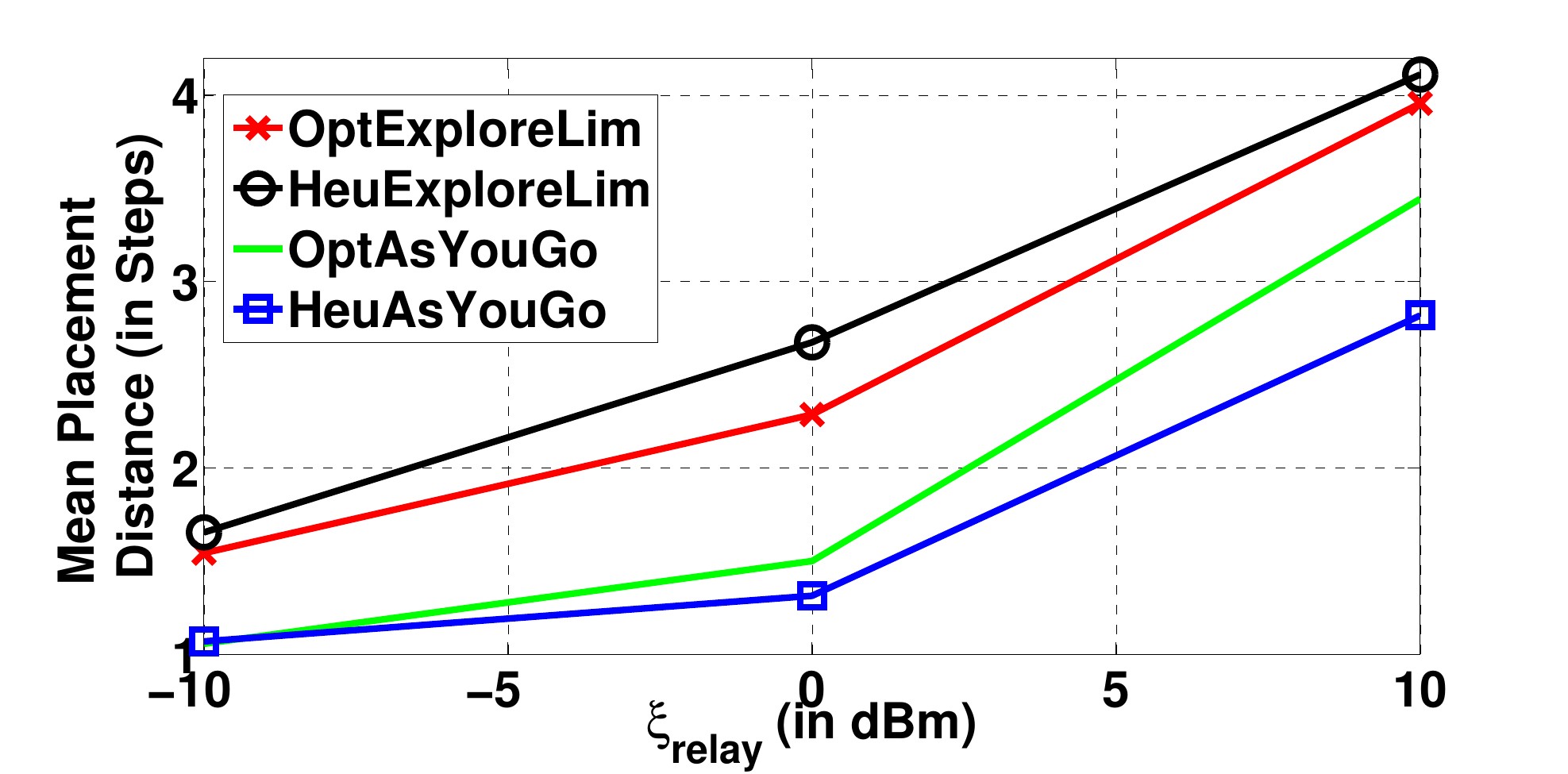}
\includegraphics[width=\linewidth, height=3.5cm]{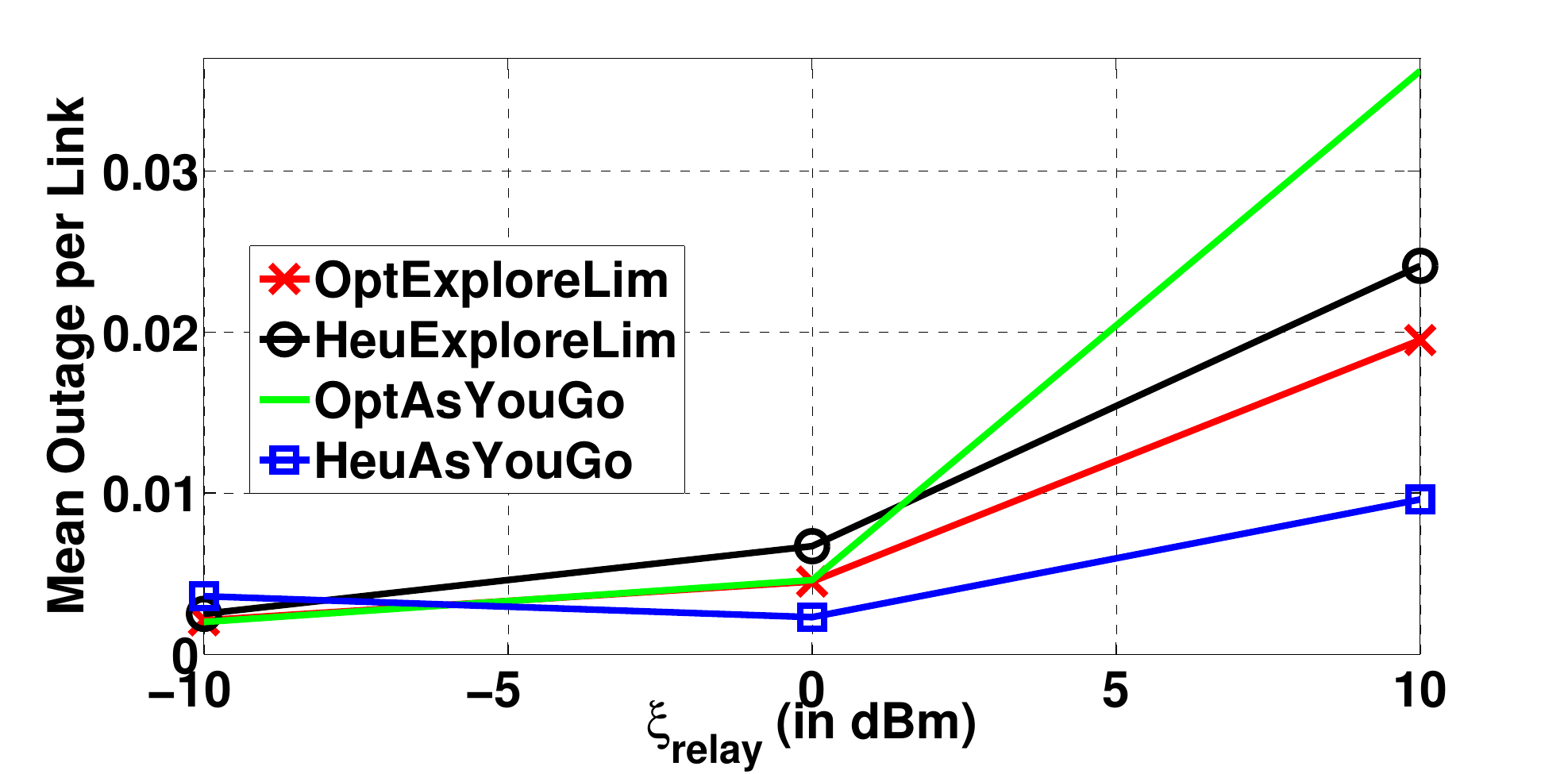}
\label{fig:cost-model-based}}
\end{minipage}  \hfill
\end{figure*}
\vspace{-0mm}
\begin{figure*}[t]
\begin{minipage}[c]{0.48\linewidth}
\subfigure{
\includegraphics[width=\linewidth, height=3.5cm]{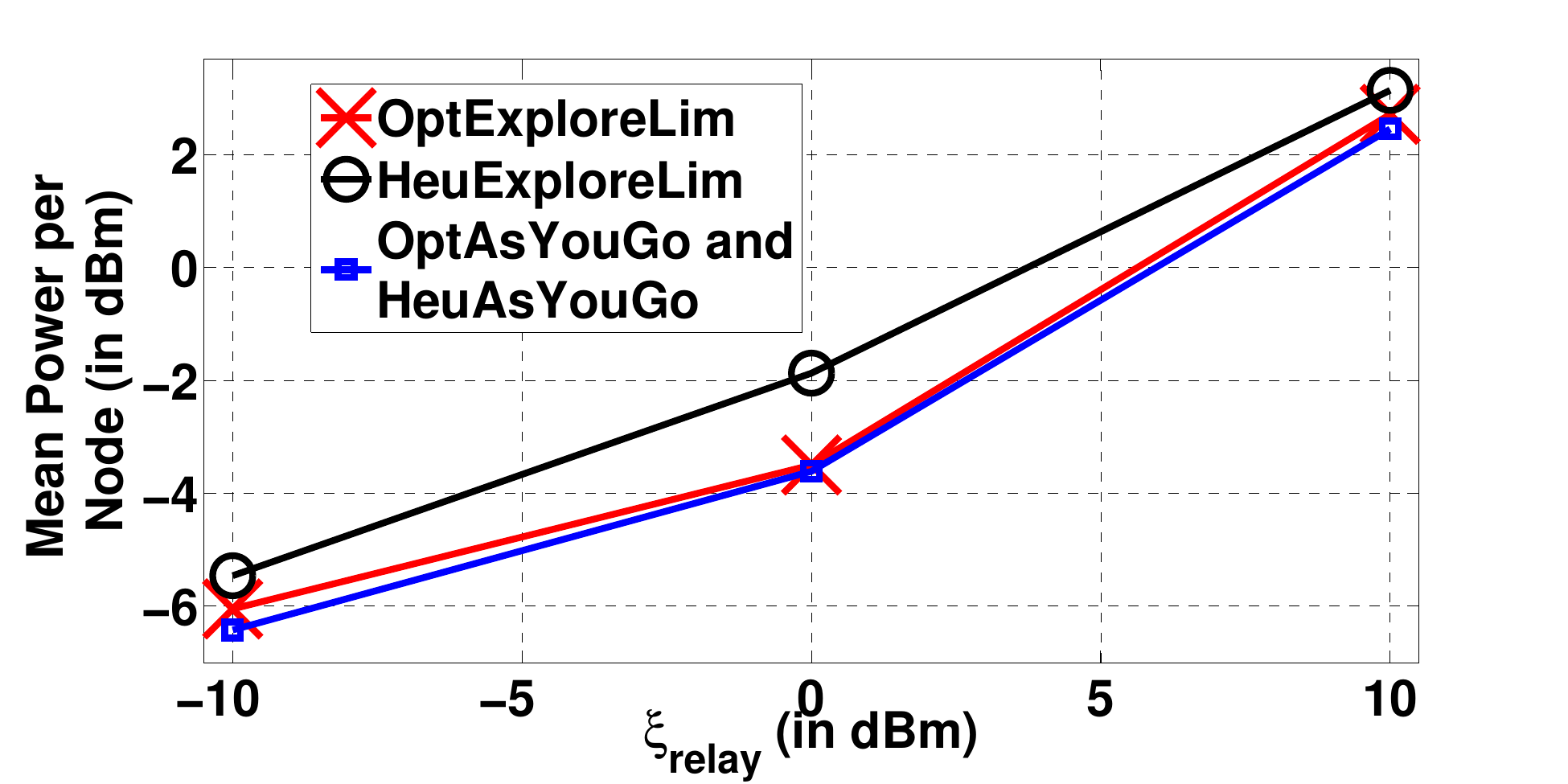}
\includegraphics[width=\linewidth, height=3.5cm]{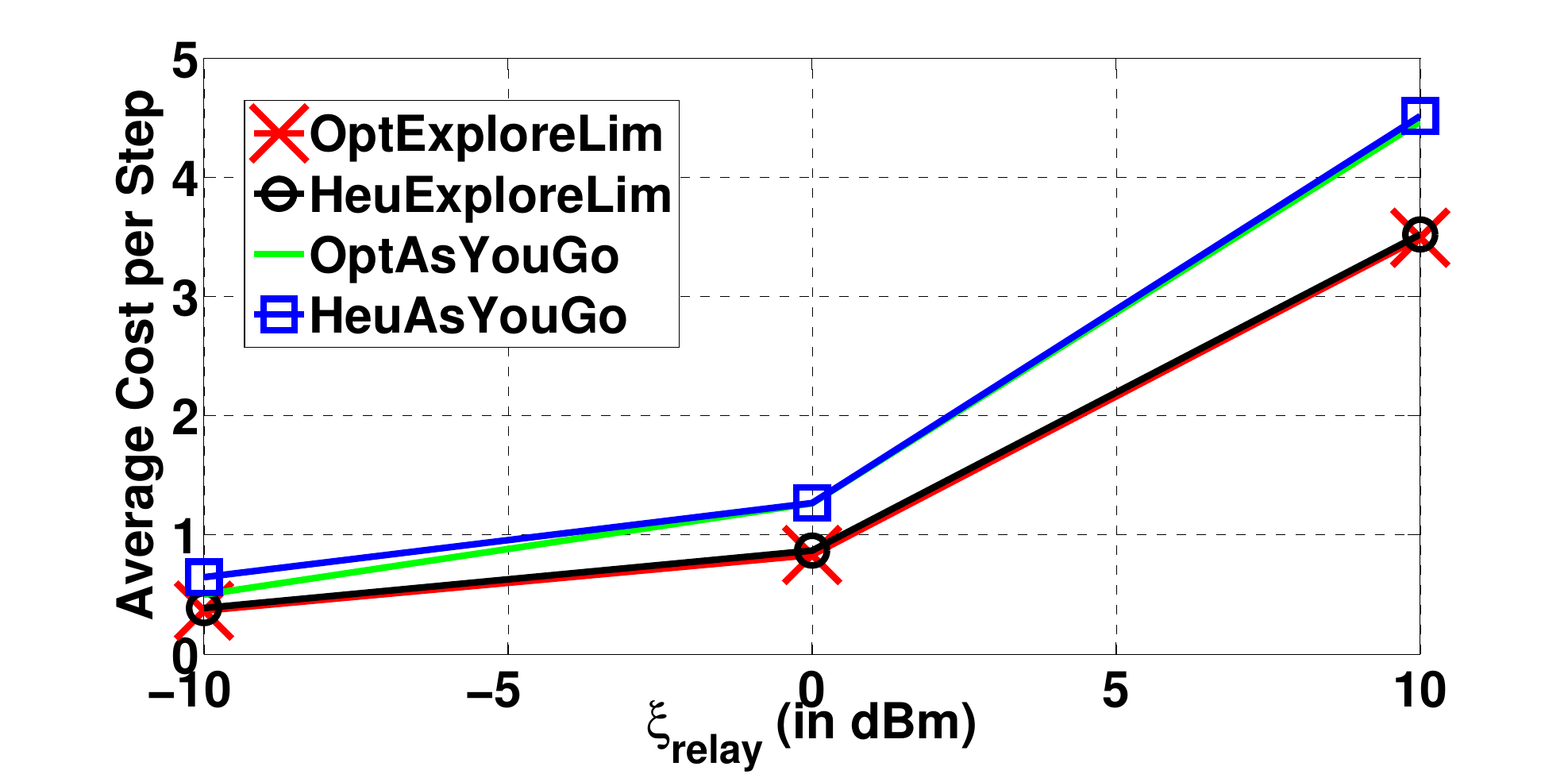}
\label{fig:outage-model-based}}
\end{minipage} \hfill
\caption{Results for $\xi_{out}=100$: mean cost per step, mean power per link, mean outage per link and 
mean placement distance (steps) vs. $\xi_{relay}$ for the four algorithms: 
OptExploreLim, OptAsYouGo, HeuExploreLim, and HeuAsYouGo. Unit of $\xi_{relay}$ is actually mW, but in this figure 
it is shown in dBm. $\xi_{relay}$, when expressed in dBm, is equal to $10 \log_{10}(\xi_{relay})$. 
In the Power plot, the HeuAsYouGo plot overlaps the OptAsYouGo plot, since 
the node power in the HeuAsYouGo algorithm was taken to be the same as the mean node power with the OptAsYouGo algorithm.}
\label{fig:cost-power-outage-distance-model-based_vs_xi_relay}
\vspace{-0mm}
\end{figure*}

\textbf{Proof of Theorem~\ref{theorem:comparison-backtracking-no-backtracking}}
Note that for the average cost problem with pure as-you-go, there exists an optimal threshold policy 
(similar to Theorem~\ref{theorem:policy_structure_sum_power_sum_outage_no_backtracking}), since the optimal policy for problem 
(\ref{eqn:sum_power_discounted_no_backtracking}) achieves $\lambda_{ayg}^*$ average cost per step for $\theta$ sufficiently close to $0$. 
So, let one such optimal policy be given by the set of thresholds $\{c_{th}(r)\}_{A+1 \leq r \leq A+B-1}$. 

Now, let us consider the average cost minimization problem with explore-forward. Consider the policy where we 
first measure $w_{A+1},w_{A+2},\cdots,w_{A+B}$ and decide 
to place a relay $u$ steps away from the previous relay (where $A+1 \leq u \leq A+B-1$) if 
$\min_{\gamma \in \mathcal{S}} (\gamma+\xi_{out} Q_{out}(r,\gamma,w_r)) > c_{th}(r)$ for all $r \leq (u-1)$ and 
$\min_{\gamma \in \mathcal{S}} (\gamma+\xi_{out} Q_{out}(u,\gamma,w_u)) \leq c_{th}(u)$. 
We must place if we reach at a distance $(A+B)$ from the previous relay. 
But this is a particular policy for the problem where we gather 
$w_{A+1},w_{A+2},\cdots,w_{A+B}$ and then decide where to place the relay, and clearly the average cost per step 
for this policy is $\lambda_{ayg}^*$ which cannot be less than the optimal average cost $\lambda_{ef}^*$.\qed

\subsection{Optimal Policy Structure for the Pure As-You-Go Approach}
\label{subsection:numerical_policy_structure_sum_power_sum_outage_no_backtracking_average_cost}

The variation of $c_{th}(r)$ (see Section~\ref{subsection:policy_structure_sum_power_sum_outage_no_backtracking} and 
Section~\ref{subsection:average-cost-no-backtracking-limit approach}, we have taken $\theta$ sufficiently close to $0$) 
with $r$, for various values of the relay cost $\xi_{relay}$ and the cost of outage $\xi_{out}$, has been shown 
in Figure~\ref{fig:threshold_vs_distance_various_relay_cost} and Figure~\ref{fig:threshold_vs_distance_various_outage_cost}. 
For a fixed $\xi_{out}$, $c_{th}(r)$ decreases with $\xi_{relay}$; i.e., as the cost of placing 
a relay increases, we place relays less 
frequently. On the other hand, for a fixed $\xi_{relay}$, $c_{th}(r)$ increases with $\xi_{out}$. 
This happens because if the cost of 
outage increases, we cannot tolerate outage and place the relays close to each other. 
Note also that, $c_{th}(r)$ increases in $r$ as stated in 
Algorithm~\ref{algorithm:OptAsYouGo}.

\vspace{-2mm}
\subsection{Comparison Among Various Deployment Algorithms}
\label{subsection:comparison_backtracking_no_backtracking_average_cost}
\vspace{-2mm}

Next, assuming a system model as described in Section~\ref{sec:system_model_and_notation} and assuming the parameter 
values as in 
Section~\ref{subsection:parameter_values}, we computed 
the mean cost per step, mean power per node, mean outage per link and mean placement distance (between successive relays) 
for four deployment algorithms presented so far\footnote{Note that, these computations were done on MATLAB; they did not 
involve any field deployment. Field experimentations were done only to validate the assumptions (such as independent shadowing 
assumption) and to compute the values of the  parameters such as $\eta$ and $\sigma$.}. Some of the results are shown  
in Figure~\ref{fig:cost-power-outage-distance-model-based_vs_xi_relay}. 
In order to make a fair comparison, {\em we used the mean power per node 
for OptAsYouGo as the fixed node transmit power for HeuAsYouGo, and the mean outage per link of OptAsYouGo 
as the pre-fixed target outage for HeuAsYouGo.} 
The following observations are from the plots in Figure~\ref{fig:cost-power-outage-distance-model-based_vs_xi_relay}.

\subsubsection{Mean Placement Distance (see the top left 
panel of Figure~\ref{fig:cost-power-outage-distance-model-based_vs_xi_relay})}
Pure as-you-go algorithms (OptAsYouGo, HeuAsYouGo) 
place relays sooner than the algorithms that explore forward  
(OptExploreLim, HeuExploreLim) before placing a relay (see Figure~\ref{fig:cost-power-outage-distance-model-based_vs_xi_relay}). 
This is as expected, since pure as-you-go algorithms 
do not have the advantage of 
exploring over several locations and then picking the best. A pure as-you-go approach tends to be cautious, 
and therefore tries to avoid a high outage by placing relays frequently. 
As $\xi_{relay}$ (cost of a relay) increases, relays will be placed less frequently 
(according to Theorem~\ref{theorem:outage_decreasing_with_xio_placement_rate_decreasing_with_xir}). 

\subsubsection{Mean Outage per Link (see the top right 
panel of Figure~\ref{fig:cost-power-outage-distance-model-based_vs_xi_relay})}
As $\xi_{relay}$ increases, the mean outage per link increases because 
we will place fewer relays with higher inter-relay distances. 
 Pure as-you-go algorithms have link outage 
probability comparable to explore-forward algorithms, but they place relays too frequently. 
We observe that the per-link outage of HeuAsYouGo is different from that of OptAsYouGo. This happens because whenever we 
place a node using HeuAsYouGo, the exact outage target is never met with equality. 
Also, the per-link outage may decrease with $\xi_{relay}$ for HeuAsYouGo. As 
$\xi_{relay}$ increases, the node power and the target outage (chosen from OptAsYouGo) increases 
in such a way that the per-link outage for HeuAsYouGo behaves in this fashion. 

We have also observed that, as $\xi_{out}$, the penalty for outage,  increases, 
the mean outage per link decreases. But that result has not been shown  here.

\subsubsection{Mean Power per Link (see the bottom left 
panel of Figure~\ref{fig:cost-power-outage-distance-model-based_vs_xi_relay})}

     Increasing $\xi_{relay}$ will place relays 
      less frequently, hence the transmit power increases. OptAsYouGo has smaller placement distance compared to 
      OptExploreLim and HeuExploreLim, and hence it uses less power at each hop; we note, however, 
that OptAsYouGo places more relays, and, hence, could still end up using more power per step. 

In the power plot, the HeuAsYouGo plot overlaps the OptAsYouGo plot, 
since the node power in the HeuAsYouGo algorithm was taken to be the same as the
mean node power with the OptAsYouGo algorithm.

We have also seen that increasing $\xi_{out}$ (the cost per unit outage) will lower outage and 
      hence the per-node transmit power increases. 

\subsubsection{Network Cost Per Step (see the bottom right  
panel of Figure~\ref{fig:cost-power-outage-distance-model-based_vs_xi_relay})}

The network cost per step is the optimal average cost per step; 
see (\ref{eqn:unconstrained_problem_average_cost_with_outage_cost}). 
Cost increases with $\xi_{relay}$ (see Figure~\ref{fig:cost-power-outage-distance-model-based_vs_xi_relay}) and $\xi_{out}$. 
 OptAsYouGo has a larger cost than OptExploreLim and HeuExploreLim, owing to shorter links. 
{\em The average cost per step of HeuExploreLim  is very close to OptExploreLim and cost of 
HeuAsYouGo is close to OptAsYouGo, even though the heuristic policies are not optimal.} 
However, we observed that this does not always happen. 
For example, for $\xi_{relay}=0.1$ and $\xi_{out}=1000$, we found that the average cost per step 
for OptAsYouGo and HeuAsYouGo are $1.3485$ and $1.9581$ respectively, and the average cost per step 
for OptExploreLim and HeuExploreLim are $0.9810$ and $1.0537$ respectively. 

\textbf{Discussion:} 

\begin{enumerate}[label=(\roman{*}),noitemsep,nolistsep]
 \item HeuExploreLim and HeuAsYouGo appear to be attractive at the first sight 
because they are intuitive, easy to implement, and they do not require any 
channel model for given $\xi_{out}$ and $\xi_{relay}$. But, they are suboptimal, and 
we do not have any performance guarantee (e.g., the optimality gap w.r.t. the optimal algorithms 
OptExploreLim and OptAsYouGo). Hence, if we know the radio propagation model (e.g., $\eta$ and $\sigma$) exactly, and 
if $\xi_{out}$ and $\xi_{relay}$ are given, 
it is better to compute the optimal policies and then deploy according to them.

\item Note that, the mean number of measurements 
made per step for the pure as-you-go approach is $1$, whereas it is $\frac{B}{\mathbb{E}(U)}$ 
under the explore-forward approach, where 
$\mathbb{E}(U)$ is the mean distance between successive relays. From the numerical results 
presented in this section, we find that, under the explore-forward approach, 
the mean number of measurements required will be at most $3$, and can be even less than $2$ depending on the situation. 
For applications that do not require rapid deployment, such 
as deployment in a large forest for monitoring purpose, this many measurements is affordable. {\em Hence, for the learning 
algorithms, we   consider only explore-forward approach.}

\item  More importantly, in practice 
the propagation environment will not be known, and, in order to solve 
the problem defined in (\ref{eqn:constrained_problem_average_cost_with_outage_cost}), we need to choose 
$\xi_{out}^*$ and $\xi_{relay}^*$ while deploying 
(as explained in Theorem~\ref{theorem:how-to-choose-optimal-Lagrange-multiplier}), if possible. 
But we cannot choose this pair if we do not have a prior knowledge of the propagation environment. 
Poor choice of $\xi_{out}$ and $\xi_{relay}$ might lead to violation of the constraints 
in the constrained problem defined in (\ref{eqn:constrained_problem_average_cost_with_outage_cost}), or might 
result in a higher mean power per step compared to the optimal mean power per step under the constraints. 
Hence, we need to adapt $\xi_{out}$ and $\xi_{relay}$ as deployment 
progresses. The adaptive algorithms use the structure of the 
optimal policy OptExploreLim.

\end{enumerate}

\section{OptExploreLimLearning: Learning with Explore-Forward, for Given 
$\xi_{out}$ and $\xi_{relay}$}\label{appendix:learning_backtracking_given_xio_xir}

\textbf{Proof of Theorem~\ref{theorem:learning_backtracking_given_xio_xir_general_step_size}:}

Let us denote the shadowing random variable in the link between the potential locations 
 located at distances $i \delta$ and $j \delta$ from the sink node by $W_{i,j}$. 
 The sample space $\Omega$ associated with the deployment process is the collection of all $\omega$ (each $\omega$ 
 corresponds to a fixed realization $\{w_{i,j}: i \geq 0, j \geq 0, i>j, A+1 \leq i-j \leq A+B \}$ 
 of all possible shadowing random variables that might be encountered in the measurement process for 
 deployment up to infinity).  
 Let $\mathcal{F}$ be the Borel $\sigma$-algebra on  $\Omega$. Let $S_k=\sum_{i=1}^k U_i$ be the distance 
 (in steps) of the $k$-th relay from the sink ($S_0:=0$), and 
 $\mathcal{F}_k:=\sigma \bigg(\lambda^{(0)}; W_{i,j}:i \geq 0, j \geq 0, i>j, A+1 \leq i-j \leq A+B, 
 i \leq S_{k-1}+A+B, j \leq S_{k-1}+A+B \bigg)$. The sequence of $\sigma$-algebras $\mathcal{F}_k$ is increasing in $k$, and 
 $\mathcal{F}_k$  captures the history of the deployment process 
up to the deployment of the $k$-th relay.

Note that, 
we can rewrite the update equation in Algorithm~\ref{algorithm:learning_backtracking_given_xio_xir_general_step_size} as 
follows: 

\footnotesize
\begin{equation*}
 \lambda^{(k+1)}=\lambda^{(k)} + a_{k+1}\bigg( f(\lambda^{(k)}) + N_{k+1} \bigg)
\end{equation*}
\normalsize
 
where

\footnotesize
\begin{equation*}
f(\lambda)=\mathbb{E_{\underline{W}}} \min_{u,\gamma}\bigg( \gamma+ \xi_{out} Q_{out}(u,\gamma, W_u)+\xi_{relay} -\lambda u \bigg)
 \end{equation*}
\normalsize

and 

\footnotesize
\begin{eqnarray*}
N_{k+1}&=&\min_{u,\gamma} \bigg( \gamma+ \xi_{out} Q_{out}(u,\gamma, W_u)+\xi_{relay} -\lambda^{(k)} u \bigg) - \\
      && \mathbb{E_{\underline{W}}} \min_{u,\gamma}\bigg( \gamma+ \xi_{out} Q_{out}(u,\gamma, W_u)+\xi_{relay} -\lambda^{(k)} u \bigg)
 \end{eqnarray*}
\normalsize

Note that, $( \gamma+ \xi_{out} Q_{out}(u,\gamma, W_u)+\xi_{relay} -\lambda u )$ is a linearly 
decreasing function in $\lambda$. Hence, 
$ \min_{u,\gamma} ( \gamma+ \xi_{out} Q_{out}(u,\gamma, W_u)+\xi_{relay} -\lambda u )$ is a concave, strictly decreasing 
function in $\lambda$. The function $f(\lambda)$ is a nonnegative linear 
combination of concave, strictly decreasing functions of 
$\lambda$.  Hence, $f(\lambda)$ is strictly decreasing, concave function of $\lambda$ 
for $\lambda \in [0,\infty)$. Hence, $f(\lambda)$ is continuous in $\lambda$. Now, $f(0)>0$ and 
$\lim_{\lambda \rightarrow \infty} f(\lambda)=-\infty$. {\em Hence, $f(\lambda)=0$ will have a unique positive solution.}

Also, if we increase $\lambda$ 
by an amount $\Delta$, then we will have $(A+1)\Delta \leq|f(\lambda+\Delta)-f(\lambda)| \leq (A+B) \Delta$. 
{\em Hence, $f(\cdot)$ is Lipschitz continuous with Lipschitz constant $(A+B)$.}

Let us invoke four conditions from 
Chapter~2 of \cite{borkar08stochastic-approximation-book} as follows:

\begin{enumerate}[label=(\roman{*})]
  \item $f(\cdot)$ is a Lipschitz continuous function.
  \item $\sum_{k=1}^{\infty} a_k = \infty$, $\sum_{k=1}^{\infty} a_k^2 < \infty$.
  \item $\{N_k \}_{k \geq 1}$ is a Martingale difference sequence w.r.t the sigma field $\mathcal{F}_k$ and 
        $\mathbb{E}(|N_{k+1}|^2|\mathcal{F}_k) \leq K (1+|\lambda^{(k)}|^2)$ for some $K>0$.
  \item $\sup_{k \geq 1}|\lambda^{(k)}|< \infty$ almost surely. 
\end{enumerate}

By Theorem~2 (in Chapter~2) of 
\cite{borkar08stochastic-approximation-book}, if the four conditions are satisfied, then $\lambda^{(k)}$ will almost surely 
converge to the unique zero of $f(\cdot)$. But, that unique zero is the optimal average cost per step $\lambda^*$ 
which satisfies $f(\lambda^*)=0$ (by Theorem~\ref{theorem:optimality-condition-required-for-backtracking-average-cost-learning}). 
Hence, the problem reduces to checking the conditions (i)-(iv).

Since $f(\lambda)$ is Lipschitz continuous with Lipschitz constant $(A+B)$, condition~(i) 
is satisfied. Condition~(ii) is satisfied by the choice of $a_k$. 

By definition of $N_k$, we have 
$\mathbb{E}_{\underline{W}} (N_{k+1}| \mathcal{F}_k)=\mathbb{E}_{\underline{W}} (N_{k+1}| \lambda^{(k)})=0$ 
(since shadowing is i.i.d. across links, the shadowing values encountered in the process of measurement for 
placing a new node are independent of the shadowing values encountered in the measurement process for 
deploying the previous nodes) which implies that $\{N_{k+1} \}_{k \geq 1}$ is a Martingale difference 
sequence w.r.t. $\mathcal{F}_k$. 
Now, since the conditional second moment is greater than conditional variance almost surely, we have (almost surely): 

\footnotesize
\begin{eqnarray*}
 \mathbb{E}(|N_{k+1}|^2 |\mathcal{F}_k) &\leq& \mathbb{E} \bigg( \bigg( \min_{u,\gamma} ( \gamma+ \xi_{out} Q_{out}(u,\gamma, W_u) \nonumber\\
&&+\xi_{relay} -\lambda^{(k)} u ) \bigg)^2| \mathcal{F}_k \bigg)
\end{eqnarray*}
\normalsize

Now, we know that $\gamma \leq P_M$, $A+1 \leq u \leq A+B$, outage probability is always in $[0,1]$, and $\xi_{out}$ and $\xi_{relay}$ are fixed. Hence, 
$\mathbb{E}(|N_{k+1}|^2 |\mathcal{F}_k)$ can be upper bounded by $K (1+|\lambda^{(k)}|^2)$ for some $K>0$. 
Hence, condition~(iii) is also satisfied. Condition~(iv) is satisfied by the following lemma.

\begin{lemma}\label{lemma:boundedness_lambda_iteration_optexplorelimlearning}
 For the iterates $\{\lambda^{(k)}\}_{k \geq 1}$ in (\ref{eqn:learning_backtracking_given_xio_xir_update_part}), 
 $\sup_{k \geq 1}|\lambda^{(k)}|< \infty$ almost surely.
\end{lemma}
\begin{proof}
 Let us define $K_0$ to be the smallest integer such that $a_k(A+B)<1$ for all $k \geq K_0$ ($K_0$ exists 
since  $a_k \downarrow 0$). 
For any starting value $\lambda^{(0)}$, it is easy to find a positive real number $d$ (depending on 
the value of $\lambda^{(0)}$) 
such that $\lambda^{(k)} \in [-d,d]$ for all $k \leq K_0$; this is easy to see because 
the node transmit power, node outage probability  
and placement distance for each node are bounded quantities. 

Without loss of generality, we can take $d>P_M+\xi_{out}+\xi_{relay}$ where $P_M$ is the maximum transmit power level 
of a node. We already have that $\lambda^{(k)} \in [-d,d]$ for all $k \leq K_0$. 
Now we will show that $\lambda^{(k)} \in [-d,d]$ for all $k \geq K_0$. To this end, let us assume, as our 
induction hypothesis, that $\lambda^{(k)} \in [-d,d]$ for some $k \geq K_0$. If we can show that 
$\lambda^{(k+1)} \in [-d,d]$, we will be done with the proof.

From the update equation (\ref{eqn:learning_backtracking_given_xio_xir_update_part}), we can write 
that (using $(A+B) \geq u_{k+1} \geq 1$ and $0 \leq a_{k+1}u_{k+1} < 1$):

\footnotesize
\begin{eqnarray*}
 \lambda^{(k+1)} &\leq& \lambda^{(k)}+a_{k+1}(P_M+\xi_{out}+\xi_{relay}-\lambda^{(k)}u_{k+1}) \\
&=& (1-a_{k+1}u_{k+1})\lambda^{(k)}+a_{k+1}(P_M+\xi_{out}+\xi_{relay}) \\
&\leq& (1-a_{k+1}u_{k+1})\lambda^{(k)}+a_{k+1}u_{k+1}(P_M+\xi_{out}+\xi_{relay}) \\
                  &\leq& \max \{ \lambda^{(k)}, P_M+\xi_{out}+\xi_{relay} \} \\
                   &\leq&  d
\end{eqnarray*}
\normalsize

On the other hand:

\footnotesize
\begin{eqnarray*}
 \lambda^{(k+1)} & \geq & \lambda^{(k)}+a_{k+1}(0-\lambda^{(k)}u_{k+1}) \\
            &=& (1-a_{k+1}u_{k+1})\lambda^{(k)} \\
            &\geq & -(1-a_{k+1}u_{k+1})d \\
            &\geq& -d
\end{eqnarray*}
\normalsize

Hence, $\lambda^{(k+1)} \in [-d,d]$ and the lemma is proved.
\end{proof}

Now, since conditions~(i)-(iv) are satisfied, 
by Theorem~2, Chapter~2 of \cite{borkar08stochastic-approximation-book}, 
$\lambda^{(k)} \rightarrow \lambda^*$ almost surely.\qed

\textbf{Proof of Corollary~\ref{corollary:learning_backtracking_given_xio_xir_special_step_size}:} 
Suppose that, we choose $a_k=\frac{1}{\sum_{i=1}^k u_i}$ 
in Algorithm~\ref{algorithm:learning_backtracking_given_xio_xir_general_step_size}. Then, $\sum_{k=1}^{\infty}a_k \geq \sum_{k=1}^{\infty} \frac{1}{k(A+B)×}=\infty$ almost 
surely and $\sum_{k=1}^{\infty}a_k^2 \leq \sum_{k=1}^{\infty} \frac{1}{k^2(A+1)^2}< \infty$ almost surely. 

Now, with this step size, 

\footnotesize
\begin{eqnarray*}
 \lambda_1 &=& \lambda^{(0)}+a_1 (\gamma_1+\xi_{out}Q_{out}^{(1,0)}+\xi_{relay}-\lambda^{(0)} u_1) \\
           &=& \lambda^{(0)}+ \frac{1}{u_1} (\gamma_1+\xi_{out}Q_{out}^{(1,0)}+\xi_{relay}-\lambda^{(0)} u_1) \\
           &=& \frac{\gamma_1+\xi_{out}Q_{out}^{(1,0)}+\xi_{relay}}{u_1}
\end{eqnarray*}
\normalsize

and, in general, 

\footnotesize
\begin{eqnarray*}
 \lambda^{(k+1)} &=& \lambda^{(k)}+a_{k+1} (\gamma_{k+1}+\xi_{out}Q_{out}^{(k+1,k)}+\xi_{relay}-\lambda^{(k)} u_{k+1}) \\
           &=& \lambda^{(k)}+\frac{(\gamma_{k+1}+\xi_{out}Q_{out}^{(k+1,k)}+\xi_{relay}-\lambda^{(k)} u_{k+1})}{\sum_{i=1}^{k+1} u_i}  \\
           &=& \frac{\lambda^{(k)} \sum_{i=1}^k u_i +(\gamma_{k+1}+\xi_{out}Q_{out}^{(k+1,k)}+\xi_{relay})}{\sum_{i=1}^{k+1} u_i} \\
           &=&  \frac{ \sum_{i=1}^{k+1} (\gamma_i+\xi_{out}Q_{out}^{(i,i-1)}+\xi_{relay})}{\sum_{i=1}^{k+1} u_i} \\
\end{eqnarray*}
\normalsize

Hence, in (\ref{eqn:learning_backtracking_given_xio_xir_update_part}) 
of Algorithm~\ref{algorithm:learning_backtracking_given_xio_xir_general_step_size}, we can 
replace $\lambda^{(k)}=\frac{\sum_{i=1}^k (\gamma_i+\xi_{out}Q_{out}^{(i,i-1)}+\xi_{relay})}{\sum_{i=1}^k u_i}$, and 
this proves the theorem.

\section{OptExploreLimAdaptiveLearning with Constraint on Outage Probability and Relay Placement Rate}
\label{appendix:learning_backtracking_adaptive_with_outage_cost}

\subsection{\textbf{Proof of Theorem~\ref{theorem:structure_of_mathcal_K_q_n}}}
\label{subsection:proof_of_structure_of_mathcal_K_q_n}

{\em Proof of the first statement:} Let us assume that $\pi^*(\xi_{out}^*,\xi_{relay}^*)$ 
satisfies both constraints in (\ref{eqn:constrained_problem_average_cost_with_outage_cost}) with equality for some 
$\xi_{out}^*>0$, $\xi_{relay}^*>0$, i.e., $\pi^*(\xi_{out}^*,\xi_{relay}^*)$ is an optimal policy 
for problem (\ref{eqn:constrained_problem_average_cost_with_outage_cost}). Now, let us assume that 
there exists $\xi_{out}'\geq 0,\xi_{relay}' \geq 0$ satisfying (i) 
  $(\lambda^*(\xi_{out}',\xi_{relay}'),\xi_{out}',\xi_{relay}') \in \mathcal{K}(\overline{q},\overline{N})$,  
  and (ii) $\frac{\overline{Q}_{out}^*(\xi_{out}',\xi_{relay}')}{\overline{U}^*(\xi_{out}',\xi_{relay}')} < \overline{q}$. 
We will show that this leads to a contradiction. 

Let us consider the problem of minimizing the mean outage per step subject to a constraint 
$\frac{\overline{\Gamma}^*(\xi_{out}^*,\xi_{relay}^*)}{\overline{U}^*(\xi_{out}^*,\xi_{relay}^*)}$ on the 
mean power per step and a constraint 
$\frac{1}{\overline{U}^*(\xi_{out}^*,\xi_{relay}^*)}=\overline{N}$ on the mean number of relays per step. 
Clearly, by Theorem~\ref{theorem:how-to-choose-optimal-Lagrange-multiplier}, 
$\pi^*(\xi_{out}^*,\xi_{relay}^*)$ is an optimal policy for this problem since it satisfies both constraints with equality. 
Note that, $\pi^*(\xi_{out}^*,\xi_{relay}^*)$ has a mean outage per step $\overline{q}$. But, we also see that 
the policy $\pi^*(\xi_{out}',\xi_{relay}')$ has the same mean power per step and a smaller mean number of relays per step 
compared to $\pi^*(\xi_{out}^*,\xi_{relay}^*)$ 
(since $(\lambda^*(\xi_{out}',\xi_{relay}'),\xi_{out}',\xi_{relay}') \in \mathcal{K}(\overline{q},\overline{N})$), 
and has a {\em strictly} smaller mean outage per step 
compared to $\pi^*(\xi_{out}^*,\xi_{relay}^*)$. This leads to a contradiction since 
$\pi^*(\xi_{out}^*,\xi_{relay}^*)$ is an optimal policy for the  problem of 
minimizing the mean outage per step subject to a constraint 
$\frac{\overline{\Gamma}^*(\xi_{out}^*,\xi_{relay}^*)}{\overline{U}^*(\xi_{out}^*,\xi_{relay}^*)}$ on the 
mean power per step and a constraint 
$\frac{1}{\overline{U}^*(\xi_{out}^*,\xi_{relay}^*)}=\overline{N}$ on the mean 
number of relays per step. 

Similarly, we can show a contradiction if, instead of assuming 
$\frac{\overline{Q}_{out}^*(\xi_{out}',\xi_{relay}')}{\overline{U}(\xi_{out}',\xi_{relay}')} < \overline{q}$, we had assumed 
$\frac{1}{\overline{U}(\xi_{out}',\xi_{relay}')} < \overline{N}$.

Hence, the first statement is proved.

{\em Proof of the second statement:} This statement follows from the fact that for any $\xi_{relay}\geq 0$, 
$\pi^*(0,\xi_{relay})$ always places at a distance $(A+B)$ and uses the smallest power $P_1$, thereby 
incurring a mean outage per step equal to 
$\frac{\mathbb{E}_{W}Q_{out}(A+B,P_1,W)}{A+B}$.\qed

\begin{figure*}[!t]
\begin{footnotesize}
\begin{eqnarray}
 \lambda^{(k)} &=&  \lambda^{(k-1)}+ a_k  \min_{u, \gamma } \bigg( \gamma + \xi_{out}^{(k-1)} Q_{out}(u,\gamma,w_u) + \xi_{relay}^{(k-1)}-\lambda^{(k-1)} u \bigg)\nonumber\\
\xi_{out}^{(k)}&=& \xi_{out}^{(k-1)}+ b_k \lim_{\beta \downarrow 0} \frac{\Lambda_{[0,A_2]}\bigg(\xi_{out}^{(k-1)}+ \beta (Q_{out}(u_k,\gamma_k, w_{u_k})-\overline{q}u_k) \bigg)-\xi_{out}^{(k-1)}}{\beta} + o(b_k)\nonumber\\
&=& \xi_{out}^{(k-1)}+ a_k \bigg(  \frac{b_k}{a_k} \bigg( \lim_{\beta \downarrow 0} \frac{\Lambda_{[0,A_2]}\bigg(\xi_{out}^{(k-1)}+ \beta (Q_{out}(u_k,\gamma_k, w_{u_k})-\overline{q}u_k) \bigg)-\xi_{out}^{(k-1)}}{\beta}+ \frac{o(b_k)}{b_k} \bigg) \bigg) \nonumber\\
\xi_{relay}^{(k)}&=& \xi_{relay}^{(k-1)} +b_k \lim_{\beta \downarrow 0} \frac{\Lambda_{[0,A_3]}\bigg(\xi_{relay}^{(k-1)} + \beta (1-\overline{N}u_k) \bigg)-\xi_{relay}^{(k-1)}}{\beta} + o(b_k) \nonumber\\
&=& \xi_{relay}^{(k-1)} +a_k \bigg(  \frac{b_k}{a_k} \bigg( \lim_{\beta \downarrow 0} \frac{\Lambda_{[0,A_3]}\bigg(\xi_{relay}^{(k-1)} + \beta (1-\overline{N}u_k) \bigg)-\xi_{relay}^{(k-1)}}{\beta} + \frac{o(b_k)}{b_k} \bigg) \bigg)
\label{eqn:convergence_learning_backtracking_adaptive_with_outage_cost_first_approximation}
\end{eqnarray}
\end{footnotesize}
\hrule
\end{figure*}

\subsection{\textbf{Proof of Theorem~\ref{theorem:placement_rate_mean_outage_per_step_continuous_in_xio_and_xir}}}
\label{subsection:proof_of_placement_rate_mean_outage_per_step_continuous_in_xio_and_xir}

Denote by $g(r, \gamma), 
r \in \{A+1,A+2,\cdots,A+B\}, \gamma \in \mathcal{S}$ 
the joint distribution of $(U_k, \Gamma_k)$ when $\lambda$ in (\ref{eqn:smdp-optimal-policy_general_form}) 
is replaced by $\lambda^*(\xi_{out}, \xi_{relay})$, i.e., when deployment is done using the OptExploreLim algorithm 
(Algorithm~\ref{algorithm:policy_structure_smdp_backtracking}).

Let us assume that $g(r,\gamma)$ is continuous in both $\xi_{out}$ and $\xi_{relay}$ (we will prove this 
assertion in Lemma~\ref{lemma:grgamma_continuous_in_xio_xir} at the end of the proof of the theorem). 
By Lemma~\ref{lemma:grgamma_continuous_in_xio_xir}, 
the mean placement distance $\overline{U}^*(\xi_{out},\xi_{relay})= \sum_{r=A+1}^{A+B}\sum_{\gamma \in \mathcal{S}} r g (r,\gamma)$ 
is continuous in $\xi_{out}$ and $\xi_{relay}$. Similarly, the mean power per link 
$\overline{\Gamma}^*(\xi_{out},\xi_{relay})=\sum_{r=A+1}^{A+B}\sum_{\gamma \in \mathcal{S}} \gamma g (r,\gamma)$ 
is continuous in $\xi_{out}$ and $\xi_{relay}$. 

Let us denote by $\lambda^*(\xi_{out}, \xi_{relay})$ the optimal average cost per step for the problem in 
(\ref{eqn:unconstrained_problem_average_cost_with_outage_cost}), for given $\xi_{out}$ and $\xi_{relay}$. 
By Renewal-Reward Theorem, 

\footnotesize
\begin{eqnarray*}
\lambda^*(\xi_{out}, \xi_{relay}) = \frac{ \overline{\Gamma}^*(\xi_{out},\xi_{relay})+\xi_{out}\overline{Q}_{out}^*(\xi_{out},\xi_{relay})+\xi_{relay}  }{\overline{U}^*(\xi_{out},\xi_{relay})}
\end{eqnarray*}
\normalsize

Since $\lambda^*(\xi_{out}, \xi_{relay})$ is continuous in $\xi_{out}$ and $\xi_{relay}$ 
(by Theorem~\ref{theorem:smdp-cost-vs-xi}), we conclude that $\overline{Q}_{out}^*(\xi_{out},\xi_{relay})$ 
is continuous in $\xi_{out}$ and $\xi_{relay}$. Hence, $\frac{\overline{\Gamma}^*(\xi_{out},\xi_{relay})}{\overline{U}^*(\xi_{out},\xi_{relay})}$, 
$\frac{\overline{Q}_{out}^*(\xi_{out},\xi_{relay})}{\overline{U}^*(\xi_{out},\xi_{relay})}$ and $\frac{1}{\overline{U}^*(\xi_{out},\xi_{relay})}$ are 
continuous in $\xi_{out}$ and $\xi_{relay}$. Hence, the theorem is proved. \qed

\begin{lemma}\label{lemma:grgamma_continuous_in_xio_xir}
 Under Assumption~\ref{assumption:shadowing_continuous_random_variable}, 
$g(r,\gamma)$ is continuous in $\xi_{out}$ and $\xi_{relay}$.
\end{lemma}
\begin{proof}
Let us fix any $r \in \{A+1,\cdots,A+B\}$ and any $\gamma \in \mathcal{S}$. 
We will show that $g(r,\gamma)$ is continuous in $\xi_{out}$. The continuity of $g(r,\gamma)$ w.r.t. 
$\xi_{relay}$ will follow the same line of arguments. 

Consider any sequence $\{\xi_n\}_{n \geq 1}$ such that $\xi_n \rightarrow \xi_{out}$. Let us 
denote the joint probability distribution of placement distance and node transmit power 
by $g_n(r,\gamma)$, if the cost per unit outage is $\xi_n$ and if OptExploreLim is used in the 
deployment process. We will show that $g_n(r,\gamma) \rightarrow g(r,\gamma)$ as $n \rightarrow \infty$.

Let us define the sets 
$\mathcal{E}_{\gamma'}=\bigg\{\underline{w}: \gamma+\xi_{out}Q_{out}(r,\gamma,w_r) < \gamma'+\xi_{out}Q_{out}(r,\gamma',w_r) \bigg\}$
and $\mathcal{E}_{u,\gamma'}=\bigg\{\underline{w}: \gamma+\xi_{out}Q_{out}(r,\gamma,w_r)
  +\xi_{relay}-\lambda^*(\xi_{out},\xi_{relay}) r  <  \gamma'+\xi_{out}Q_{out}(u,\gamma',w_u)  
  +\xi_{relay}-\lambda^*(\xi_{out},\xi_{relay}) u  \bigg\}$. 
  
  In state $\underline{w}$, the OptExploreLim algorithm (Algorithm~\ref{algorithm:policy_structure_smdp_backtracking}) 
  will place 
the next relay at distance $r$ and decide power level $\gamma$ if $\underline{w} \in \mathcal{E}_{\gamma'}$ 
for all $\gamma' \neq \gamma, \gamma' \in \mathcal{S}$ and if 
$\underline{w} \in \mathcal{E}_{u,\gamma'}$ for all $u \neq r, \gamma' \in \mathcal{S}$.

Let us  define $\mathcal{E}=\cap_{\gamma' \neq \gamma}\mathcal{E}_{\gamma'}  
\cap_{u \neq r, \gamma' \in \mathcal{S}}  \mathcal{E}_{u,\gamma'}$. 

Note that, $g(r,\gamma)=\mathbb{P}(\mathcal{E})=\mathbb{E}(\mathbb{I}_{\mathcal{E}})$, where $\mathbb{I}$ denotes 
the indicator function, and the expectation is over the joint distribution of the shadowing vector $\underline{W}$ 
(shadowing random variables from $B$ locations).

Now, for any $\gamma' \neq \gamma$, 
we have $\mathbb{P}\bigg( \gamma+\xi_{out}Q_{out}(r,\gamma,W_r) = \gamma'+\xi_{out}Q_{out}(r,\gamma',W_r) \bigg)=0$. 
Also, $\mathbb{P}\bigg( \gamma+\xi_{out}Q_{out}(r,\gamma,W_r)
  +\xi_{relay}-\lambda^*(\xi_{out},\xi_{relay}) r  =   \gamma'+\xi_{out}Q_{out}(u,\gamma',W_u)  
  +\xi_{relay}-\lambda^*(\xi_{out},\xi_{relay}) u  \bigg)=0$ if $\gamma' \in \mathcal{S}$, $u \neq r$. 
These two assertions follow from Assumption~\ref{assumption:shadowing_continuous_random_variable} and 
the fact that $Q_{out}(r,\gamma,w)$ is continuous in $w$. 
Hence, we discard these zero probability events in our analysis and safely assume that:

\begin{itemize}
 \item For $\gamma' \neq \gamma$, the complement $\overline{\mathcal{E}_{\gamma'}}$ has the same expression as 
$\mathcal{E}_{\gamma'}$ except that the $<$ sign is replaced by $>$ sign.
\item For $\gamma' \in \mathcal{S}$, $u \neq r$, $\overline{\mathcal{E}_{u,\gamma'}}$ 
has the same expression as 
$\mathcal{E}_{u,\gamma'}$ except that the $<$ sign is replaced by $>$ sign.
\end{itemize}

Now, consider any sequence $\{\xi_n\}_{n \geq 1}$ such that $\xi_n \rightarrow \xi_{out}$. 
Let $\mathcal{E}_{\gamma'}^{(n)}$, $\mathcal{E}_{u,\gamma'}^{(n)}$ and 
$\mathcal{E}^{(n)}$ be the sets obtained where we replace $\xi_{out}$ by $\xi_n$ in the expressions of the sets 
$\mathcal{E}_{\gamma'}$, $\mathcal{E}_{u,\gamma'}$ and $\mathcal{E}$ respectively. 
Clearly, we can make similar claims for  $\mathcal{E}_{\gamma'}^{(n)}$, $\mathcal{E}_{u,\gamma'}^{(n)}$ for 
any $n \geq 1$.

Recall that, $g(r,\gamma)=\mathbb{P}(\mathcal{E})=\mathbb{E}(\mathbb{I}_{\mathcal{E}})$.  
Clearly, if we can show that $\mathbb{E}(\mathbb{I}_{\mathcal{E}^{(n)}}) \rightarrow \mathbb{E}(\mathbb{I}_{\mathcal{E}})$, 
the lemma will be proved.

\begin{claim}\label{claim:convergence_of_indicator_function}
 $\mathbb{I}_{ \mathcal{E}_{u,\gamma'}^{(n)} } \rightarrow \mathbb{I}_{ \mathcal{E}_{u,\gamma'}} $ 
almost surely as $n \rightarrow \infty$, for $u \neq r, \gamma' \in \mathcal{S}$. 
Also, $\mathbb{I}_{ \mathcal{E}_{\gamma'}^{(n)} } \rightarrow \mathbb{I}_{ \mathcal{E}_{\gamma'} } $ 
almost surely as $n \rightarrow \infty$, for $\gamma' \neq \gamma$.
\end{claim}
\begin{proof}
 Suppose that, for some value of $\underline{w}$, 
$\mathbb{I}_{ \mathcal{E}_{u,\gamma'}}(\underline{w})=1$, i.e., 
$\gamma+\xi_{out}Q_{out}(r,\gamma,w_r)
  +\xi_{relay}-\lambda^*(\xi_{out},\xi_{relay}) r  <  \gamma'+\xi_{out}Q_{out}(u,\gamma',w_u)  
  +\xi_{relay}-\lambda^*(\xi_{out},\xi_{relay}) u $. 
Now, by Theorem~\ref{theorem:smdp-cost-vs-xi}, $\lambda^*(\xi_{out},\xi_{relay})$ is  continuous in 
$\xi_{out}$ and $\xi_{relay}$. Hence, there exists an integer $n_0$ sufficiently large such that 
for all $n > n_0$, we have $\gamma+\xi_n Q_{out}(r,\gamma,w_r)
  +\xi_{relay}-\lambda^*(\xi_n,\xi_{relay}) r  <  \gamma'+\xi_n Q_{out}(u,\gamma',w_u)  
  +\xi_{relay}-\lambda^*(\xi_n,\xi_{relay}) u $, i.e., 
$\mathbb{I}_{ \mathcal{E}_{u,\gamma'}^{(n)} } (\underline{w})=1$ for all $n > n_0$. Hence, 
$\mathbb{I}_{ \mathcal{E}_{u,\gamma'}^{(n)} } (\underline{w}) \rightarrow \mathbb{I}_{ \mathcal{E}_{u,\gamma'} } 
(\underline{w}) $ if $\mathbb{I}_{ \mathcal{E}_{u,\gamma'} } 
(\underline{w})=1$. Similar argument works when $\mathbb{I}_{ \mathcal{E}_{u,\gamma'} } (\underline{w})=0$. Hence, 
the first part of the claim is proved.

The second part of the claim is proved in a similar way.
\end{proof}

Note that, $\mathbb{I}_{\mathcal{E}^{(n)}}=\prod_{\gamma' \neq \gamma} \mathbb{I}_{\mathcal{E}_{\gamma'}^{(n)}} 
\prod_{u \neq r, \gamma' \in \mathcal{S}} \mathbb{I}_{\mathcal{E}_{u,\gamma'}^{(n)}}$. 
By Claim~\ref{claim:convergence_of_indicator_function}, $\mathbb{I}_{\mathcal{E}^{(n)}} \rightarrow \mathbb{I}_{\mathcal{E}}$ 
almost surely. Since indicator functions always take values in the set $\{0,1\}$, 
we have $\mathbb{E}(\mathbb{I}_{\mathcal{E}^{(n)}}) \rightarrow \mathbb{E}(\mathbb{I}_{\mathcal{E}})$ by 
Dominated Convergence Theorem. 

Hence, the lemma is proved.
\end{proof}

\subsection{\textbf{Proof of Theorem~\ref{theorem:convergence_learning_backtracking_adaptive_with_outage_cost_algorithm}}}
\label{subsection:convergence_proof_of_optexplorelimadaptive_learning}

Let us denote the shadowing random variable in the link between the potential locations 
 located at distances $i \delta$ and $j \delta$ from the sink node by $W_{i,j}$. 
 The sample space $\Omega$ associated with the deployment process is the collection of all $\omega$ (each $\omega$ 
 corresponds to a fixed realization $\{w_{i,j}: i \geq 0, j \geq 0, i>j, A+1 \leq i-j \leq A+B \}$ 
 of all possible shadowing random variables that might be encountered in the measurement process for 
 deployment up to infinity).  
 Let $\mathcal{F}$ be the Borel $\sigma$-algebra on  $\Omega$. Let $S_k=\sum_{i=1}^k U_i$ be the distance 
 (in steps) of the $k$-th relay from the sink ($S_0:=0$), and 
 $\mathcal{F}_k:=\sigma \bigg((\lambda^{(0)},\xi_{out}^{(0)},\xi_{relay}^{(0)}); W_{i,j}:i \geq 0, j \geq 0, i>j, A+1 \leq i-j \leq A+B, 
 i \leq S_{k-1}+A+B, j \leq S_{k-1}+A+B \bigg)$. The sequence of $\sigma$-algebras $\mathcal{F}_k$ is increasing in $k$, and 
 $\mathcal{F}_k$  captures the history of the deployment process 
up to the deployment of the $k$-th relay.

Let us recall the outline of the proof of 
Theorem~\ref{theorem:convergence_learning_backtracking_adaptive_with_outage_cost_algorithm} in 
Section~\ref{subsection:learning_backtracking_adaptive_with_outage_cost_algorithm}.

\subsubsection{\textbf{Almost sure boundedness of the $\lambda^{(k)}$ iterates}}
\label{subsubsection:proof_of_boundedness_lambda_iteration_optexplorelimadaptivelearning}

\begin{lemma}\label{lemma:boundedness_lambda_iteration_optexplorelimadaptivelearning}
 The iterates $\{\lambda^{(k)}\}_{k \geq 1}$ in (\ref{eqn:learning_backtracking_adaptive_with_outage_cost_update_part}) 
are bounded almost surely.
\end{lemma}
\begin{proof}
 Let us define $K_0$ to be the smallest integer such that $a_k(A+B)<1$ for all $k \geq K_0$ ($K_0$ exists 
since $a_k \downarrow 0$). 
For any starting value $\lambda^{(0)}$, it is easy to find a positive real number $d$ (depending on 
the value of $\lambda^{(0)}$) 
such that $\lambda^{(k)} \in [-d,d]$ for all $k \leq K_0$; this is easy to see because 
$\xi_{out}^{(k)} \in [0,A_2]$, $\xi_{relay}^{(k)} \in [0,A_3]$ for all $k$, 
and the node transmit power, node outage probability  
and placement distance for each node are bounded quantities. 

Without loss of generality, we can take $d>P_M+A_2+A_3$ where $P_M$ is the maximum transmit power level 
of a node. We already have that $\lambda^{(k)} \in [-d,d]$ for all $k \leq K_0$. 
Now we will show that $\lambda^{(k)} \in [-d,d]$ for all $k \geq K_0$. To this end, let us assume, as 
our induction hypothesis, that 
$\lambda^{(k)} \in [-d,d]$ for some $k \geq K_0$. If we can show that 
$\lambda^{(k+1)} \in [-d,d]$, we will be done with the proof.

From the update equation (\ref{eqn:learning_backtracking_adaptive_with_outage_cost_update_part}), we can write 
that (using $(A+B) \geq u_{k+1} \geq 1$ and $0 \leq a_{k+1}u_{k+1} < 1$):  

\footnotesize
\begin{eqnarray*}
 \lambda^{(k+1)} &\leq& \lambda^{(k)}+a_{k+1}(P_M+A_2+A_3-\lambda^{(k)}u_{k+1}) \\
&=& (1-a_{k+1}u_{k+1})\lambda^{(k)}+a_{k+1}(P_M+A_2+A_3) \\
&\leq& (1-a_{k+1}u_{k+1})\lambda^{(k)}+a_{k+1}u_{k+1}(P_M+A_2+A_3) \\
                  &\leq& \max \{ \lambda^{(k)}, P_M+A_2+A_3 \} \\
                   &\leq&  d
\end{eqnarray*}
\normalsize

On the other hand:

\footnotesize
\begin{eqnarray*}
 \lambda^{(k+1)} & \geq & \lambda^{(k)}+a_{k+1}(0-\lambda^{(k)}u_{k+1}) \\
            &=& (1-a_{k+1}u_{k+1})\lambda^{(k)} \\
            &\geq & -(1-a_{k+1}u_{k+1})d \\
            &\geq& -d
\end{eqnarray*}
\normalsize

Hence, $\lambda^{(k+1)} \in [-d,d]$ and the lemma is proved.
\end{proof}

\begin{figure*}[!t]
\begin{footnotesize}
\begin{eqnarray}
 \xi_{out}^{(k)} &=& \Lambda_{\mathcal{G}} \bigg( \xi_{out}^{(k-1)}+b_k \bigg( Q_{out}(U_k,\Gamma_k,W_{U_k})-\overline{q}U_k \bigg) \bigg) \nonumber\\
 &=& \Lambda_{\mathcal{G}} \bigg( \xi_{out}^{(k-1)}+b_k \bigg( \underbrace{\overline{Q}_{out}^*(\xi_{out}^{(k-1)}, \xi_{relay}^{(k-1)})-\overline{q}\overline{U}^*(\xi_{out}^{(k-1)}, \xi_{relay}^{(k-1)})}_{:=f_1(\xi_{out}^{(k-1)},\xi_{relay}^{(k-1)})}        \nonumber\\
 && +   \underbrace{\overline{Q}_{out}(\lambda^{(k-1)},\xi_{out}^{(k-1)}, \xi_{relay}^{(k-1)})-\overline{q}\overline{U}(\lambda^{(k-1)},\xi_{out}^{(k-1)}, \xi_{relay}^{(k-1)})-f_1(\xi_{out}^{(k-1)},\xi_{relay}^{(k-1)})}_{:=g_1(\lambda^{(k-1)},\xi_{out}^{(k-1)},\xi_{relay}^{(k-1)})}    \nonumber\\
 &&  +  \underbrace{Q_{out}(U_k,\Gamma_k,W_{U_k})-\overline{q}U_k-\bigg( \overline{Q}_{out}(\lambda^{(k-1)},\xi_{out}^{(k-1)}, \xi_{relay}^{(k-1)})-\overline{q}\overline{U}(\lambda^{(k-1)},\xi_{out}^{(k-1)}, \xi_{relay}^{(k-1)}) \bigg) }_{:=M_1^{(k)}}   \bigg)      \bigg) \nonumber\\
 &=& \Lambda_{\mathcal{G}} \bigg( \xi_{out}^{(k-1)}+b_k \bigg( f_1(\xi_{out}^{(k-1)},\xi_{relay}^{(k-1)})+ g_1(\lambda^{(k-1)},\xi_{out}^{(k-1)},\xi_{relay}^{(k-1)})+ M_{1}^{(k)} \bigg) \bigg) \nonumber\\
\xi_{relay}^{(k)}&=& \Lambda_{\mathcal{G}} \bigg( \xi_{out}^{(k-1)}+b_k \bigg( 1-\overline{N}U_k \bigg) \bigg) \nonumber\\   
&=& \Lambda_{\mathcal{G}} \bigg( \xi_{relay}^{(k-1)}+b_k \bigg( \underbrace{1-\overline{N}\overline{U}^*(\xi_{out}^{(k-1)}, \xi_{relay}^{(k-1)})}_{:=f_2(\xi_{out}^{(k-1)},\xi_{relay}^{(k-1)})}        \nonumber\\
 && +   \underbrace{1-\overline{N}\overline{U}(\lambda^{(k-1)},\xi_{out}^{(k-1)}, \xi_{relay}^{(k-1)})-f_2(\xi_{out}^{(k-1)},\xi_{relay}^{(k-1)})}_{:=g_2(\lambda^{(k-1)},\xi_{out}^{(k-1)},\xi_{relay}^{(k-1)})}    \nonumber\\
 &&  +  \underbrace{ 1-\overline{N}U_k-\bigg( 1-\overline{N}\overline{U}(\lambda^{(k-1)},\xi_{out}^{(k-1)}, \xi_{relay}^{(k-1)})  \bigg) }_{:=M_2^{(k)}}     \bigg)    \bigg) \nonumber\\
&=& \Lambda_{\mathcal{G}} \bigg( \xi_{relay}^{(k-1)}+b_k \bigg( f_2(\xi_{relay}^{(k-1)},\xi_{relay}^{(k-1)})+ g_2(\lambda^{(k-1)},\xi_{relay}^{(k-1)},\xi_{relay}^{(k-1)})+ M_{2}^{(k)} \bigg) \bigg)
\label{eqn:optexplorelimadaptivelearning_slower_timescale_in_kushner_form}
\end{eqnarray}
\end{footnotesize}
\hrule
\end{figure*}

\subsubsection{\textbf{Analyzing the Faster Time-Scale Iteration of $\lambda^{(k)}$}} 
\label{subsubsection:lemma_used_in_convergence_proof_of_optexplorelimadaptive_learning}

Let us denote by $\lambda^*(\xi_{out}, \xi_{relay})$ the optimal average cost per step for the problem in 
(\ref{eqn:unconstrained_problem_average_cost_with_outage_cost}), for given $\xi_{out}$ and $\xi_{relay}$. 

\begin{lemma}\label{lemma:adaptive_learning_faster_timescale_convergence}
\label{lemma:lemma_used_in_convergence_proof_of_optexplorelimadaptive_learning}
 For Algorithm~\ref{algorithm:learning_backtracking_adaptive_with_outage_cost_algorithm}, we have 
 $(\lambda^{(k)},\xi_{out}^{(k)},\xi_{relay}^{(k)}) \rightarrow \{(\lambda^*(\xi_{out},\xi_{relay}),\xi_{out},\xi_{relay}): 
(\xi_{out},\xi_{relay}) \in [0,A_2] \times [0,A_3] \}$ and 
$\lim_{k \rightarrow \infty}|\lambda^{(k)}-\lambda^*(\xi_{out}^{(k)},\xi_{relay}^{(k)})|=0$ almost surely.
\end{lemma}

\begin{proof}
We follow the proof of Lemma~$1$, Chapter~$6$ of \cite{borkar08stochastic-approximation-book}.

Using the first order Taylor series expansion of the function $\Lambda_{[0,A_2]}(\cdot)$, and 
using the fact that $\Lambda_{[0,A_2]}(\xi_{out}^{(k-1)})=\xi_{out}^{(k-1)}$ (since $\xi_{out}^{(k-1)} \in [0,A_2]$), 
the update equation (\ref{eqn:learning_backtracking_adaptive_with_outage_cost_update_part}) can be 
rewritten as (\ref{eqn:convergence_learning_backtracking_adaptive_with_outage_cost_first_approximation}).

Consider the update equation for $\xi_{relay}$ in 
(\ref{eqn:convergence_learning_backtracking_adaptive_with_outage_cost_first_approximation}). 
Note that:

\footnotesize
\begin{eqnarray*}
&& \lim_{\beta \downarrow 0} \frac{\Lambda_{[0,A_3]}\bigg(\xi_{relay}^{(k-1)} + \beta (1-\overline{N}u_k) \bigg)-\xi_{relay}^{(k-1)}}{\beta} \\
&=& (1-\overline{N}u_k) \mathbb{I} \{0< \xi_{relay}^{(k-1)} <A_3\}\\
&+& (1-\overline{N}u_k)^+ \mathbb{I} \{\xi_{relay}^{(k-1)} =0 \} \\
&-& (1-\overline{N}u_k)^- \mathbb{I} \{\xi_{relay}^{(k-1)} =A_3 \} \\
\end{eqnarray*}
\normalsize

where $x^+=\max\{x,0\}$ and $x^-=-\min \{x,0\}$. 
A similar expression holds for the  
$\xi_{out}^{(k)}$ update. Since $Q_{out}(\cdot,\cdot,\cdot)$ and $u_k$ are bounded quantities, and since 
$\lim_{k \rightarrow 0}\frac{b_k}{a_k}=0$, we have:

\footnotesize
\begin{eqnarray*}
&& \lim_{k \rightarrow \infty}  \bigg(\frac{b_k}{a_k} \bigg( \lim_{\beta \downarrow 0} \bigg(\Lambda_{[0,A_2]}\bigg(\xi_{out}^{(k-1)}+ \beta (Q_{out}(u_k,\gamma_k, w_{u_k})-\overline{q}u_k) \bigg) \\
&& -\xi_{out}^{(k-1)}\bigg)/ \beta  + \frac{o(b_k)}{b_k} \bigg) \bigg)=0
\end{eqnarray*}
\normalsize

and
\footnotesize
\begin{eqnarray*}
&& \lim_{k \rightarrow \infty} \bigg( \frac{b_k}{a_k} \bigg( \lim_{\beta \downarrow 0} \frac{\Lambda_{[0,A_3]}\bigg(\xi_{relay}^{(k-1)} + \beta (1-\overline{N}u_k) \bigg)-\xi_{relay}^{(k-1)}}{\beta} \\
&& + \frac{o(b_k)}{b_k} \bigg) \bigg)=0
\end{eqnarray*}
\normalsize

Now, note that, the function 
$f(\lambda,\xi_{out},\xi_{relay})=\mathbb{E_{\underline{W}}} \min_{u,\gamma}\bigg( \gamma+ \xi_{out} Q_{out}(u,\gamma, W_u)+\xi_{relay} -\lambda u \bigg)$ 
is Lipschitz continuous in all arguments, and the 
o.d.e. $\dot{\lambda}(t)=f(\lambda(t),\xi_{out},\xi_{relay})$ has a unique 
globally asymptotically stable equilibrium $\lambda^*(\xi_{out},\xi_{relay})$ for any 
$\xi_{out} \geq 0$, $\xi_{relay} \geq 0$ (see the proof of 
Theorem~\ref{theorem:learning_backtracking_given_xio_xir_general_step_size}). The quantity $\lambda^*(\xi_{out},\xi_{relay})$ 
is Lipschitz continuous in $\xi_{out}$ and $\xi_{relay}$. Also by 
Lemma~\ref{lemma:boundedness_lambda_iteration_optexplorelimadaptivelearning} and the 
projection operation in the slower timescale, the iterates are bounded almost surely.

Hence, by a similar argument as in the proof of Lemma~$1$, Chapter~$6$ of \cite{borkar08stochastic-approximation-book}, 
and by using Theorem~\ref{theorem:optimality-condition-required-for-backtracking-average-cost-learning} 
and Theorem~\ref{theorem:learning_backtracking_given_xio_xir_general_step_size}, 
$(\lambda^{(k)},\xi_{out}^{(k)},\xi_{relay}^{(k)})$ 
converges to the internally chain transitive invariant sets of the o.d.e. 
$\dot{\lambda}(t)=f(\lambda(t),\xi_{out}(t),\xi_{relay}(t))$, 
$\dot{\xi}_{out}(t)=0$, $\dot{\xi}_{relay}(t)=0$. Hence, 
$(\lambda^{(k)},\xi_{out}^{(k)},\xi_{relay}^{(k)}) \rightarrow \{(\lambda^*(\xi_{out},\xi_{relay}),\xi_{out},\xi_{relay}): 
(\xi_{out},\xi_{relay}) \in [0,A_2] \times [0,A_3] \}$ and 
$\lim_{k \rightarrow \infty}|\lambda^{(k)}-\lambda^*(\xi_{out}^{(k)},\xi_{relay}^{(k)})|=0$. 
Hence, the lemma is proved.
\end{proof}

{\em Remark:} Lemma~\ref{lemma:adaptive_learning_faster_timescale_convergence} tells us that 
the faster time-scale iterate $\lambda^{(k)}$ closely tracks $\lambda^*(\xi_{out}^{(k)},\xi_{relay}^{(k)})$. But it is 
important to note that this lemma does not guarantee the convergence 
of the slower timescale iterates to a single point in the two-dimensional Euclidean plane.

\subsubsection{\textbf{The slower timescale iteration}}
\label{subsubsection:posing_slower_timescale_as_projected_stochastic_approximation_by_kushner}

Let us recall the notation $\overline{Q}_{out}(\lambda, \xi_{out}, \xi_{relay})$, 
$\overline{U}(\lambda, \xi_{out}, \xi_{relay})$, $\overline{Q}_{out}^*(\xi_{out}, \xi_{relay})$ 
and $\overline{U}^*(\xi_{out}, \xi_{relay})$ as defined in 
Section~\ref{subsec:smdp-policy-structure}. 
Let us also recall the update equation (\ref{eqn:learning_backtracking_adaptive_with_outage_cost_update_part}) 
in Algorithm~\ref{algorithm:learning_backtracking_adaptive_with_outage_cost_algorithm}. 
We will analyze the slower timescale update equations as a 
projected stochastic approximation (see Equation~$5.3.1$ of \cite{kushner-clark78SA-constrained-unconstrained}).

Let us denote by $\mathcal{G}$ the compact subset $[0,A_2] \times [0,A_3]$ of the 
Euclidean space. Clearly, the set $\mathcal{G}$ can be defined by the following set of 
constraints on the variables $\xi_{out}$ and $\xi_{relay}$:
\begin{eqnarray}\label{eqn:constraint_equations_in_the_slower_timescale}
 -\xi_{out} \leq 0, \xi_{out} \leq A_2, -\xi_{relay} \leq 0, \xi_{relay} \leq A_3
\end{eqnarray}

We rewrite the slower timescale update equations in (\ref{eqn:learning_backtracking_adaptive_with_outage_cost_update_part}) 
as (\ref{eqn:optexplorelimadaptivelearning_slower_timescale_in_kushner_form}). Note that, the 
functions $f_1(\xi_{out},\xi_{relay})$, $f_2(\xi_{out},\xi_{relay})$, $g_1(\lambda,\xi_{out},\xi_{relay})$, and 
$g_2(\lambda,\xi_{out},\xi_{relay})$ have been defined in 
(\ref{eqn:optexplorelimadaptivelearning_slower_timescale_in_kushner_form}). 
The quantities $M_{1}^{(k)}$ and $M_{2}^{(k)}$ are two zero mean Martingale difference noise sequences w.r.t.  
$\mathcal{F}_{k-1}$; this can be seen as follows. 
Since shadowing is i.i.d. across links, the
shadowing values encountered in the process of measurement
for placing the $k$-th node are independent of the history of the process up to 
the placement of node $(k-1)$. Hence, $\mathbb{E}_{\underline{W}} \bigg(M_1^{(k)} | \mathcal{F}_{k-1} \bigg)=\mathbb{E}_{\underline{W}} \bigg(M_1^{(k)} | (\lambda^{(k-1)},\xi_{out}^{(k-1)}, \xi_{relay}^{(k-1)}) \bigg)=0$ 
and $\mathbb{E}_{\underline{W}} \bigg(M_2^{(k)} | \mathcal{F}_{k-1} \bigg)=\mathbb{E}_{\underline{W}} \bigg(M_2^{(k)} | (\lambda^{(k-1)},\xi_{out}^{(k-1)}, \xi_{relay}^{(k-1)}) \bigg)=0$.

The update is done as follows. We compute 
$\tilde{\xi}_{relay}^{(k)}=\xi_{relay}^{(k-1)}+b_k \bigg( f_2(\xi_{relay}^{(k-1)},\xi_{relay}^{(k-1)})+ g_2(\lambda^{(k-1)},\xi_{relay}^{(k-1)},\xi_{relay}^{(k-1)})+ M_{2}^{(k)} \bigg)$ 
and compute $\xi_{relay}^{(k)}=\Lambda_{[0,A_3]}(\tilde{\xi}_{relay}^{(k)})$. We compute 
$\xi_{out}^{(k)}$ in a similar fashion. {\em Hence, projection onto the set 
$\mathcal{G}$ is nothing but coordinatewise projection.}

{\em Note that, 
(\ref{eqn:optexplorelimadaptivelearning_slower_timescale_in_kushner_form}) is in the same form as the 
standard projected stochastic approximation 
(Equation~$5.3.1$ of \cite{kushner-clark78SA-constrained-unconstrained}). In order to show that 
the iterates in (\ref{eqn:optexplorelimadaptivelearning_slower_timescale_in_kushner_form}) converge to the right 
set, we will make use of Theorem~$5.3.1$ from \cite{kushner-clark78SA-constrained-unconstrained}. 
To use this theorem, we need to check whether 
(\ref{eqn:optexplorelimadaptivelearning_slower_timescale_in_kushner_form}) satisfies five conditions 
from \cite{kushner-clark78SA-constrained-unconstrained}, i.e., 
$A5.1.3$, $A5.1.4$, $A5.1.5$, $A5.3.1.$ and $A5.3.2$. This is done in the next subsection.} 
\qed

\subsubsection{\bf{\em Checking the five conditions from \cite{kushner-clark78SA-constrained-unconstrained}}}
\label{subsubsection:checking_five_conditions_kushner_optexplorelimadaptivelearning}
Before checking the five conditions, we will present a lemma that will be useful for checking one condition.
\begin{lemma}\label{lemma:placement_rate_mean_outage_per_step_continuous_in_lambda_xio_and_xir}
Suppose that Assumption~\ref{assumption:shadowing_continuous_random_variable} holds. Under the 
decision rule given by (\ref{eqn:smdp-optimal-policy_general_form}),   
the mean power per step $\frac{\overline{\Gamma}(\lambda,\xi_{out},\xi_{relay})}{\overline{U}(\lambda,\xi_{out},\xi_{relay})}$, 
mean number of relays per step $\frac{1}{\overline{U}(\lambda,\xi_{out},\xi_{relay})}$ 
and mean outage per step 
$\frac{\overline{Q}_{out}(\lambda,\xi_{out},\xi_{relay})}{\overline{U}(\lambda,\xi_{out},\xi_{relay})}$ 
are continuous in $\lambda$, $\xi_{out}$ and $\xi_{relay}$.
\end{lemma}
\begin{proof}
  The proof is similar to that of Theorem~\ref{theorem:placement_rate_mean_outage_per_step_continuous_in_xio_and_xir}.
\end{proof}

Now, we will check that conditions $A5.1.3$, $A5.1.4$, $A5.1.5$, $A5.3.1.$ and $A5.3.2$ from 
\cite{kushner-clark78SA-constrained-unconstrained} are satisfied.

{\em Checking Condition~$A5.1.3$:} This condition requires 
that $f_1(\cdot,\cdot)$ and $f_2(\cdot,\cdot)$ are continuous functions. This condition is satisfied 
as a consequence of Theorem~\ref{theorem:placement_rate_mean_outage_per_step_continuous_in_xio_and_xir}.\qed

{\em Checking Condition~$A5.1.4$:} This condition is satisfied 
since $b_k>0$, $b_k \rightarrow 0$ as $k \rightarrow \infty$ and $\sum_{k=1}^{\infty}b_k=\infty$.\qed

{\em Checking Condition~$A5.1.5$:} This 
condition requires that $\lim_{k \rightarrow \infty}g_1(\lambda^{(k-1)},\xi_{out}^{(k-1)},\xi_{relay}^{(k-1)})=0$ and  
$\lim_{k \rightarrow \infty}g_2(\lambda^{(k-1)},\xi_{out}^{(k-1)},\xi_{relay}^{(k-1)})=0$ almost surely, 
and that the sequences $g_1(\lambda^{(k-1)},\xi_{out}^{(k-1)},\xi_{relay}^{(k-1)})$ and 
$g_2(\lambda^{(k-1)},\xi_{out}^{(k-1)},\xi_{relay}^{(k-1)})$ are bounded almost surely.

By Lemma~\ref{lemma:boundedness_lambda_iteration_optexplorelimadaptivelearning}, we can find an interval 
$[-d,d]$ such that 
$(\lambda^{(k)}, \xi_{out}^{(k)}, \xi_{relay}^{(k)})$ lies inside the compact set $[-d,d]\times[0,A_2]\times[0,A_3]$ for 
all $k \geq 1$ almost surely. 

Note that, $\overline{Q}_{out}(\lambda, \xi_{out}, \xi_{relay})$ is continuous in each argument 
(by Lemma~\ref{lemma:placement_rate_mean_outage_per_step_continuous_in_lambda_xio_and_xir}). 
Hence, $\overline{Q}_{out}(\lambda, \xi_{out}, \xi_{relay})$ is uniformly continuous over the 
compact set $[-d,d]\times[0,A_2]\times[0,A_3]$ and similarly $\overline{U}(\lambda, \xi_{out}, \xi_{relay})$ 
is uniformly continuous over the compact set $[-d,d]\times[0,A_2]\times[0,A_3]$. 

Now,  by Lemma~\ref{lemma:lemma_used_in_convergence_proof_of_optexplorelimadaptive_learning}, 
the Euclidean distance between  
$(\lambda^{(k)}, \xi_{out}^{(k)}, \xi_{relay}^{(k)})$ and 
$(\lambda^*(\xi_{out}^{(k)}, \xi_{relay}^{(k)}), \xi_{out}^{(k)}, \xi_{relay}^{(k)})$ converges to $0$ 
almost surely as $k \rightarrow \infty$.  
Hence, by uniform continuity, we conclude that $\lim_{k \rightarrow \infty}|\overline{Q}_{out}(\lambda^{(k)}, \xi_{out}^{(k)}, \xi_{relay}^{(k)})-
\overline{Q}_{out}(\lambda^*(\xi_{out}^{(k)}, \xi_{relay}^{(k)}), \xi_{out}^{(k)}, \xi_{relay}^{(k)})|=0$ 
and $\lim_{k \rightarrow \infty}|\overline{U}(\lambda^{(k)}, \xi_{out}^{(k)}, \xi_{relay}^{(k)})-
\overline{U}(\lambda^*(\xi_{out}^{(k)}, \xi_{relay}^{(k)}), \xi_{out}^{(k)}, \xi_{relay}^{(k)})|=0$ almost surely. 
Hence, $\lim_{k \rightarrow \infty}g_1(\lambda^{(k-1)},\xi_{out}^{(k-1)},\xi_{relay}^{(k-1)})=0$ and 
$\lim_{k \rightarrow \infty}g_2(\lambda^{(k-1)},\xi_{out}^{(k-1)},\xi_{relay}^{(k-1)})=0$ almost surely. 

Also, $g_1(\lambda^{(k)},\xi_{out}^{(k)},\xi_{relay}^{(k)})$ and 
$g_2(\lambda^{(k)},\xi_{out}^{(k)},\xi_{relay}^{(k)})$ are uniformly bounded across $k \geq 1$, since the 
outage probabilities and placement distances are bounded quantities.

{\em Checking Condition~$A5.3.1$:} This condition 
requires that $\mathcal{G}=[0,A_2] \times [0,A_3]$ is the closure of its interior, which is true in our problem. It also requires 
that the L.H.S. of each constraint inequality in (\ref{eqn:constraint_equations_in_the_slower_timescale}) 
is continuously differentiable, which is also true in our problem. 

Note that, the L.H.S. of each constraint inequality in (\ref{eqn:constraint_equations_in_the_slower_timescale}) 
is a function of $\xi_{out}$ and $\xi_{relay}$. 
Condition~$A5.3.1$ of \cite{kushner-clark78SA-constrained-unconstrained} needs that for 
each point on the boundary of $\mathcal{G}$, the gradients of the functions 
(in the L.H.S. of (\ref{eqn:constraint_equations_in_the_slower_timescale})) 
corresponding to the active constraints are 
linearly independent. Note that on each point of the boundary of $\mathcal{G}$, at most two constraints 
can be simultaneously active (see (\ref{eqn:constraint_equations_in_the_slower_timescale})). If there are exactly two 
active constraints, one will be for $\xi_{out}$ and the other one will be for 
$\xi_{relay}$. Clearly, the gradients (with respect to the tuple $(\xi_{out},\xi_{relay})$) 
of the active constraint(s) at any boundary 
point of $\mathcal{G}$ are orthogonal, and hence linearly independent.\qed

\begin{figure*}[!t]
\begin{eqnarray}
\xi_{out}^{(k)} &=&  \Lambda_{\mathcal{G}} \bigg( \xi_{out}^{(k-1)}+b_k \overline{U}^*(\xi_{out}^{(k-1)},\xi_{relay}^{(k-1)}) \frac{ \bigg( f_1(\xi_{out}^{(k-1)},\xi_{relay}^{(k-1)})+ g_1(\lambda^{(k-1)},\xi_{out}^{(k-1)},\xi_{relay}^{(k-1)})+ M_{1}^{(k)} \bigg) }{\overline{U}^*(\xi_{out}^{(k-1)},\xi_{relay}^{(k-1)})} \bigg) \nonumber\\
\xi_{relay}^{(k)} &=&  \Lambda_{\mathcal{G}} \bigg( \xi_{relay}^{(k-1)}+b_k \overline{U}^*(\xi_{out}^{(k-1)},\xi_{relay}^{(k-1)}) \frac{ \bigg( f_2(\xi_{out}^{(k-1)},\xi_{relay}^{(k-1)})+ g_2(\lambda^{(k-1)},\xi_{out}^{(k-1)},\xi_{relay}^{(k-1)})+ M_{2}^{(k)} \bigg) }{\overline{U}^*(\xi_{out}^{(k-1)},\xi_{relay}^{(k-1)})} \bigg)
\label{eqn:optexplorelimadaptivelearning_slower_timescale_in_kushner_form_with_scaled_slower_timescale}
\end{eqnarray}
\hrule
\end{figure*}

{\em Checking Condition~$A5.3.2$:}
Let $m(t):=\sup\{n \geq 1: \sum_{i=1}^n b_i \leq t\}$. We have to show that, there exists a $T>0$ such that 
for any $\epsilon>0$, 
we have $\lim_{n \rightarrow \infty} \mathbb{P}\bigg(\sup_{j \geq n} \max_{t \leq T} |\sum_{i=m(jT)}^{m(jT+t)-1}b_i M_1^{(i)}| >\epsilon \bigg)=0$ 
and $\lim_{n \rightarrow \infty} \mathbb{P}\bigg(\sup_{j \geq n} \max_{t \leq T} |\sum_{i=m(jT)}^{m(jT+t)-1}b_i M_2^{(i)}| >\epsilon \bigg)=0$. 
We will prove the first result, for any $T>0$ and any $\epsilon>0$. 
Let us define the event $E_j:=\{\max_{t \leq T} |\sum_{i=m(jT)}^{m(jT+t)-1}b_i M_1^{(i)}| >\epsilon\}$. Hence,

\footnotesize
\begin{eqnarray*}
&& \lim_{n \rightarrow \infty} \mathbb{P}\bigg(\sup_{j \geq n} \max_{t \leq T} |\sum_{i=m(jT)}^{m(jT+t)-1}b_i M_1^{(i)}| >\epsilon \bigg) \nonumber\\
&=& \lim_{n \rightarrow \infty} \mathbb{P}\bigg(\cup_{j \geq n} E_j \bigg) \nonumber\\
&=& \mathbb{P}\bigg(\cap_{n \geq 1} \cup_{j \geq n} E_j \bigg) \nonumber\\
&=& \mathbb{P}\bigg( \limsup_{n \rightarrow \infty} E_n \bigg)
\end{eqnarray*}
\normalsize

where the second equality follows from the continuity of probability.

Now, $\mathbb{P}(E_n) \leq \frac{\mathbb{E}|\sum_{i=m(nT)}^{m(nT+T)-1}b_i M_1^{(i)}|^2}{\epsilon^2}$ 
(by Doob's inequality for Martingales). Since, $|M_1^{(i)}| \leq C$ for some $C>0$ 
(since outage probability and placement distance are two bounded quantities; see 
the expression for $M_1^{(i)}$ in (\ref{eqn:optexplorelimadaptivelearning_slower_timescale_in_kushner_form})) and since 
$\mathbb{E}(M_1^{(i)}M_1^{(j)})=0$ for $i \neq j$, 
the above quantity can be 
upper-bounded by $\mathbb{P}(E_n) \leq 
\frac{C^2 \sum_{i=m(nT)}^{m(nT+T)-1}b_i^2}{\epsilon^2}$. Hence, 
$\sum_{n=1}^{\infty}\mathbb{P}(E_n) \leq \frac{C^2 \sum_{i=1}^{\infty}b_i^2}{\epsilon^2}< \infty$. Hence, by Borel-Cantelli 
lemma, $\mathbb{P}\bigg( \limsup_{n \rightarrow \infty} E_n \bigg)=0$, which completes checking 
Condition~$A5.3.2$ of \cite{kushner-clark78SA-constrained-unconstrained}.\qed

\subsubsection{{\bf Finishing the Proof of 
Theorem~\ref{theorem:convergence_learning_backtracking_adaptive_with_outage_cost_algorithm}}}
\label{subsubsection:finishing_two_timescale_adaptive_learning_with_outage_convergence_proof}

Now, we will invoke Theorem~$5.3.1$ from \cite{kushner-clark78SA-constrained-unconstrained} to complete the proof. 

Let us rewrite (\ref{eqn:optexplorelimadaptivelearning_slower_timescale_in_kushner_form}) as 
(\ref{eqn:optexplorelimadaptivelearning_slower_timescale_in_kushner_form_with_scaled_slower_timescale}). 
Note that, $A+1 \leq \overline{U}^*(\xi_{out}^{(k-1)}, \xi_{relay}^{(k-1)}) \leq A+B$. Hence, if we use step size  
$b_k \overline{U}^*(\xi_{out}^{(k-1)}, \xi_{relay}^{(k-1)})$ and use the modified functions like 
$\frac{  f_1(\xi_{out}^{(k-1)},\xi_{relay}^{(k-1)})  }{\overline{U}^*(\xi_{out}^{(k-1)},\xi_{relay}^{(k-1)})}$ as 
in (\ref{eqn:optexplorelimadaptivelearning_slower_timescale_in_kushner_form_with_scaled_slower_timescale}), the 
conditions checked in the previous subsection will still hold. This is evident from the fact that, 
once we know $\xi_{out}^{(k-1)}$ and 
$\xi_{relay}^{(k-1)}$, $\overline{U}^*(\xi_{out}^{(k-1)},\xi_{relay}^{(k-1)})$ becomes 
a deterministic quantity, and that the randomness in the computation of the new iterates 
$\xi_{out}^{(k)}$ and $\xi_{relay}^{(k)}$ comes from the random shadowing in the links measured in the process of deploying 
the $k$-th node. Hence, $\frac{M_{1}^{(k)}}{\overline{U}^*(\xi_{out}^{(k-1)},\xi_{relay}^{(k-1)})}$ and 
$\frac{M_{2}^{(k)}}{\overline{U}^*(\xi_{out}^{(k-1)},\xi_{relay}^{(k-1)})}$ are also Martingale difference sequences. 
It is easy to check conditions $A5.1.3$, $A5.1.5$ and $A5.3.2$ for 
(\ref{eqn:optexplorelimadaptivelearning_slower_timescale_in_kushner_form_with_scaled_slower_timescale}), 
and the condition in $A5.1.4$ is satisfied almost surely.

Hence, from now on, let us consider the slower timescale iteration 
(\ref{eqn:optexplorelimadaptivelearning_slower_timescale_in_kushner_form_with_scaled_slower_timescale}).

For the function 
$h(\xi_{out},\xi_{relay}):=\bigg(\frac{f_1(\xi_{out},\xi_{relay})}{\overline{U}^*(\xi_{out},\xi_{relay})},\frac{f_2(\xi_{out},\xi_{relay})}{\overline{U}^*(\xi_{out},\xi_{relay})} \bigg) 
=\bigg(\frac{\overline{Q}_{out}^*(\xi_{out},\xi_{relay})}{\overline{U}^*(\xi_{out},\xi_{relay})}-\overline{q},\frac{1}{\overline{U}^*(\xi_{out},\xi_{relay})}-\overline{N} \bigg)$,  
let us define the map:

\footnotesize
\begin{eqnarray}
&& \overline{\Lambda}_{\mathcal{G}}(h(\xi_{out},\xi_{relay})) \nonumber\\
&=& \lim_{0<\beta \rightarrow 0} \frac{\Lambda_{\mathcal{G}}\bigg((\xi_{out},\xi_{relay})+\beta h(\xi_{out},\xi_{relay}))\bigg)-(\xi_{out},\xi_{relay})}{\beta} \nonumber\\
\label{eqn:definition_of_Lambda_G}
\end{eqnarray}
\normalsize

We want to show that the iterates $(\xi_{out}^{(k)}, \xi_{relay}^{(k)})$ will converge almost surely to 
the set of stationary points of the o.d.e.  
$(\dot{\xi}_{out}(t), \dot{\xi}_{relay}(t))=\overline{\Lambda}_{\mathcal{G}}\bigg(\frac{f_1(\xi_{out}(t),\xi_{relay}(t))}{\overline{U}^*(\xi_{out}(t),\xi_{relay}(t))},\frac{f_2(\xi_{out}(t),\xi_{relay}(t))}{\overline{U}^*(\xi_{out}(t),\xi_{relay}(t))} \bigg)$. 
This will follow from Theorem~$5.3.1$ from \cite{kushner-clark78SA-constrained-unconstrained}, if we can 
show that $\bigg(-\frac{f_1(\xi_{out},\xi_{relay})}{\overline{U}^*(\xi_{out},\xi_{relay})},-\frac{f_2(\xi_{out},\xi_{relay})}{\overline{U}^*(\xi_{out},\xi_{relay})} \bigg)$ is the gradient of a 
continuously differentiable function.

Let us denote, by $\overline{\Gamma}_{\pi}$, $\overline{U}_{\pi}$ and 
$\overline{Q}_{out,\pi}$, the mean power per link, mean placement distance per link and mean 
outage per link respectively, under any given stationary deployment policy $\pi$. 
Let us define the function 

\footnotesize
\begin{eqnarray}
 G(\xi_{out},\xi_{relay})& := & \inf_{\pi} \bigg( \frac{\overline{\Gamma}_{\pi}}{\overline{U}_{\pi}}+ \xi_{out} (\frac{\overline{Q}_{out,\pi}}{\overline{U}_{\pi}}-\overline{q}) + \xi_{relay}(\frac{1}{\overline{U}_{\pi}}-\overline{N})  \bigg) \nonumber\\
 && \label{eqn:expression_for_G_xi_out_xi_relay}
\end{eqnarray}
\normalsize

\begin{lemma}\label{lemma:continuously_differentiable_adaptive_learning_convergence_proof}
 $G(\xi_{out},\xi_{relay})$ is continuously differentiable and its gradient is 
$\bigg(\frac{f_1(\xi_{out},\xi_{relay})}{\overline{U}^*(\xi_{out},\xi_{relay})},\frac{f_2(\xi_{out},\xi_{relay})}{\overline{U}^*(\xi_{out},\xi_{relay})} \bigg)$. 
\end{lemma}
\begin{proof}
The proof of Lemma~\ref{lemma:continuously_differentiable_adaptive_learning_convergence_proof} will be provided later in this section.
\end{proof}

Now, $\bigg(-\frac{f_1(\xi_{out},\xi_{relay})}{\overline{U}^*(\xi_{out},\xi_{relay})},-\frac{f_2(\xi_{out},\xi_{relay})}{\overline{U}^*(\xi_{out},\xi_{relay})} \bigg)$ 
is the gradient of a continuously differentiable function $-G(\xi_{out},\xi_{relay})$. 
Hence, by Theorem~$5.3.1$ from \cite{kushner-clark78SA-constrained-unconstrained}, the 
iterates $(\xi_{out}^{(k)},\xi_{relay}^{(k)})$ will almost surely converge to 
the set of stationary points of the o.d.e. 
$(\dot{\xi}_{out}(t), \dot{\xi}_{relay}(t))=\overline{\Lambda}_{\mathcal{G}}\bigg(\frac{f_1(\xi_{out}(t),\xi_{relay}(t))}{\overline{U}^*(\xi_{out}(t),\xi_{relay}(t))},\frac{f_2(\xi_{out}(t),\xi_{relay}(t))}{\overline{U}^*(\xi_{out}(t),\xi_{relay}(t))} \bigg)$.

\begin{lemma}\label{lemma:stationary_point_is_optimal_point_adaptive_learning_proof}
If $(\xi_{out},\xi_{relay}) \in [0,A_2] \times [0,A_3]$ is a zero of 
$\overline{\Lambda}_{\mathcal{G}}\bigg(\frac{f_1(\xi_{out},\xi_{relay})}{\overline{U}^*(\xi_{out},\xi_{relay})},\frac{f_2(\xi_{out},\xi_{relay})}{\overline{U}^*(\xi_{out},\xi_{relay})} \bigg)$, 
then $(\lambda^*(\xi_{out},\xi_{relay}),\xi_{out},\xi_{relay}) \in \mathcal{K}(\overline{q},\overline{N})$, 
provided that $A_2$ and $A_3$ are chosen properly. 
\end{lemma}
\begin{proof}
 The proof of Lemma~\ref{lemma:stationary_point_is_optimal_point_adaptive_learning_proof} will be provided later in this section. 
 One way of choosing $A_2$ and $A_3$ has been described 
 before the proof of this lemma.
\end{proof}

{\em We have already shown that $(\xi_{out}^{(k)},\xi_{relay}^{(k)})$ converges to the set of $(\xi_{out},\xi_{relay})$ pairs for which 
$\overline{\Lambda}_{\mathcal{G}}\bigg(\frac{f_1(\xi_{out},\xi_{relay})}{\overline{U}^*(\xi_{out},\xi_{relay})},\frac{f_2(\xi_{out},\xi_{relay})}{\overline{U}^*(\xi_{out},\xi_{relay})} \bigg)=(0,0)$. 
Hence, by Lemma~\ref{lemma:adaptive_learning_faster_timescale_convergence} and 
Lemma~\ref{lemma:stationary_point_is_optimal_point_adaptive_learning_proof}, 
$(\lambda^{(k)},\xi_{out}^{(k)},\xi_{relay}^{(k)}) \rightarrow \mathcal{K}(\overline{q},\overline{N})$ 
almost surely, which completes the proof of 
Theorem~\ref{theorem:convergence_learning_backtracking_adaptive_with_outage_cost_algorithm}.} 

Now we will prove Lemma~\ref{lemma:continuously_differentiable_adaptive_learning_convergence_proof} and 
Lemma~\ref{lemma:stationary_point_is_optimal_point_adaptive_learning_proof}. 
Before we prove Lemma~\ref{lemma:stationary_point_is_optimal_point_adaptive_learning_proof}, we 
will explain how $A_2$ and $A_3$ have to be chosen.

\textbf{Proof of Lemma~\ref{lemma:continuously_differentiable_adaptive_learning_convergence_proof}}
Suppose that, for a given $(\xi_{out},\xi_{relay})$, the partial derivative 
 $\frac{\partial G}{\partial \xi_{out}}$ exists. We will first show that this partial derivative 
 is equal to 
 $\frac{f_1(\xi_{out},\xi_{relay})}{\overline{U}^*(\xi_{out},\xi_{relay})}=\frac{\overline{Q}_{out}^*(\xi_{out},\xi_{relay})}{\overline{U}^*(\xi_{out},\xi_{relay})}-\overline{q}$. 
Note that, the right partial derivative w.r.t. $\xi_{out}$ (if it exists) is:
 
 \begin{eqnarray*}
  \frac{\partial G}{\partial \xi_{out}+}=\lim_{0<\Delta \rightarrow 0}\frac{ G(\xi_{out}+\Delta,\xi_{relay})-G(\xi_{out},\xi_{relay}) }{\Delta}
 \end{eqnarray*}
 
 Now, the optimal policy $\pi^*(\xi_{out},\xi_{relay})$ 
 for the unconstrained problem in (\ref{eqn:unconstrained_problem_average_cost_with_outage_cost}) 
 will also  minimize the expression for 
 $G(\xi_{out},\xi_{relay})$ in (\ref{eqn:expression_for_G_xi_out_xi_relay}). But, 
 the policy $\pi^*(\xi_{out},\xi_{relay})$ will be 
 suboptimal for the pair $(\xi_{out}+\Delta, \xi_{relay})$. 
 Hence, we have:
 
 \footnotesize
 \begin{eqnarray*}
 && G(\xi_{out}+\Delta,\xi_{relay}) \\
 &=& \inf_{\pi} \bigg( \frac{\overline{\Gamma}_{\pi}}{\overline{U}_{\pi}}+ (\xi_{out}+\Delta) (\frac{\overline{Q}_{out,\pi}}{\overline{U}_{\pi}}-\overline{q}) + \xi_{relay}(\frac{1}{\overline{U}_{\pi}}-\overline{N})  \bigg) \\ 
  & \leq & \bigg( \frac{\overline{\Gamma}^*(\xi_{out},\xi_{relay})}{\overline{U}^*(\xi_{out},\xi_{relay})}+ (\xi_{out}+\Delta) (\frac{\overline{Q}^*(\xi_{out},\xi_{relay})}{\overline{U}^*(\xi_{out},\xi_{relay})}-\overline{q}) \\
  && + \xi_{relay}(\frac{1}{\overline{U}^*(\xi_{out},\xi_{relay})}-\overline{N})  \bigg) \\ 
 &=& G(\xi_{out},\xi_{relay})+ \Delta \bigg(\frac{\overline{Q}_{out}^*(\xi_{out},\xi_{relay})}{\overline{U}^*(\xi_{out},\xi_{relay})}-\overline{q} \bigg)
 \end{eqnarray*}
\normalsize

which implies that, 

 \begin{eqnarray*}
  \frac{\partial G}{\partial \xi_{out}+} \leq \bigg(\frac{\overline{Q}_{out}^*(\xi_{out},\xi_{relay})}{\overline{U}^*(\xi_{out},\xi_{relay})}-\overline{q} \bigg)
 \end{eqnarray*}
 
In a similar manner, by using the fact that 
$\pi^*(\xi_{out},\xi_{relay})$ is suboptimal for the pair $(\xi_{out}-\Delta, \xi_{relay})$, we can claim that 

 \begin{eqnarray*}
  \frac{\partial G}{\partial \xi_{out}-} \geq \bigg(\frac{\overline{Q}_{out}^*(\xi_{out},\xi_{relay})}{\overline{U}^*(\xi_{out},\xi_{relay})}-\overline{q}\bigg)
 \end{eqnarray*}
 
Since we have assumed that $\frac{\partial G}{\partial \xi_{out}}$ exists, we must have 
$\frac{\partial G}{\partial \xi_{out}+} = \frac{\partial G}{\partial \xi_{out}-}$, 
which proves that the partial derivative w.r.t. $\xi_{out}$ will be equal to 
 $(\frac{\overline{Q}_{out}^*(\xi_{out},\xi_{relay})}{\overline{U}^*(\xi_{out},\xi_{relay})}-\overline{q})$.
 
 We now turn to the existence of $\frac{\partial G}{\partial \xi_{out}}$. 
  Note that, since $G(\xi_{out},\xi_{relay})$ is the minimum of a family of affine functions of 
 $\xi_{out}$ and $\xi_{relay}$, $G(\xi_{out},\xi_{relay})$ is concave and hence coordinatewise 
 concave. Hence, for any given $\xi_{relay}$, there are 
 only at most countably many values of $\xi_{out}$ where $\frac{\partial G}{\partial \xi_{out}}$ does not exist. To see this, 
 let us define the function $H(\xi_{out},\xi_{relay})$ to be the supremum of the subgradients of 
 $G(\xi_{out},\xi_{relay})$  with respect to $\xi_{out}$ (keeping $\xi_{relay}$ fixed), at a point $(\xi_{out},\xi_{relay})$. 
 Since $G(\xi_{out},\xi_{relay})$ is concave, $H(\xi_{out},\xi_{relay})$ will be decreasing in $\xi_{out}$. 
 But any monotone real-valued function has an at most countable number of discontinuities 
 (see \cite{rudin76principles-of-mathematical-analysis}, Theorem~$4.30$). Hence, for a given $\xi_{relay}$, 
 the function $H(\xi_{out},\xi_{relay})$ is discontinuous for an at most countable number of values of 
 $\xi_{out}$, and consequently $\frac{\partial G}{\partial \xi_{out}}$ 
   exists everywhere except for an at most countable set of values of $\xi_{out}$.  

 For a given $\xi_{relay}$, let $\xi_{out}'$ be one such value where $\frac{\partial G}{\partial \xi_{out}}$ 
 does not exist. Then, there exists a sequence $\{ \zeta_n \}_{n \geq 1} \downarrow 0$ such that 
 $ \frac{\partial G}{ \partial \xi_{out}}$ exists at each $\xi_{out}=\xi_{out}'+\zeta_n$. This follows 
 from the fact that for any $\zeta>0$, we can find one $\xi_{out} \in (\xi_{out}',\xi_{out}'+\zeta)$ where 
 $ \frac{\partial G}{ \partial \xi_{out}}$ exists, otherwise the number of points where $ \frac{\partial G}{ \partial \xi_{out}}$ 
 does not exist will become uncountable. Similarly, there exists a sequence $\{ \kappa_n \}_{n \geq 1} \downarrow 0$ such that 
 $ \frac{\partial G}{ \partial \xi_{out}}$ exists at each $\xi_{out}=\xi_{out}'-\kappa_n$.

 Note that, by concavity, 
 $\lim_{n \rightarrow \infty} \frac{\partial G}{ \partial \xi_{out}}|_{\xi_{out}'-\kappa_n} 
 \geq \frac{\partial G}{ \partial \xi_{out}-}|_{\xi_{out}'}  \geq \frac{\partial G}{ \partial \xi_{out}+}|_{\xi_{out}'} 
 \geq \lim_{n \rightarrow \infty} \frac{\partial G}{ \partial \xi_{out}}|_{\xi_{out}'+\zeta_n} $. 
 The last term in this chain of inequalities is equal to 
  $\lim_{n \rightarrow \infty}(\frac{\overline{Q}_{out}^*(\xi_{out}'+\zeta_n,\xi_{relay})}{\overline{U}^*(\xi_{out}'+\zeta_n,\xi_{relay})}-\overline{q}) 
 =(\frac{\overline{Q}_{out}^*(\xi_{out}',\xi_{relay})}{\overline{U}^*(\xi_{out}',\xi_{relay})}-\overline{q}) $ by 
 the arguments in the beginning of this proof and by the continuity results in 
Theorem~\ref{theorem:placement_rate_mean_outage_per_step_continuous_in_xio_and_xir}. Same arguments 
hold for the first term in the chain of inequalities. Hence, 
$\frac{\partial G}{ \partial \xi_{out}-}|_{\xi_{out}'}  =\frac{\partial G}{ \partial \xi_{out}+}|_{\xi_{out}'}=
\frac{\partial G}{ \partial \xi_{out}}|_{\xi_{out}'}=(\frac{\overline{Q}_{out}^*(\xi_{out}',\xi_{relay})}{\overline{U}^*(\xi_{out}',\xi_{relay})}-\overline{q})$.

In a similar way, we can show that 
$\frac{\partial G}{\partial \xi_{relay}}=(\frac{1}{\overline{U}^*(\xi_{out},\xi_{relay})}-\overline{N})$. 

Now we see that both of the partial derivatives of $G$ exist at all points and the partial derivatives are continuous 
in both $\xi_{out}$ and $\xi_{relay}$ (by 
Theorem~\ref{theorem:placement_rate_mean_outage_per_step_continuous_in_xio_and_xir}). Hence, by Theorem~$12.11$ of 
\cite{apostol81mathematical-analysis}, $G(\cdot,\cdot)$ is differentiable. 
Hence, the lemma is proved.\qed

\textbf{Choice of $A_2$ and $A_3$:} Let us consider the scenario where   
$\overline{\Lambda}_{\mathcal{G}}\bigg(\frac{f_1(\xi_{out},\xi_{relay})}{\overline{U}^*(\xi_{out},\xi_{relay})},\frac{f_2(\xi_{out},\xi_{relay})}{\overline{U}^*(\xi_{out},\xi_{relay})} \bigg)$ 
has a zero $(\xi_{out}',\xi_{relay}')$ (on the boundary of $\mathcal{G}$) such that 
$(\lambda^*(\xi_{out}',\xi_{relay}'),\xi_{out}',\xi_{relay}') \not \in \mathcal{K}(\overline{q},\overline{N})$. In this case, 
if $(\lambda^{(k)},\xi_{out}^{(k)},\xi_{relay}^{(k)}) \rightarrow (\lambda^*(\xi_{out}',\xi_{relay}'),\xi_{out}',\xi_{relay}')$ 
(depending on the sample path of the iterates in the OptExploreLimAdaptiveLearning algorithm),  
then we cannot expect the desired performance from the OptExploreLimAdaptiveLearning algorithm. 
To alleviate this problem, we need to choose 
$A_2$ and $A_3$ in a proper way. One method of choosing $A_2$ and $A_3$ is given below.

We will first explain how $A_2$ has to be chosen. 
Note that, for any given link of length $u$ and shadowing realization $w$, 
$\argmin_{\gamma \in \mathcal{S}}(\gamma+\xi_{out} Q_{out}(u,\gamma,w))=P_M$ if we choose 
$\xi_{out}$  sufficiently large.  
We use this fact in the 
choice of  $A_2$. 
The number $A_2$ has to be chosen so large 
that under $\xi_{out}=A_2$ and for all $A+1 \leq u \leq A+B$, we will have 
$\mathbb{P}(\argmin_{\gamma \in \mathcal{S}}(\gamma+A_2 Q_{out}(u,\gamma,W))=P_M)>1-\kappa$ 
for some small enough $\kappa>0$. Such a choice of $A_2$ ensures that 
(i) the mean power per link (under policy $\pi^*(A_2,\xi_{relay})$),  
$\overline{\Gamma}^*(A_2,\xi_{relay}) \geq (1-\kappa) P_M +\kappa P_1$ (which is close enough to $P_M$), which, 
further, ensures that (ii) $\frac{\overline{\Gamma}^*(A_2,\xi_{relay})}{1/ \overline{N}}$ is greater than or equal to  
the optimal mean power per step for problem (\ref{eqn:constrained_problem_average_cost_with_outage_cost}). The 
second claim is easy to see, since 
$\overline{\Gamma}^*(A_2,\xi_{relay}) \geq (1-\kappa) P_M +\kappa P_1 \geq \overline{\Gamma}^*(\xi_{out}^*,\xi_{relay}^*)$, 
and since $\overline{U}^*(\xi_{out}^*,\xi_{relay}^*) \geq \frac{1}{\overline{N}}$ 
(recall Assumption~\ref{assumption:existence_of_xio_xir} about the existence 
of $\xi_{out}^*$ and $\xi_{relay}^*$). Note that, the choice of $\kappa$ depends on 
$(\overline{q},\overline{N})$ and the radio propagation parameters, and, hence, must be made carefully so 
that the condition is satisfied. 
In the proof of Lemma~\ref{lemma:stationary_point_is_optimal_point_adaptive_learning_proof}, we will see that 
this condition ensures that for any stationary point of the form 
$\xi_{out}=A_2,\xi_{relay} \in (0,A_3)$, we have 
$\frac{\overline{Q}_{out}^*(A_2,\xi_{relay})}{\overline{U}^*(A_2,\xi_{relay})}=\overline{q}$ 
and $\frac{1}{\overline{U}^*(A_2,\xi_{relay})}=\overline{N}$, and, 
consequently, the point $(A_2,\xi_{relay})$ will be in 
$\mathcal{K}(\overline{q},\overline{N})$.

The choice of $A_2$ must satisfy another condition. We need to choose $A_2$ so large that 
$\frac{\overline{Q}_{out}^*(A_2,0)}{\overline{U}^*(A_2,0)} \leq \overline{q}$. 
Note that, if $(\overline{q},\overline{N})$ is a feasible constraint pair, then a constraint $\overline{q}$ on the 
mean outage per step alone (if we drop the constraint on the relay placement rate) 
is also feasible. Let us consider the problem of minimizing the mean power per step subject 
to a constraint $\overline{q}$ on the mean outage per step. Then, we will choose $\xi_{relay}=0$. The mean outage per 
step under policy $\pi^*(\xi_{out},0)$ will still decrease as $\xi_{out}$ increases 
(by Theorem~\ref{theorem:outage_decreasing_with_xio_placement_rate_decreasing_with_xir}). 
Hence, we can choose an $A_2$ which satisfies this condition. 
This condition will be used in showing that if $(A_2,0)$ is a stationary point 
of the o.d.e., then $(\lambda^*(A_2,0),A_2,0) \in \mathcal{K}(\overline{q},\overline{N})$.

$A_2$ has to be chosen (according to the two criteria mentioned above) via prior computation, 
using the prior knowledge of the propagation environment; 
if we know the range of values of 
radio propagation parameters (e.g., $\eta$ and $\sigma$), we can compute what value of $A_2$ will 
satisfy the criteria under all possible radio propagation parameters.

Once $A_2$ is chosen, we need to choose $A_3$. The number $A_3$ has to be chosen so large that 
for any $\xi_{out} \in [0,A_2]$, we will have 
$\overline{U}^*(\xi_{out},A_3) > \frac{1}{\overline{N}}$ (provided that $\frac{1}{\overline{N}}<A+B$). This is possible 
and obvious from the structure of OptExploreLim (Algorithm~\ref{algorithm:policy_structure_smdp_backtracking}); 
by choosing $\xi_{relay}$ large enough, we can achieve a mean placement distance equal to $(A+B)$, provided that 
$\xi_{out} \in [0,A_2]$. For example, if we choose $A_3 = 100 (A+B) (P_M+A_2)$, then:
\begin{eqnarray*}
&&\lambda^*(\xi_{out},A_3)\\
&=& \frac{\overline{\Gamma}^*(\xi_{out},A_3)+\xi_{out}\overline{Q}_{out}^*(\xi_{out},A_3)+A_3}{ \overline{U}^*(\xi_{out},A_3) } \\
& \geq & \frac{A_3}{A+B}=100 (P_M+A_2)
\end{eqnarray*}
and $\pi^*(\xi_{out},A_3)$  will always place at a distance of $(A+B)$. 
This choice of $A_3$ ensures that the policy $\pi^*(\xi_{out},A_3)$ satisfies the 
constraint on the relay placement rate with strict inequality, and hence no point 
of the form $(\xi_{out},A_3)$ is a stationary point of the o.d.e.

The numbers $A_2$ and $A_3$ have to be chosen so large that there exists at least one 
$(\xi_{out}',\xi_{relay}') \in [0,A_2] \times [0,A_3]$ such that 
$(\lambda^*(\xi_{out}',\xi_{relay}'),\xi_{out}',\xi_{relay}') \in \mathcal{K}(\overline{q},\overline{N})$.\qed

\textbf{Proof of Lemma~\ref{lemma:stationary_point_is_optimal_point_adaptive_learning_proof}:}
Suppose that $(\xi_{out},\xi_{relay}) \in [0,A_2] \times [0,A_3]$ is a zero of 
$\overline{\Lambda}_{\mathcal{G}}\bigg(\frac{f_1(\xi_{out},\xi_{relay})}{\overline{U}^*(\xi_{out},\xi_{relay})},\frac{f_2(\xi_{out},\xi_{relay})}{\overline{U}^*(\xi_{out},\xi_{relay})} \bigg)$. 

Note that, 
$\overline{\Lambda}_{\mathcal{G}}\bigg(\frac{f_1(\xi_{out},\xi_{relay})}{\overline{U}^*(\xi_{out},\xi_{relay})},\frac{f_2(\xi_{out},\xi_{relay})}{\overline{U}^*(\xi_{out},\xi_{relay})} \bigg)$ 
is equal to $(\frac{f_1(\xi_{out},\xi_{relay})}{\overline{U}^*(\xi_{out},\xi_{relay})},\frac{f_2(\xi_{out},\xi_{relay})}{\overline{U}^*(\xi_{out},\xi_{relay})} \bigg)$ 
if $(\xi_{out},\xi_{relay})$ lies in the interior of $[0,A_2] \times [0,A_3]$. 
Thus, for any stationary point $(\xi_{out},\xi_{relay}) \in (0,A_2) \times (0,A_3)$, the optimal policy 
$\pi^*(\xi_{out},\xi_{relay})$ meets both constraints in 
(\ref{eqn:constrained_problem_average_cost_with_outage_cost}) with R.H.S.=0. For such a stationary point, 
$(\lambda^*(\xi_{out},\xi_{relay}), \xi_{out},\xi_{relay}) \in \mathcal{K}(\overline{q},\overline{N})$ 
(by Theorem~\ref{theorem:how-to-choose-optimal-Lagrange-multiplier}).

 A point $(\xi_{out},\xi_{relay})$ on the boundary has $\xi_{out}=0$ 
 or $\xi_{out}=A_2$ or $\xi_{relay}=0$ or $\xi_{relay}=A_3$. 
 
 Let us recall  Assumption~\ref{assumption:existence_of_xio_xir} and the definition of 
 $\overline{\Lambda}_{\mathcal{G}}\bigg(\frac{f_1(\xi_{out},\xi_{relay})}{\overline{U}^*(\xi_{out},\xi_{relay})},\frac{f_2(\xi_{out},\xi_{relay})}{\overline{U}^*(\xi_{out},\xi_{relay})} \bigg)$ 
 (equation (\ref{eqn:definition_of_Lambda_G})). 
The first component of this vector-valued function at $\xi_{out}=0$ is equal to 
$\frac{f_1(0,\xi_{relay})}{\overline{U}^*(0,\xi_{relay})}$  if $f_1(0,\xi_{relay}) \geq 0$, and $0$ otherwise. We can make similar 
observations at $\xi_{out}=A_2$, $\xi_{relay}=0$ and $\xi_{relay}=A_3$.

If $\xi_{relay}=A_3$, then by the choice of $A_3$ as suggested in 
the OptExploreLimAdaptiveLearning algorithm 
(Algorithm~\ref{algorithm:learning_backtracking_adaptive_with_outage_cost_algorithm}), we will 
have $f_2(\xi_{out},A_3)<0$ (since $\overline{U}^*(\xi_{out},A_3) > \frac{1}{\overline{N}}$). 
This implies that no point on $\xi_{relay}=A_3$ can be a zero of 
$\overline{\Lambda}_{\mathcal{G}}\bigg(\frac{f_1(\xi_{out},\xi_{relay})}{\overline{U}^*(\xi_{out},\xi_{relay})},\frac{f_2(\xi_{out},\xi_{relay})}{\overline{U}^*(\xi_{out},\xi_{relay})} \bigg)$, 
since the second component of this function will be 
$\frac{f_2(\xi_{out},A_3)}{\overline{U}^*(\xi_{out},A_3)}$ which is strictly negative.

Suppose that that there is a zero of 
$\overline{\Lambda}_{\mathcal{G}}\bigg(\frac{f_1(\xi_{out},\xi_{relay})}{\overline{U}^*(\xi_{out},\xi_{relay})},\frac{f_2(\xi_{out},\xi_{relay})}{\overline{U}^*(\xi_{out},\xi_{relay})} \bigg)$ 
of the form 
$\xi_{out}=A_2,\xi_{relay} \in (0,A_3)$. Then 
$\overline{\Lambda}_{\mathcal{G}}\bigg(\frac{f_1(\xi_{out},\xi_{relay})}{\overline{U}^*(\xi_{out},\xi_{relay})},\frac{f_2(\xi_{out},\xi_{relay})}{\overline{U}^*(\xi_{out},\xi_{relay})} \bigg)$ 
will be zero if and only if $f_1(A_2,\xi_{relay}) \geq 0 $ and $f_2(A_2,\xi_{relay})=0$. 
If $f_1(A_2,\xi_{relay}) = 0 $ and $f_2(A_2,\xi_{relay})=0$, then 
$(\lambda^*(A_2,\xi_{relay}),A_2,\xi_{relay})$ will belong to $\mathcal{K}(\overline{q},\overline{N})$ 
(by Theorem~\ref{theorem:how-to-choose-optimal-Lagrange-multiplier}), since the corresponding optimal 
policy $\pi^{*}(A_2,\xi_{relay})$ will satisfy both constraints in (\ref{eqn:constrained_problem_average_cost_with_outage_cost}) 
with equality. Now, we will show that, if 
$A_2$ is chosen appropriately as explained before, the case $f_1(A_2,\xi_{relay}) > 0 $ and $f_2(A_2,\xi_{relay})=0$ will never arise. 
Suppose that $f_1(A_2,\xi_{relay}) > 0 $ and $f_2(A_2,\xi_{relay})=0$ for some $\xi_{relay} \in (0,A_3)$. 
Consider a new problem of minimizing the mean 
outage per step, subject to a constraint $\frac{\overline{\Gamma}^*(A_2,\xi_{relay})}{\overline{U}^*(A_2,\xi_{relay})}$ on the 
mean power per step and a constraint $\frac{1}{\overline{U}^*(A_2,\xi_{relay})}=\overline{N}$ 
on the mean number of relays per step. 
By Theorem~\ref{theorem:how-to-choose-optimal-Lagrange-multiplier}, 
$\pi^{*}(A_2,\xi_{relay})$ is the optimal policy for this new problem, since it satisfies both 
constraints with equality. But the policy $\pi^*(\xi_{out}^*,\xi_{relay}^*)$  has the following properties: (i) 
$\frac{1}{\overline{U}^*(A_2,\xi_{relay})}=\overline{N} \geq \frac{1}{\overline{U}^*(\xi_{out}^*,\xi_{relay}^*)}$  
(see Assumption~\ref{assumption:existence_of_xio_xir} 
in Section~\ref{section:learning_backtracking_adaptive_with_outage_cost}), i.e., 
$\pi^*(\xi_{out}^*,\xi_{relay}^*)$ has a smaller relay placement rate compared 
to $\pi^*(A_2,\xi_{relay})$ 
(since $\pi^*(\xi_{out}^*,\xi_{relay}^*)$ satisfies the constraint $\overline{N}$ on the mean number of relays per step), (ii) 
$\frac{\overline{\Gamma}^*(A_2,\xi_{relay})}{\overline{U}^*(A_2,\xi_{relay})} \geq 
\frac{(1-\kappa)P_M+\kappa P_1}{\frac{1}{\overline{N}}} \geq \frac{\overline{\Gamma}^*(\xi_{out}^*,\xi_{relay}^*)}{\overline{U}^*(\xi_{out}^*,\xi_{relay}^*)}$, i.e., 
$\pi^*(\xi_{out}^*,\xi_{relay}^*)$ has a smaller mean power per step compared 
to $\pi^*(A_2,\xi_{relay})$ (by the choice of $A_2$, see the previous discussion on the choice of $A_2$), 
and (iii) $\frac{\overline{Q}_{out}^*(A_2,\xi_{relay})}{\overline{U}^*(A_2,\xi_{relay})} >\overline{q} \geq \frac{\overline{Q}_{out}^*(\xi_{out}^*,\xi_{relay}^*)}{\overline{U}^*(\xi_{out}^*,\xi_{relay}^*)}$, 
i.e., $\pi^*(\xi_{out}^*,\xi_{relay}^*)$ has a strictly smaller mean outage per step compared 
to $\pi^*(A_2,\xi_{relay})$ 
(since $\pi^*(\xi_{out}^*,\xi_{relay}^*)$ satisfies the constraint $\overline{q}$ on the mean outage per step 
and since $f_1(A_2,\xi_{relay})>0$. 
This leads to a contradiction since it violates the optimality of the policy $\pi^*(A_2,\xi_{relay})$ 
for the new problem. Hence, 
$\overline{\Lambda}_{\mathcal{G}}\bigg(\frac{f_1(\xi_{out},\xi_{relay})}{\overline{U}^*(\xi_{out},\xi_{relay})},\frac{f_2(\xi_{out},\xi_{relay})}{\overline{U}^*(\xi_{out},\xi_{relay})} \bigg)$ 
cannot have a zero of the form $\xi_{out}=A_2$, $\xi_{relay} \in (0,A_3)$ 
such that $f_1(A_2,\xi_{relay})>0$ and $f_2(A_2,\xi_{relay})=0$.

Now consider any stationary point of the form $\xi_{out} \in (0,A_2), \xi_{relay}=0$. 
Then we must have $f_1(\xi_{out},0)=0$ and $f_2(\xi_{out},0) \leq 0$. Now, consider a new problem of 
minimizing the mean power per step subject to a constraint $\overline{q}$ on the mean outage per step (with no 
constraint on the relay placement rate); 
an optimal policy for this problem is $\pi^*(\xi_{out},0)$ 
(by Theorem~\ref{theorem:how-to-choose-optimal-Lagrange-multiplier}, 
since $\pi^*(\xi_{out},0)$ satisfies the outage constraint with equality). 
Since (\ref{eqn:constrained_problem_average_cost_with_outage_cost}) 
has one more constraint, it will have a larger mean power per step, i.e., 
$\frac{\overline{\Gamma}^*(\xi_{out},0)}{\overline{U}^*(\xi_{out},0)} \leq 
\frac{\overline{\Gamma}^*(\xi_{out}^*,\xi_{relay}^*)}{\overline{U}^*(\xi_{out}^*,\xi_{relay}^*)}$ 
(recall Assumption~\ref{assumption:existence_of_xio_xir} about the existence 
of $\xi_{out}^*$ and $\xi_{relay}^*$). If they are equal, then 
$(\xi_{out},0)$ will be an optimal pair for (\ref{eqn:constrained_problem_average_cost_with_outage_cost}) and 
$(\lambda^*(\xi_{out},0),\xi_{out},0)$ will be in $\mathcal{K}(\overline{q},\overline{N})$. If 
$\frac{\overline{\Gamma}^*(\xi_{out},0)}{\overline{U}^*(\xi_{out},0)} < 
\frac{\overline{\Gamma}^*(\xi_{out}^*,\xi_{relay}^*)}{\overline{U}^*(\xi_{out}^*,\xi_{relay}^*)}$, then the optimality 
of $\pi^*(\xi_{out}^*,\xi_{relay}^*)$ for the problem (\ref{eqn:constrained_problem_average_cost_with_outage_cost}) 
will be violated, since $\pi^*(A_2,0)$ will produce a strictly smaller mean power per step while  
meeting the outage constraint with equality 
(since $f_1(\xi_{out},0)=0$) and the relay placement rate constraint (since $f_2(\xi_{out},0) \leq 0$). 

We can take care of stationary points of the form $\xi_{relay} \in (0,A_3), \xi_{out}=0$ in a similar way.

If $(0,0)$ is a stationary point, then $\pi^*(0,0)$ satisfies both constraints. Also, 
$\pi^*(0,0)$ places at distance $(A+B)$ steps and uses the minimum power level for all links. Then 
$\pi^*(0,0)$ is optimal for our original problem (\ref{eqn:unconstrained_problem_average_cost_with_outage_cost}).

At $(A_2,0)$, we will have a stationary point if and only if $f_1(A_2,0) \geq 0$ and $f_2(A_2,0) \leq 0$. 
If $f_1(A_2,0) \geq 0$ and $f_2(A_2,0) = 0$, then we can make similar claims as in the 
$\xi_{out}=A_2$ and $\xi_{relay} \in (0,A_3)$ case. If $f_1(A_2,0) = 0$ and $f_2(A_2,0) < 0$, 
then we can make similar claims as in the 
$\xi_{out} \in (0,A_2)$ and $\xi_{relay} =0$ case. By the choice of $A_2$, $\pi^*(A_2,0)$ satisfies the 
outage constraint $\overline{q}$. Hence, the case $f_1(A_2,0) > 0$ will not arise.

Hence, the lemma is proved.\qed

\subsection{\textbf{Proof of Theorem~\ref{theorem:expected_average_cost_performance_of_optexplorelimadaptivelearning}}}
\label{subsection:proof_of_expected_average_cost_performance_of_optexplorelimadaptivelearning}

We will only prove that $\limsup_{x \rightarrow \infty} \frac{\mathbb{E}_{\pi_{oelal}}\sum_{i=1}^{N_x}\Gamma_i}{x} \leq \gamma^*$ almost surely.

Let us denote the shadowing random variable in the link between the potential locations 
 located at distances $i \delta$ and $j \delta$ from the sink node by $W_{i,j}$. 
 The sample space $\Omega$ associated with the deployment process is the collection of all $\omega$ (each $\omega$ 
 corresponds to a fixed realization $\{w_{i,j}: i \geq 0, j \geq 0, i>j, A+1 \leq i-j \leq A+B \}$ 
 of all possible shadowing random variables that might be encountered in the measurement process for 
 deployment up to infinity).  
 Let $\mathcal{F}$ be the Borel $\sigma$-algebra on  $\Omega$. Let $S_k=\sum_{i=1}^k U_i$ be the distance 
 (in steps) of the $k$-th relay from the sink ($S_0:=0$), and 
 $\mathcal{F}_k:=\sigma \bigg((\lambda^{(0)},\xi_{out}^{(0)},\xi_{relay}^{(0)}); W_{i,j}:i \geq 0, j \geq 0, i>j, A+1 \leq i-j \leq A+B, 
 i \leq S_{k-1}+A+B, j \leq S_{k-1}+A+B \bigg)$. The sequence of $\sigma$-algebras $\mathcal{F}_k$ is increasing in $k$, and 
 $\mathcal{F}_k$  captures the history of the deployment process 
up to the deployment of the $k$-th relay. 

Let us fix an $\epsilon>0$.

Let us recall (from Section~\ref{subsection:learning_backtracking_adaptive_with_outage_cost_algorithm}) 
the definition of the set 
$\hat{\mathcal{K}}(\overline{q},\overline{N}):=\mathcal{K}(\overline{q},\overline{N}) \cap ([0, (P_M+A_2+A_3)] \times [0,A_2] \times [0,A_3])$.

Now, by Lemma~\ref{lemma:placement_rate_mean_outage_per_step_continuous_in_lambda_xio_and_xir} 
(see Appendix~\ref{appendix:learning_backtracking_adaptive_with_outage_cost}, 
 Section~\ref{subsubsection:checking_five_conditions_kushner_optexplorelimadaptivelearning}),  
the quantities $\overline{\Gamma}(\lambda,\xi_{out},\xi_{relay})$, 
$\overline{Q}_{out}(\lambda,\xi_{out},\xi_{relay})$ and $\overline{U}(\lambda,\xi_{out},\xi_{relay})$ 
(recall the notation from Section~\ref{subsec:smdp-policy-structure}) are continuous 
in $(\lambda,\xi_{out},\xi_{relay})$.  
Hence, the ratios $\frac{\overline{\Gamma}(\lambda,\xi_{out},\xi_{relay})}{\overline{U}(\lambda,\xi_{out},\xi_{relay})}$, 
$\frac{\overline{Q}_{out}(\lambda,\xi_{out},\xi_{relay})}{ \overline{U}(\lambda,\xi_{out},\xi_{relay})}$ and 
$\frac{1}{\overline{U}(\lambda,\xi_{out},\xi_{relay})}$ are 
uniformly continuous over the compact set $[0, 2(P_M+A_2+A_3)] \times [0,A_2] \times [0,A_3]$. 
Hence, for any given $\epsilon>0$, we can find a $\delta_{\epsilon}>0$ such that 
if $(\lambda,\xi_{out},\xi_{relay})$ belongs to a $\delta_{\epsilon}$-neighbourhood of 
$\hat{\mathcal{K}}(\overline{q},\overline{N})$, then $(\lambda,\xi_{out},\xi_{relay})$ also belongs to the set 
$\hat{\mathcal{K}}_{\epsilon}(\overline{q},\overline{N})$ where:

\footnotesize
\begin{eqnarray*}
\hat{\mathcal{K}}_{\epsilon}(\overline{q},\overline{N}) &=& \bigg\{(\lambda,\xi_{out},\xi_{relay}): \\
 && \frac{\overline{\Gamma}(\lambda,\xi_{out},\xi_{relay})}{\overline{U}(\lambda,\xi_{out},\xi_{relay})} \in [\gamma^*-\epsilon, \gamma^*+\epsilon] \\
 && \frac{\overline{Q}_{out}(\lambda,\xi_{out},\xi_{relay})}{ \overline{U}(\lambda,\xi_{out},\xi_{relay})} \leq \overline{q}+ \epsilon, \\
 && \frac{1}{\overline{U}(\lambda,\xi_{out},\xi_{relay})} \leq \overline{N}+\epsilon, \\
 && 0 \leq \lambda \leq 2(P_M+A_2+A_3), \\
 && 0 \leq \xi_{out} \leq A_2, 0 \leq \xi_{relay} \leq A_3 \bigg\}
\end{eqnarray*}
\normalsize

But, by Theorem~\ref{theorem:convergence_learning_backtracking_adaptive_with_outage_cost_algorithm}, 
$(\lambda^{(k)},\xi_{out}^{(k)},\xi_{relay}^{(k)}) \rightarrow \hat{\mathcal{K}}(\overline{q},\overline{N})$ almost surely. 
Hence, there exists an integer-valued random variable $T$ such that (i) $(\lambda^{(k)},\xi_{out}^{(k)},\xi_{relay}^{(k)})$ 
belongs to a $\delta_{\epsilon}$ neighbourhood of 
$\hat{\mathcal{K}}(\overline{q},\overline{N})$ for all $k \geq T$, and (ii) $\mathbb{P}(T<\infty)=1$. In other words, 
for a sample path $\omega$ 
(for $\omega$ lying in a set of probability $1$), there exists $T(\omega)<\infty$ 
such that $(\lambda^{(k)},\xi_{out}^{(k)},\xi_{relay}^{(k)})$ 
belongs to a $\delta_{\epsilon}$ neighbourhood of 
$\hat{\mathcal{K}}(\overline{q},\overline{N})$ for all $k \geq T(\omega)$. Hence, for a sample path $\omega$ 
(for $\omega$ lying in a set of probability $1$), there exists $T(\omega)<\infty$ 
such that $(\lambda^{(k)},\xi_{out}^{(k)},\xi_{relay}^{(k)}) \in \hat{\mathcal{K}}_{\epsilon}(\overline{q},\overline{N})$ 
for all $k \geq T(\omega)$.

Using the boundedness of $\Gamma_i$ in the first equality, we obtain:

\footnotesize
\begin{eqnarray}
 &&\limsup_{x \rightarrow \infty} \frac{\mathbb{E}_{\pi_{oelal}}\sum_{i=1}^{N_x}\Gamma_i}{x} \nonumber\\
 &=&\limsup_{x \rightarrow \infty} \frac{\mathbb{E}_{\pi_{oelal}}\sum_{i=1}^{N_x+1}\Gamma_i}{x} \nonumber\\
& \leq &  \limsup_{x \rightarrow \infty} \frac{\mathbb{E}_{\pi_{oelal}} \bigg( \mathbb{I}(T <N_x+1) \sum_{i=1}^{T}\Gamma_i \bigg) }{x} \nonumber\\
&& + \limsup_{x \rightarrow \infty} \frac{\mathbb{E}_{\pi_{oelal}} \bigg( \mathbb{I}(T < N_x+1) \sum_{i=T+1}^{N_x+1}\Gamma_i \bigg) }{x} \nonumber\\
&& + \limsup_{x \rightarrow \infty} \frac{\mathbb{E}_{\pi_{oelal}} \bigg( \mathbb{I}(T \geq N_x+1) \sum_{i=1}^{N_x+1}\Gamma_i \bigg)}{x}\nonumber\\
& \leq & \mathbb{E}_{\pi_{oelal}} \limsup_{x \rightarrow \infty} \frac{ \mathbb{I}(T <N_x+1) \sum_{i=1}^{T}\Gamma_i }{x} \nonumber\\
&& +  \limsup_{x \rightarrow \infty} \frac{\mathbb{E}_{\pi_{oelal}} \bigg( \mathbb{I}(T < N_x+1) \sum_{i=T+1}^{N_x+1}\Gamma_i \bigg) }{x} \nonumber\\
&& + \mathbb{E}_{\pi_{oelal}} \limsup_{x \rightarrow \infty} \frac{\mathbb{I}(T \geq N_x+1) \sum_{i=1}^{N_x+1}\Gamma_i }{x}\nonumber\\
&=&   \limsup_{x \rightarrow \infty} \frac{\mathbb{E}_{\pi_{oelal}} \bigg( \mathbb{I}(T < N_x+1) \sum_{i=T+1}^{N_x+1}\Gamma_i \bigg) }{x} \nonumber\\
& = & \limsup_{x \rightarrow \infty} \bigg( \frac{\mathbb{E}_{\pi_{oelal}}\bigg( \mathbb{I}(T <N_x+1) \sum_{i=T+1}^{N_x+1}\Gamma_i \bigg) }{\mathbb{E}_{\pi_{oelal}}\sum_{i=T+1}^{N_x+1}U_i } \nonumber\\
&& \times \mathbb{E}_{\pi_{oelal}}\bigg(\frac{\sum_{i=T+1}^{N_x+1} U_i }{x} \bigg) \bigg) \nonumber\\
\label{eqn:equation_for_expected_average_cost_proof} 
\end{eqnarray}
\normalsize

Here the second inequality follows from Fatou's lemma. The second equality follows from the facts that 
$0 \leq \limsup_{x \rightarrow \infty} \frac{\sum_{i=1}^{T}\Gamma_i \mathbb{I}(T <N_x+1)}{x} \leq 
\limsup_{x \rightarrow \infty} \frac{\sum_{i=1}^{T}\Gamma_i }{x}
\leq \limsup_{x \rightarrow \infty} \frac{P_M T }{x}=0$ almost surely and 
$0 \leq \limsup_{x \rightarrow \infty} \frac{\sum_{i=1}^{N_x+1}\Gamma_i \mathbb{I}(T \geq N_x+1)}{x} \leq \frac{P_M}{A+1} \limsup_{x \rightarrow \infty} \mathbb{I}(T \geq N_x+1) =0$ almost 
surely (since $\mathbb{P}(T<\infty)=1$ and $\lim_{x \rightarrow \infty}N_x=\infty$ almost surely).

Now, 

\footnotesize
\begin{eqnarray*}
&& \limsup_{x \rightarrow \infty} \mathbb{E}_{\pi_{oelal}}\bigg(\frac{\sum_{i=T+1}^{N_x+1} U_i }{x} \bigg) \\
& \leq & \limsup_{x \rightarrow \infty} \mathbb{E}_{\pi_{oelal}}\frac{\sum_{i=1}^{N_x+1} U_i}{x} \\
& \leq &  \mathbb{E}_{\pi_{oelal}} \limsup_{x \rightarrow \infty} \frac{\sum_{i=1}^{N_x+1} U_i}{x} \\
   &=& 1
 \end{eqnarray*}
\normalsize

Here the second inequality follows from  Fatou's lemma, and the equality follows from the fact that $\lim_{x \rightarrow \infty} \frac{\sum_{i=1}^{N_x+1} U_i}{x}=1$ almost 
surely.

Hence, from (\ref{eqn:equation_for_expected_average_cost_proof}),

\footnotesize
\begin{eqnarray}
 &&\limsup_{x \rightarrow \infty} \frac{\mathbb{E}_{\pi_{oelal}}\sum_{i=1}^{N_x}\Gamma_i}{x} \nonumber\\
& \leq & \limsup_{x \rightarrow \infty}  \frac{\mathbb{E}_{\pi_{oelal}} \bigg( \mathbb{I}(T <N_x+1) \sum_{i=T+1}^{N_x+1}\Gamma_i \bigg)}{\mathbb{E}_{\pi_{oelal}}\sum_{i=T+1}^{N_x+1}U_i }\nonumber\\  
& = & \limsup_{x \rightarrow \infty}  \frac{\mathbb{E}_{\pi_{oelal}}\sum_{i=T+1}^{N_x+1}\Gamma_i }{\mathbb{E}_{\pi_{oelal}}\sum_{i=T+1}^{N_x+1}U_i } 
\label{eqn:equation_1_for_expected_average_cost_proof} 
 \end{eqnarray}
\normalsize

Let us denote by $\mathbb{E}_{\pi_{oelal},t}(\cdot)$ the conditional expectation under policy $\pi_{oelal}$ given that $T=t$. 
Now,

\footnotesize
\begin{eqnarray*}
&& \mathbb{E}_{\pi_{oelal}}\sum_{i=T+1}^{N_x+1}\Gamma_i  \nonumber\\
& = & \mathbb{E}_{\pi_{oelal}}\sum_{i=T+1}^{\infty}\Gamma_i  \mathbb{I}(i \leq N_x+1) \nonumber\\
& = & \sum_{t=1}^{\infty}\mathbb{P}_{\pi_{oelal}}(T=t) \nonumber\\
&& \times \mathbb{E}_{\pi_{oelal}} \bigg( \sum_{i=t+1}^{\infty}\Gamma_i  \mathbb{I}(i \leq N_x+1) \bigg| T=t \bigg)\nonumber\\
& = & \sum_{t=1}^{\infty}\mathbb{P}_{\pi_{oelal}}(T=t) \nonumber\\
&& \times \mathbb{E}_{\pi_{oelal}} \bigg( \sum_{i=t+1}^{\infty}\Gamma_i \mathbb{I}(N_x \geq i-1) \bigg| T=t \bigg)\nonumber\\
& = & \sum_{t=1}^{\infty}\mathbb{P}_{\pi_{oelal}}(T=t)\mathbb{E}_{\pi_{oelal},t} \bigg( \sum_{i=t+1}^{\infty}\Gamma_i \mathbb{I}(N_x \geq i-1)  \bigg)\nonumber\\
& = & \sum_{t=1}^{\infty}\mathbb{P}_{\pi_{oelal}}(T=t)    \sum_{i=t+1}^{\infty} \mathbb{E}_{\pi_{oelal},t} \bigg( \Gamma_i \mathbb{I}(N_x \geq i-1)  \bigg)  \nonumber\\
& = & \sum_{t=1}^{\infty}\mathbb{P}_{\pi_{oelal}}(T=t)  \nonumber\\
&& \times  \sum_{i=t+1}^{\infty} \mathbb{E}_{\pi_{oelal},t} \bigg( \mathbb{E}_{\pi_{oelal},t} \bigg( \Gamma_i \mathbb{I}(N_x \geq i-1 ) \bigg| \mathcal{F}_{i-1}  \bigg)  \bigg)  \nonumber\\
\end{eqnarray*}
\normalsize

\footnotesize
\begin{eqnarray}
& = & \sum_{t=1}^{\infty}\mathbb{P}_{\pi_{oelal}}(T=t)  \nonumber\\
&& \times  \sum_{i=t+1}^{\infty} \mathbb{E}_{\pi_{oelal},t} \bigg( \mathbb{I}(N_x \geq i-1) \mathbb{E}_{\pi_{oelal},t} \bigg( \Gamma_i  \bigg| \mathcal{F}_{i-1}  \bigg)  \bigg)  \nonumber\\
& \leq & (\gamma^*+\epsilon) \sum_{t=1}^{\infty}\mathbb{P}_{\pi_{oelal}}(T=t) \times \nonumber\\
&&  \sum_{i=t+1}^{\infty} \mathbb{E}_{\pi_{oelal},t} \bigg( \mathbb{I}(N_x \geq i-1)  \mathbb{E}_{\pi_{oelal},t} \bigg( U_i  \bigg| \mathcal{F}_{i-1}  \bigg)  \bigg)  \nonumber\\
\label{eqn:equation_2_for_expected_average_cost_proof} 
 \end{eqnarray}
\normalsize

where the fifth equality follows from the Monotone Convergence Theorem, and 
the last equality follows from the fact that the random variable  
$\mathbb{I}(N_x \geq i-1)=\mathbb{I}(\sum_{k=1}^{i-1}U_k \leq x)$ is measurable with respect to 
$\mathcal{F}_{i-1}$. The last inequality follows from that fact that 
$\frac{\mathbb{E}_{\pi_{oelal},t} \bigg( \Gamma_i  \bigg| \mathcal{F}_{i-1}  \bigg)}{\mathbb{E}_{\pi_{oelal},t} \bigg( U_i  \bigg| \mathcal{F}_{i-1}  \bigg)} \leq \gamma^*+\epsilon$ 
almost surely  
for $i>t$, given that $T=t$ (since $(\lambda^{(i-1)},\xi_{out}^{(i-1)},\xi_{relay}^{(i-1)}) \in \hat{\mathcal{K}}_{\epsilon} 
(\overline{q},\overline{N})$ for all $i-1 \geq T$).

On the other hand, we can show that:

\footnotesize
\begin{eqnarray}
&& \mathbb{E}_{\pi_{oelal}}\sum_{i=T+1}^{N_x+1} U_i  \nonumber\\
& = & \sum_{t=1}^{\infty}\mathbb{P}_{\pi_{oelal}}(T=t) \times \nonumber\\
&&   \sum_{i=t+1}^{\infty} \mathbb{E}_{\pi_{oelal},t} \bigg( \mathbb{I}(N_x \geq i-1) \mathbb{E}_{\pi_{oelal},t} \bigg( U_i  \bigg| \mathcal{F}_{i-1}  \bigg)  \bigg) \nonumber\\
\label{eqn:equation_3_for_expected_average_cost_proof} 
 \end{eqnarray}
\normalsize

From (\ref{eqn:equation_1_for_expected_average_cost_proof}), 
(\ref{eqn:equation_2_for_expected_average_cost_proof}) and 
(\ref{eqn:equation_3_for_expected_average_cost_proof}), we obtain that 
$\limsup_{x \rightarrow \infty} \frac{\mathbb{E}_{\pi_{oelal}}\sum_{i=1}^{N_x}\Gamma_i}{x} \leq \gamma^*+ \epsilon$. 
Since $\epsilon>0$ is arbitrary, we have 
$\limsup_{x \rightarrow \infty} \frac{\mathbb{E}_{\pi_{oelal}}\sum_{i=1}^{N_x}\Gamma_i}{x} \leq \gamma^*$. But 
$\gamma^*$ is the optimal mean power per step for problem 
(\ref{eqn:constrained_problem_average_cost_with_outage_cost}). Hence, 
$\limsup_{x \rightarrow \infty} \frac{\mathbb{E}_{\pi_{oelal}}\sum_{i=1}^{N_x}\Gamma_i}{x} = \gamma^*$.

In a similar manner, we can show that 
$\limsup_{x \rightarrow \infty} \frac{\mathbb{E}_{\pi_{oelal}}\sum_{i=1}^{N_x}Q_{out}^{(i,i-1)}}{x} \leq \overline{q}$ and 
$\limsup_{x \rightarrow \infty} \frac{\mathbb{E}_{\pi_{oelal}} N_x}{x} \leq \overline{N}$.\qed

\end{document}